\documentclass[article]{amsart}
\date{\today}

\usepackage{amsthm}
\usepackage{amsfonts}
\usepackage{amssymb}
\usepackage{amsmath}
\usepackage{color}
\usepackage[sans]{dsfont}
\usepackage{soul}
\usepackage[normalem]{ ulem }

\definecolor{Greenn}{rgb}{0.0,0.50,0.01}

\numberwithin{equation}{section}

\newcommand{\D}{\mathfrak {D}}

\def\L{L} 

\DeclareMathOperator*{\wl}{w-\lim}

\def\1{1\hskip-0.09cm{\rm l}} 

\def\hD{H_D}

\def\hZ0{H_{Z^0}}

\newcommand{\R}{{\mathbb R}} 
\newcommand{\C}{{\mathbb C}} 
\newcommand{\N}{{\mathbb N}} 
\newcommand{\gH}{{\mathfrak {H}}} 
\newcommand{\gF}{{\mathfrak {F}}}
\newcommand{\cH}{{\mathcal{H}}} 

\def\p{\mathbf{p}}
\def\q{\mathbf{q}}
\def\la{\langle}
\def\ra{\rangle}

\def\z0{Z^0}

\def\mz{m_{Z^0}} 
\def\mw{m_{W}} 

\def\cg{C_{\mathrm{gap}}} 

\newtheorem{prop}{Proposition}[section]

\newtheorem{theo}[prop]{Theorem}
\newtheorem{lem}[prop]{Lemma}
\newtheorem{defi}[prop]{Definition}
\newtheorem{rema}[prop]{Remark}
\newtheorem{hypo}{Hypothesis}

\title[Local decay for weak interactions]{Local Decay for Weak Interactions with Massless Particles}

\author{Jean-Marie Barbaroux}

\address{Jean-Marie Barbaroux\\
 Aix Marseille Univ, Universit\'e de Toulon, CNRS, CPT, Marseille, France.}
\email{barbarou@univ-tln.fr}
\author{J\'er\'emy Faupin}
\address{
J\'er\'emy Faupin\\
Institut Elie Cartan de Lorraine, Universit{\'e} de Lorraine\\
57045 Metz Cedex 1, France.
}
\email{jeremy.faupin@univ-lorraine.fr}
\author{Jean-Claude Guillot}
\address{
Jean-Claude Guillot\\
CNRS-UMR 7641, Centre de Math\'ematiques Appliqu\'ees, Ecole Polytechnique\\
91128 Palaiseau Cedex, France.
}
\email{guillot@cmapx.polytechnique.fr}

\begin{document}
\begin{abstract}
We consider a mathematical model for the weak decay of the intermediate boson $Z^0$ into neutrinos and antineutrinos. We prove that the total Hamiltonian has a unique ground state in Fock space and we establish a limiting absorption principle, local decay and a property of relaxation to the ground state for initial states and observables suitably localized in energy and position. Our proofs rest, in particular, on Mourre's theory and a low-energy decomposition.
\end{abstract}

\maketitle

\section{Introduction}\label{S1}

We consider in this paper a mathematical model for the weak decay of the intermediate boson $Z^0$ into neutrinos and antineutrinos. This is a part of a program devoted to the study of mathematical models for the weak decays, as patterned according to the Standard model in Quantum Field Theory; See \cite{AGG,ABFG,BFG4,BFG3,bg3,bg4,Guillot2015}.

In \cite{ABFG}, W. Aschbacher and the authors studied the spectral theory of the Hamiltonian associated to the weak decay of the intermediate bosons $W^{\pm}$ into the full family of leptons. In this paper, we consider the weak decay of the boson $Z^0$, and, for simplicity, we restrict our study to the model representing the decay of $Z^0$ into the neutrinos and antineutrinos associated to the electrons. Hence, neglecting the small masses of neutrinos and antineutrinos, we define a total Hamiltonian $H$ acting in an appropriate Fock space and involving two fermionic, massless particles -- the neutrinos and antineutrinos -- and one massive bosonic particle -- the boson $Z^0$. In order to obtain a well-defined operator, we approximate the physical kernels of the interaction Hamiltonian by square integrable functions and we introduce high-energy cutoffs. In particular the Hamiltonian that we consider is not translation invariant. We emphasize, however, that we do not need to impose any low-energy regularization in the present work. We use in fact the spectral representation of the massless Dirac operator by the sequence of spherical waves (see \cite{GreinerMuller,ref8} and Appendix \ref{appendixA}). The precise definition of $H$ as a self-adjoint operator is given in Section \ref{S2}.

By adapting to our context methods of previous papers \cite{ABFG,BFP1,bg4,Pizzo}, we prove that $H$ has a unique ground state for sufficiently small values of the coupling constant. This ground state is expected to be an equilibrium state in the sense that any initial state relaxes to the ground state by emitting particles that propagate to infinity. Rigorously proving such a statement requires to develop a full scattering theory for our model and is beyond the scope of this paper. Nevertheless, we are able to establish a property of relaxation to the ground state for any initial state of energy smaller than the mass of $Z^0$ and ``\emph{localized in the position variable}'' for the neutrinos and antineutrinos, and for any observable localized in a similar sense.

To prove our main result, we use Mourre's theory in the spirit of \cite{FGS1} and \cite{ABFG}. This gives a limiting absorption principle and local energy decay in any spectral interval above the ground state energy and below the mass of $Z^0$. In particular, the limiting absorption principle shows that the spectrum between the ground state energy and the first threshold is purely absolutely continuous. For this part of the proof, the main difference with \cite{ABFG} is that we have to deal with two different species of massless particles. This leads to some technical issues.

Now, the local decay property obtained by Mourre's theory depends on the energy of the initial state under consideration. More precisely, the rate of decay tends to $0$ as the energy of the initial state approaches the ground state energy. However, in our model, neutrinos and antineutrinos are massless particles and it should be expected that their speed of propagation is energy-independent. Justifying this fact is the main novelty and one of the main achievements of this paper. As a consequence, as mentioned above, we  establish relaxation to the ground state in a suitable sense for localized observables and states. The uniformity in energy of the local decay property is obtained by adapting the proof in \cite{BH1,BH2} and \cite{BF} to the present context, with the delicate issue that we have to deal, in our case, with a more singular Hamiltonian in the infrared region than in \cite{BF}. We refer to Section \ref{subsec:mainresults} for a more detailed explanation of our strategy and a comparison with the literature.

Our paper is organized as follows. We begin with introducing the physical model that we consider and stating our main results in Section \ref{S2}. The proofs are given in Sections \ref{section:S3} and \ref{section:results}. In Appendix \ref{appendixA}, we give estimates of the free massless Dirac  spherical waves, and we recall a technical result in Appendix \ref{appendixB}.


\section{The model and the results}\label{S2}

In \cite{BFG3}, we considered the weak decay of $Z^0$ into electrons and positrons. The model that we study here is the same as the one in \cite{BFG3}, except that the massive fermions, the electrons and positrons, are replaced by neutrinos and antineutrinos that we treat as massless fermions. Therefore, in Subsection \ref{subsec:model}, we present briefly the model studied in the present paper and we refer the reader to \cite[Section 2]{BFG3} for more details. In Subsection \ref{subsec:mainresults}, we state our main results and compare them to the literature.

\subsection{The model}\label{subsec:model}

\subsubsection{The free Hamiltonian}

We use a system of units such that $\hbar = c =1$. The total Fock space for neutrinos, antineutrinos and $Z^0$ bosons is defined as
\begin{equation*}
\mathcal{H} := \mathfrak{F}_D \otimes \mathfrak{F}_{\z0} ,
\end{equation*}
where
\begin{equation*}
\mathfrak{F}_D := {\mathfrak F}_{a} \otimes  {\mathfrak F}_{a} :=  {\mathfrak F}_{a}( \mathfrak{H}_{c} ) \otimes  {\mathfrak F}_{a}( \mathfrak{H}_{c} ) := \bigoplus_{n=0}^\infty  \otimes_a^n \mathfrak{H}_{c} \otimes \bigoplus_{n=0}^\infty  \otimes_a^n \mathfrak{H}_{c} ,
\end{equation*}
is the tensor product of antisymmetric Fock spaces for neutrinos and antineutrinos, and
\begin{equation*}
 \mathfrak{F}_{\z0} := \mathfrak{F}_s( L^2(\Sigma_3) ) := \bigoplus_{n=0}^\infty \otimes_s^n L^2( \Sigma_3 ),
\end{equation*}
is the bosonic Fock space for the boson $\z0$. Here
\begin{equation}\label{def:H_c}
\mathfrak{H}_{c} := L^2( \Sigma ; \mathbb{C}^4 ) := L^2( \mathbb{R}_+ \times \Gamma ;\C^4),
\end{equation}
where
\begin{equation*}
 \Gamma : = \Big \{ ( j ,  m_j ,  \kappa_j ) ,  j\in\N+\frac12,  m_j\in\{  -j, -j+1, \cdots, j-1, j \} ,  \kappa_j\in\{ \pm ( j+ \frac{1}{2}) \}  \Big \} ,
\end{equation*}
 represents the one-particle Hilbert space for both neutrinos and antineutrinos, labeled in terms of modulus of the momentum and angular momentum quantum numbers. We will denote by $\xi_1 := ( p_1 , \gamma_1 ) \in \Sigma = \mathbb{R}_+ \times \Gamma$ the quantum variable in the case of neutrinos, and by $\xi_2 := ( p_2 , \gamma_2 ) \in \Sigma$ the quantum variable in the case of antineutrinos. Likewise, $L^2( \Sigma_3 )$ represents the one-particle Hilbert space for the $Z^0$ bosons, with
 \begin{equation*}
 \Sigma_3 := \R^3\times\{ -1, 0, 1\} ,
 \end{equation*}
and we denote by $\xi_3 := (k, \lambda)\in\Sigma_3$ the quantum variable for $Z^0$. The vacuum in $\mathfrak{F}_D$ (respectively in $\mathfrak{F}_{Z^0}$) is denoted by $\Omega_D$ (respectively by $\Omega_{Z^0}$).

The total free Hamiltonian $H_0$, acting on $\mathcal{H}$, is defined by
\begin{equation*}
  H_0 :=H_D \otimes \1_{ \mathfrak{F}_{Z^0} } + \1_{ \mathfrak{F}_D } \otimes\hZ0  ,
\end{equation*}
where $H_D$ is the Hamiltonian of the quantized Dirac field, acting on $\mathfrak{F}_D$ and given by
\begin{align*}
H_D  & := H_{0,+} + H_{0,-} \phantom{\int} \\
& := \mathrm{d} \Gamma( \omega( p_1 )  ) \otimes \1_{ \mathfrak{F}_a } + \1_{ \mathfrak{F}_a } \otimes \mathrm{d} \Gamma( \omega( p_2 ) )   \\
&:= \int \omega(p_1) b_{+}^*(\xi_1) b_{+}(\xi_1) d \xi_1  +   \int  \omega(p_2) b_{-}^*(\xi_2) b_{-}(\xi_2) d \xi_2 ,
\end{align*}
and $H_{Z^0}$ is the Hamiltonian of the bosonic field, acting on $\mathfrak{F}_{\z0}$, and given by
\begin{equation*}
H_{Z^0} :=  \mathrm{d} \Gamma( \omega_3( k ) )  := \int \omega_3(k) a^*(\xi_3) a(\xi_3)  d \xi_3 .
\end{equation*}
The massless dispersion relation for the neutrinos and antineutrinos is $\omega(p) := p$, the dispersion relation for the massive boson $Z^0$ is $\omega_3( k ) := \sqrt{|k|^2 + \mz^2}$, with $\mz$ the mass of $Z^0$. The operator-valued distributions $b^\sharp_+( \xi_1 )$ (respectively $b^\sharp_-( \xi_2 )$), with $b^\sharp = b^*$ or $b$, are the fermionic annihilation and creation operators for the neutrinos (respectively antineutrinos) and $a^\sharp(\xi_3)$ are the bosonic creation and annihilation operators for the $Z^0$ bosons satisfying the usual canonical commutation relations.  In addition, following the convention described in
\cite[section~4.1]{W} and
\cite[section~4.2]{W}, we will assume that fermionic creation and annihilation operators of neutrino anticommute with fermionic creation and annihilation operators of antineutrino. Therefore, for $\epsilon, \epsilon' = \pm$, the following canonical anticommutation and commutation relations hold (see \cite{BFG3, ABFG}  for details),
\begin{equation}\nonumber
\begin{split}
 &\{ b_{\epsilon}(\xi), b^*_{\epsilon'}(\xi')\} =
 \delta_{\epsilon \epsilon'} \delta(\xi - \xi')  , \quad
 \{ b_{\epsilon}(\xi), b_{\epsilon'}(\xi')\}
 =0\ ,\\
 &[ a(\xi_3), a^*(\xi_3')] = \delta(\xi_3 - \xi_3') \ ,\quad 
 [ a(\xi_3), a(\xi_3') ] = 0 ,\\ 
 &[ b_{\epsilon}(\xi), a(\xi_3)]
 = [ b_{\epsilon}(\xi), a^*(\xi_3)]
 = 0 ,
\end{split}
\end{equation}
where $\{b, b'\} = bb' + b'b$ and $[a,a'] = aa' - a'a$.

The spectrum of $H_0$ consists of the simple eigenvalue $0$, associated with the normalized eigenstate $\Omega_D \otimes \Omega_{Z^0}$, and the semi-axis $[0,\infty)$ of absolutely continuous spectrum.

We conclude this paragraph by introducing the number operators that will be used several times in our analysis. The number operators for neutrinos and antineutrinos, denoted respectively by $N_+$ and $N_-$, are given by
 \begin{equation*}
  N_+ := \int b_{+}^*(\xi_1) b_{+}(\xi_1) d \xi_1, \quad  N_- := \int b_{-}^*(\xi_2) b_{-}(\xi_2) d \xi_2 .
 \end{equation*}
The number operator for bosons is given by
\begin{equation*}
 N_{Z^0} := \int a^*(\xi_3) a(\xi_3) d \xi_3 .
\end{equation*}

\subsubsection{The Interaction Hamiltonian}

According to the Standard Model (see \cite[\S~21]{W2}) the intermediate boson $Z^0$ interacts with every  neutrino and antineutrino associated with the massive leptons, i.e. the electron, the muon and the tau. In order to simplify we only consider in this paper the interaction of $Z^0$ with the neutrinos and antineutrinos of the electron. In the Schr\"odinger representation, the latter is given by (see e.g. \cite[(4.139)]{GreinerMuller} and \cite[\S~21]{W2})
\begin{equation}\label{eq:ekd}
 - \frac{g}{4 \cos\theta} \int  \overline{\Psi}_{\nu}(x) \gamma^\mu (1-\gamma_5)\Psi_{\nu}(x) Z_\mu(x) d x  + \mathrm{h.c.} ,
\end{equation}
where $\cos\theta = \frac{\mw}{\mz}$, with $\mw$ the mass of the bosons $W^{\pm}$ and $\mz$ the mass of $Z^0$, and
\begin{equation*}
\frac{g^{2}}{8 \mw^{2}}=\frac{G_{F}}{\sqrt{2}}.
\end{equation*}
Here $G_{F}$ is the conventional Fermi constant, $ G_{F}\simeq (1.16)10^{-5}$(Gev)$^{-2}$. Furthermore, $\mz\simeq 91.18$ Gev and $ \mw\simeq 80.41$ Gev. In our paper, the coupling constant $g$ will be treated as a non-negative small parameter. In \eqref{eq:ekd}, $\gamma^\mu$, $\mu=0,1,2,3$, and $\gamma_5$ stands for the usual Dirac matrices, $\Psi_{\nu}(x)$ and $\overline{\Psi}_{\nu}(x)$ are the Dirac fields for the neutrinos and antineutrinos given by
\begin{equation*}
 \Psi_{\nu}(x) :=  \int \psi_+(\xi_1, x) b_+(\xi_1) d \xi_1  + \int\widetilde{\psi}_-(\xi_2, x) b^*_-(\xi_2) d \xi_2 ,
\end{equation*}
and $Z_\mu(x)$ is the massive boson field for $Z^0$ defined by (see e.g. \cite[Eq.~(5.3.34)]{W}),
\begin{equation*}
  Z_\mu(x) := {(2\pi)}^{-\frac32} \int \frac{ d \xi_3}{(2 (|k|^2  + \mz^2)^\frac12)^\frac12}  \Big(\varepsilon_\mu(k,\lambda) a(\xi_3) \mathrm{e}^{i k.x} + \varepsilon_\mu^*(k,\lambda) a^*(\xi_3) \mathrm{e}^{- i k.x}\Big) .
\end{equation*}
Here $\varepsilon_\mu(k,\lambda)$ are polarizations vectors for the (spin $1$) boson $Z^0$. Moreover, for $\xi = ( p , \gamma ) = ( p , ( j , m_j , \kappa_j ) )$,
\begin{equation*}
 \widetilde{\psi}_-(\xi, x) :=  \psi_{-}((p,(j,-m_j , - \kappa_j ) ) , x) ,
\end{equation*}
and $\psi_\pm( \xi , x )$ are the continuum eigenstates of the free Dirac operator
\begin{equation*}
D_0 := - i \boldsymbol{\alpha} \cdot \nabla ,
\end{equation*}
acting on $L^2( \mathbb{R}^3 ; \mathbb{C}^4 )$, where $\boldsymbol{\alpha} = (\alpha_1,\alpha_2,\alpha_3)$ is the triple of Dirac matrices in the standard form (see \cite{BFG3}). The generalized eigenstates satisfy, for $\xi = ( p , \gamma )$,
$D_0 \psi_\pm( \xi, x) = \pm  \omega(p) \psi_\pm( \xi , x)$, and are normalized in such a way that
\begin{align*}
& \int_{\mathbb{R}^3} \psi_\pm^\dagger( \xi , x) \psi_\pm( \xi' , x) d  x = \delta ( \xi - \xi' ) = \delta_{\gamma\gamma'}\delta(p-p'),\\
& \int_{\mathbb{R}^3} \psi_\pm^\dagger( \xi, x ) \psi_{\mp}( \xi' , x) d x = 0 ,
\end{align*}
where $\xi' = ( p' , \gamma' )$ and $\psi_\pm^\dagger( \xi , x)$ denotes the adjoint spinor of $\psi_\pm( \xi , x )$ (see Appendix~\ref{appendixA} for a more detailed description).

Expanding formally the interaction \eqref{eq:ekd} into a product of creation and annihilation operators, we obtain a finite sum of Wick monomials with integral kernels too singular to define closed operators. Physically, however, the weak interaction has a very short range and the lifetimes of the intermediate bosons are very short. In other words, for the decays of the intermediate bosons, the weak interaction acts locally in space-time. Thus, in order to obtain a well-defined Hamiltonian, we proceed as in e.g \cite{GlimmJaffe, BDG, BDG2, bg4, ABFG}, replacing the singular kernels by square integrable functions $F^{(\alpha)}$ (see \eqref{eq:inter00}--\eqref{eq:inter112c} below and Hypothesis~\ref{hypothesis-1}). In particular, we introduce cutoffs for high momenta of neutrinos, antineutrinos and $Z^0$ bosons, and we confine in space the interaction between the neutrinos/antineutrinos and the bosons by adding a function $f(|x|)$, with $f \in \mathrm{C}_0^\infty( [ 0 , \infty ) )$. Similarly as in \cite{BFG3}, the interaction Hamiltonian is thus associated to the operator
\begin{equation}\label{eq:inter00}
 H_I := H_I^{(1)} + H_I^{(2)} + \mathrm{h.c.} ,
\end{equation}
acting on $\mathcal{H}$, where
\begin{align}
H_I^{(1)} :=  \int & \Big(\int_{\mathbb{R}^3} f(|x|) \overline{\Psi}_+(\xi_1, x) \gamma^\mu(1-\gamma_5) \widetilde{\Psi}_-(\xi_2, x) \frac{\varepsilon_\mu(\xi_3)}{\sqrt{2 \omega_3(k)}} \mathrm{e}^{i k\cdot x} d x\Big) \notag \\
& \times G^{(1)}(\xi_1, \xi_2, \xi_3) b_+^*(\xi_1) b_-^*(\xi_2) a(\xi_3) d \xi_1 d \xi_2 d \xi_3 , \label{eq:inter11}
\end{align}
and
\begin{align}
H_I^{(2)} := \int & \Big( \int_{\mathbb{R}^3} f(|x|) \overline{\Psi}_+(\xi_1, x) \gamma^\mu(1-\gamma_5) \widetilde{\Psi}_-(\xi_2, x) \frac{\varepsilon^*_\mu(\xi_3)}{\sqrt{2 \omega_3(k)}} \mathrm{e}^{-i k\cdot x} d x\Big) \notag \\
& \times G^{(2)}(\xi_1, \xi_2, \xi_3) b_+^*(\xi_1) b_-^*(\xi_2) a^*(\xi_3) d \xi_1 d \xi_2 d \xi_3 . \label{eq:inter33}
\end{align}

Denoting by $h^{(1)}( \xi_1 , \xi_2 , \xi_3 )$ and $h^{(2)}( \xi_1 , \xi_2 , \xi_3 )$ the integrals w.r.t. $x$ in the expressions above, we see that $H_I^{(1)}$ and $H_I^{(2)}$ can be rewritten as
\begin{align}
&  H_I^{(1)}  := H_I^{(1)}(F^{(1)})  := \int F^{(1)}(\xi_1,\xi_2, \xi_3) b_+^*(\xi_1) b_-^*(\xi_2) a(\xi_3) d\xi_1d\xi_2 d\xi_3 , \label{eq:inter11b} \\
&  H_I^{(2)}  := H_I^{(2)}(F^{(2)}) := \int F^{(2)}(\xi_1, \xi_2, \xi_3) b_+^*(\xi_1) b_-^*(\xi_2) a^*(\xi_3) d \xi_1 d \xi_2 d\xi_3 , \label{eq:inter2244}
\end{align}
where, for $j=1,2$,
\begin{align}
 F^{(j)}(\xi_1, \xi_2, \xi_3) := h^{(j)}( \xi_1 , \xi_2 , \xi_3 ) G^{(j)}( \xi_1 , \xi_2 , \xi_3 ) . \label{eq:inter112c}
\end{align}

We are now ready to define the total Hamiltonian $H$ associated to our model.


\subsubsection{The total Hamiltonian}\label{S2.4}

\begin{defi}\label{def:hamiltonian}
The Hamiltonian of the decay of the boson $\z0$ into a neutrino and an antineutrino, acting on $\mathcal{H}$, is
\begin{equation*}
  H := H_0 + g H_I ,
\end{equation*}
where $g$ is a non-negative coupling constant.
\end{defi}
The assumption that $g \ge 0$ is made for simplicity of exposition, but, of course, since $H_g$ and $H_{-g}$ are unitarily equivalent, all our results below hold for $g \in \mathbb{R}$ ($|g|$ small enough).

\subsection{Main results}\label{subsec:mainresults}

We make the following hypothesis on the interaction $H_I$ defined by \eqref{eq:inter00}--\eqref{eq:inter112c}.
\begin{hypo}\label{hypothesis-1}
$\ $

\noindent(i) $f(\cdot)\in C^\infty([0,\infty))$ and there exists $\Lambda>0$ such that $f(x)=0$ for all $|x|>\Lambda$.

\noindent(ii) For $j=1,2$, $G^{(j)}$ is uniformly bounded in $\xi_1$, $\xi_2$ and $\xi_3$.

\noindent iii) There exists a compact set $K\subset\R^+\times\R^+\times\R^3$ such that, for $j=1,2$ and for all $(\gamma_1, \gamma_2, k)\in\Gamma\times\Gamma\times\{-1,0,1\}$, we have $G^{(j)}(p_1,\gamma_1; p_2,\gamma_2; k,\lambda) = 0$ if $(p_1,p_2,k)\not\in K$.

\end{hypo}

In order to apply Mourre's theory, we will have to strengthen Hypothesis \ref{hypothesis-1}(ii) as follows.
\begin{hypo}\label{hypothesis-2}
For $j=1,2$, $G^{(j)}$ is twice differentiable in the variables $p_1$ and $p_2$, with partial derivatives up to the order 2 in $p_1$ and $p_2$ uniformly bounded in $\xi_1$, $\xi_2$ and $\xi_3$.
\end{hypo}

Our main results are listed in the following theorems.
\begin{theo}[Self-adjointness]\label{thm:sa}
Suppose that Hypothesis~\ref{hypothesis-1} holds. There exists $g_0$ such that, for all $0 \le g \leq g_0$, the Hamiltonian $H$ given in Definition~\ref{def:hamiltonian} is self-adjoint with domain $\mathfrak{D}( H ) = \mathfrak{D}(H_0)$.
\end{theo}
Theorem \ref{thm:sa} follows from the Kato-Rellich theorem together with relative bounds of $H_I$ w.r.t. $H_0$ that will be established in Section \ref{section:S3}.
\begin{theo}[Existence of a ground state]\label{thm:GS}
Suppose that Hypothesis~\ref{hypothesis-1} holds. There exists $g_0>0$ such that, for all $0 \le g \leq g_0$, the Hamiltonian $H$ has a unique ground state associated to the ground state energy $E:=\inf\sigma(H)$.
\end{theo}
Theorem \ref{thm:GS} is proven with the help of Pizzo's iterative perturbation theory \cite{Pizzo,BFP1}. Compared to previous papers employing this method to prove the existence of a ground state in quantum field theory models, the main new issue that we encounter comes from the fact that we are considering two different species of massless particles. To overcome this difficulty, we use in particular suitable versions of the $N_\tau$ estimates of \cite{GlimmJaffe}.
\begin{theo}[Location of the essential spectrum]\label{thm:es}
Suppose that Hypothesis~\ref{hypothesis-1} holds. There exists $g_0 > 0$ such that, for all $0 \le g \le g_0$,
\begin{equation*}
 \sigma_{\mathrm{ess}}(H) = [E, \infty) .
\end{equation*}
\end{theo}
To prove Theorem \ref{thm:es}, we construct a Weyl sequence associated to $\lambda$, for any $\lambda \in [E,\infty)$, by adapting to our context arguments taken from \cite{DG}, \cite{Arai} and \cite{Takaesu2014}.

Our next spectral result concerns the absolute continuity of the spectrum of $H$ above the ground state energy $E$ and below the mass of $Z^0$. It is in fact related to the subsequent theorem on local decay, as we will explain later on.
\begin{theo}[Absolute continuity of the spectrum]\label{thm:absolute}
Suppose that Hypotheses~\ref{hypothesis-1} and \ref{hypothesis-2} hold. For all $\varepsilon > 0$, there exists $g_0 > 0$ such that, for all $0 \le g  \le g_0$, the spectrum of $H$ in $( E , E + m_{Z^0} - \varepsilon )$ is purely absolutely continuous.
\end{theo}
Let $\q_1 := i \nabla_{\p_1}$, respectively $\q_2 := i \nabla_{\p_2}$, denote the ``position'' operator for neutrinos, respectively for antineutrinos, and let $q_1 := | \q_1 |$, $q_2 := | \q_2 |$. To shorten the notations, we also set
\begin{equation*}
Q := \mathrm{d} \Gamma( q_1 ) + \mathrm{d} \Gamma( q_2 ).
\end{equation*}
Our main result is summarized in the following theorem.
\begin{theo}[Local decay and relaxation to the ground state] \label{thm:return}
Suppose that Hypotheses~\ref{hypothesis-1} and \ref{hypothesis-2} hold. For all $\varepsilon > 0$, there exists $g_0 > 0$ such that, for all $ 0 \le g  \le g_0$, $\chi \in C_0^\infty ( ( - \infty , m_{ Z^0 } - \varepsilon ) ; \R )$, $t \in \mathbb{R}$, $0 < s \le 1$ and $0 < \mu < s$,
\begin{align}
& \langle Q \rangle^{-s} e^{ - i t H }Ê\chi( H ) \langle Q \rangle^{-s} = e^{ - i t E }Ê\chi( E ) \langle Q \rangle^{-s} P_{\mathrm{gs}} \langle Q \rangle^{-s} + \mathrm{R}_0(t), \label{eq:uniflocdec}
\end{align}
where $P_{\mathrm{gs}}$ is the orthogonal projection onto the ground state of $H$ and $\mathrm{R}_0(t)$ is a bounded operator such that
\begin{equation*}
\| \mathrm{R}_0(t) \|Ê\le C_{s,\mu} \langle t \rangle^{-s + \mu} ,
\end{equation*}
with $C_{s,\mu}$ a positive constant only depending on $s$ and $\mu$. In particular, if $O$ is an observable ``localized'' in the position variable for the neutrinos and antineutrinos, in the sense that
\begin{equation*}
\big \| \langle Q \rangle^s O \langle Q \rangle^s \big \|Ê< \infty ,
\end{equation*}
and if $\phi$ is an initial state ``localized'' in position and energy, in the sense that $\phi = \chi( H ) \langle Q \rangle^{-s} \psi$ for some $\psi \in \mathcal{H}$ and $\chi \in C_0^\infty ( ( - \infty , m_{ Z^0 } - \varepsilon ) ; \R )$, then we have that
\begin{equation}
\langleÊ\phi , e^{ i t H }ÊO e^{ - i t H }Ê\phi \rangle = | \langle \varphi_{ \mathrm{gs}Ê} , \phi \rangle |^2 \langle \varphi_{ \mathrm{gs}Ê} , O \varphi_{ \mathrm{gs} } \rangle + \mathrm{R}_1(t) , \label{eq:returnunif}
\end{equation}
where $\varphi_{ \mathrm{gs}Ê}$ is a ground state of $H$, and with
\begin{equation*}
| \mathrm{R}_1(t) |Ê\le C_s \langle t \rangle^{-s + \mu}.
\end{equation*}
\end{theo}
Our proof of Theorem \ref{thm:return} (and Theorem \ref{thm:absolute}) is inspired by arguments developed previously in \cite{FGS1}, \cite{ABFG} and \cite{BF}. The approach rests on Mourre's theory \cite{Mo}. As is now well-known, given a self-adjoint operator $H$ and another self-adjoint (``conjugate'') operator $A$, a Mourre estimate $\1_I( H ) [ H , i A ]Ê\1_I( H ) \ge c_0 \1_I( H )$, $c_0 > 0$, combined with a suitable notion of regularity of $H$ w.r.t. $A$ yields a limiting absorption principle for $H$ in $I$,
\begin{equation*}
\sup_{z \in \C, \mathrm{Re}(z) \in I, 0 < | \mathrm{Im}(z) | \le 1} \big \|Ê\langle A \rangle^{-s} ( H - z )^{-1}Ê\langle A \rangle^{-s} \big \|Ê< \infty ,
\end{equation*}
(and hence absolute continuity of the spectrum of $H$ in $I$) and the local decay property
\begin{equation}
\big \|Ê\langle A \rangle^{-s} e^{-itH} \chi (H)Ê\langle A \rangle^{-s} \big \|Ê\le C \langle t \rangle^{-s} , \label{eq:locdecay}
\end{equation}
for all $\chi \in C_0^\infty( I ; \mathbb{R} )$ and certain $s>0$.

In our setting, the main difficulty we encounter to prove a Mourre estimate for the total Hamiltonian $H$ in a spectral interval $I$ close to the ground state energy $E$ is due to the presence of massless particles. If one chooses $A$ as the second quantized generator of dilatations, $E$ becomes a ``threshold'' for $H$. The analysis of the spectral and dynamical properties of a self-adjoint operator near thresholds is generally a subtle problem.

In \cite{FGS1}, for the standard model of non-relativistic QED, Fr{\"o}hlich, Griesemer and Sigal proved a Mourre estimate in any spectral interval of the form $E + [ \alpha \sigma , \beta \sigma ]$, for some fixed $0 < \alpha < \beta$, and any $0 < \sigma < \sigma_0$. Important ingredients entering the proof in \cite{FGS1} are Pizzo's iterative perturbation theory \cite{Pizzo,BFP1} and a unitary Pauli-Fierz transformation that regularizes the standard model of non-relativistic QED in the infrared region.

In \cite{ABFG}, a modification of the method of \cite{FGS1} was proposed in order to cope with a model more singular in the infrared -- the model studied in \cite{ABFG} describes the decay of the intermediate vector bosons $W^\pm$ into the full family of leptons. We mention that yet another modification of the method of \cite{FGS1} was introduced in \cite{CFFS}, in order to prove the local decay property in the translation invariant standard model of non-relativistic QED at a fixed total momentum. In the present paper, since the interaction Hamiltonian of the model we consider is not regular enough in the infrared region to follow the proof of \cite{FGS1}, we proceed as in \cite{ABFG} to obtain a Mourre estimate. A substantial difference with \cite{ABFG}, however, is that we have to deal with two species of massless particles -- neutrinos and antineutrinos -- instead of one, which leads again to a few technical issues. The precise statement of the Mourre estimate that we prove is given in Theorem \ref{thm:mourre}.

In all the previously cited works, as well as in our paper, the positive constant $c_0$ of the Mourre estimate is proportional to the low-energy parameter $\sigma$, i.e. proportional to the distance from the ground state energy $E$ to the spectral interval $I$ under consideration. This in turn implies that the constant $C$ of the local decay property \eqref{eq:locdecay} depends on $\sigma$ (more precisely it can be verified that $C$ in \eqref{eq:locdecay} is of order $\mathcal{O}( \sigma^{-s} )$, see Theorem \ref{thm:LAP} for details). In our context, \eqref{eq:locdecay} can be interpreted as a statement about the propagation of neutrinos and antineutrinos for initial states in the range of $\chi( H ) \langle A \rangle^{-s}$. That the constant $C$ depends on $\sigma$ seems to suggest that, at least for such initial states, the speed of propagation of neutrinos and antineutrinos depends on $\sigma$, i.e. depends on the energy. Yet, in weak decays, neutrinos and antineutrinos can be considered as massless particles and in this paper, we indeed disregard their masses. Therefore it should be expected that, for any initial state, the speed of propagation of neutrinos and antineutrinos is energy-independent. Justifying this fact is one of our main achievements. Besides, apart from its physical relevance, an important consequence of having a uniform local energy decay is the property of relaxation to the ground state for localized observables and states, as stated in Theorem \ref{thm:return}.

In \cite{BF}, Bony and the second author adapted to the framework of Quantum Field Theory a method introduced in \cite{BH1,BH2} in order to justify that photons propagate at the speed of light in the standard model of non-relativistic QED, for any (localized) initial states with low-energy. In this paper, we follow the general strategy of \cite{BF}. The point is to establish that one can arrive at the desired uniform local energy decay by replacing the weights $\langle A \rangle^{-s}$ expressed in terms of powers of the conjugate operator, by weights $\langle Q \rangle^{-s}$ expressed in the (second quantized) position operators. To prove this, we use the localization in energy $\chi(H)$ and a second quantized version of Hardy's inequality.  Again, the fact that we are considering two different species of massless particles here leads to some technical difficulties compared to \cite{BF}. To overcome them, we use in particular the crucial property that neutrinos and antineutrinos are fermions -- not bosons. But the main difference with \cite{BF} comes from the already mentioned fact that the model we consider here is more singular in the infrared region than the standard model of non-relativistic QED (for which the Pauli-Fierz transformation can be applied). Also the structure of the interaction Hamiltonian in our setting is very different from the one in \cite{BF}. It must be handled differently, using for instance the $N_\tau$ estimates of \cite{GlimmJaffe} in a proper way. We therefore modify the proof of \cite{BF} in several places. The main novelties will be underlined in Section \ref{section:results}. A drawback of having a more singular interaction term is that the rate of decay of the remainder term $\mathrm{R}_0$ in the statement of Theorem \ref{thm:return} is slightly smaller than the one in \cite{BF}. This might be an artefact of the method, but we believe that this rate of decay could indeed be intimately related to the infrared behavior of the interaction Hamiltonian of the model.

To conclude this section, we mention that another choice of a conjugate operator, the second quantized generator of radial translations, has been used previously in the literature to deal with comparable problems; see \cite{GGM2,FMS}. This conjugate operator is not self-adjoint, and its commutator with the total Hamiltonian is not controllable by the Hamiltonian itself. Nevertheless, an abstract extension of the usual Mourre theory covering this framework, sometimes called singular Mourre theory, has been developed in \cite{GGM1}. An advantage of using the generator of radial translations as a conjugate operator is that it gives a Mourre estimate with a positive constant $c_0$ which is uniform in the distance from $E$ to the spectral interval $I$. But unfortunately, this choice of a conjugate operator is not possible in our concrete setting, unless one imposes an artificial infrared regularization in the interaction Hamiltonian. Otherwise the Hamiltonian is not regular enough for the Mourre theory to be applied.

The next two sections are devoted to the proofs of Theorem \ref{thm:sa}--\ref{thm:return}. In what follows, $C$ will stand for a positive constant independent of the parameters that may differ from one line to another.


\section{Self-adjointness, existence of a ground state and location of the essential spectrum}\label{section:S3}

\subsection{Proof of  Theorem~\ref{thm:sa}} The self-adjointness of $H$ is a straightforward consequence of the Kato-Rellich theorem together with the relative bound for the interaction given by Proposition~\ref{prop:relative-bound} below. We begin with a technical lemma.
\begin{lem}\label{lem:kernel-bound-1}
Suppose that Hypothesis~\ref{hypothesis-1} holds. For $p\in\R^+$, $j\in\N+\frac12$ and $\ell_j = j+\frac12$, let
\begin{equation}
 A(p, \gamma) := \frac{(2p)^{\ell_j}}{\Gamma(\ell_j)} \left( \int _0^\infty |f(r)| r^{2\ell_j} (1+p^2r^2)(1+r^2+r^4) d r \right)^\frac12 ,
\end{equation}
and
\begin{equation}
 \tilde{A}(p, \gamma) := \ell_j (\ell_j -1)\frac{(2p)^{\ell_j-2}}{\Gamma(\ell_j)} \left( \int _0^\infty |f(r)| r^{2(\ell_j-1)}  d r \right)^\frac12 .
\end{equation}
Then, for $h^{(l)}$ given by \eqref{eq:inter112c} and \eqref{eq:inter11}--\eqref{eq:inter33}, there exists a constant $C_{\mz}$ such that, for all $(\xi_1,\xi_2,\xi_3) = \left(\left(p_1, (j_1, m_{j_1}, \kappa_{j_1})\right); \left(p_2,(j_2,m_{j_2}, \kappa_{j_2})\right); \left(k,\lambda\right)\right)\in\Sigma\times\Sigma\times\Sigma_3$, we have that
\begin{equation*}
 | h^{(l)}(\xi_1,\xi_2,\xi_3)| \leq  C_{\mz} (|k|^2 + \mz^2)^\frac14 A(p_1, \ell_{j_1}) A(p_2, \ell_{j_2}) ,
\end{equation*}
for $l=1,2$. Moreover,
\begin{equation*}
  \left| \frac{\partial}{\partial p_i} h^{(l)}(\xi_1,\xi_2,\xi_3)\right| \leq  C_{\mz} (|k|^2 + \mz^2)^\frac14  p_i^{-1} (\ell_{j_i}+1) A(p_1, \ell_{j_1}) A(p_2, \ell_{j_2}) ,
\end{equation*}
for $l=1,2$ and $i=1,2$, and
\begin{equation*}
\begin{split}
&  \left| \frac{\partial^2}{\partial p_1\partial p_2} h^{(l)}(\xi_1,\xi_2,\xi_3)\right| \leq  C_{\mz} (|k|^2 + \mz^2)^\frac14 \prod_{i=1,2} \left(p_i^{-1} (\ell_{j_i}+1) A(p_i, \ell_{j_i})\right) , \\
& \left| \frac{\partial^2}{\partial p_1^2} h^{(l)}(\xi_1,\xi_2,\xi_3)\right| \leq  24 C_{\mz} (|k|^2 + \mz^2)^\frac14 \left(\tilde{A}(p_1, \ell_{j_1}) + (1+\ell_{j_1} + \ell_{j_1}^2) A(p_1, \ell_{j_1})\right) \\
 & \qquad \qquad \qquad \qquad\qquad \times p_2^{-1} A(p_2, \ell_{j_2}) , \\
& \left| \frac{\partial^2}{\partial p_2^2} h^{(l)}(\xi_1,\xi_2,\xi_3)\right|  \leq 24 C_{\mz} (|k|^2 + \mz^2)^\frac14 \left(\tilde{A}(p_2, \ell_{j_2}) + (1+\ell_{j_2} + \ell_{j_2}^2) A(p_2, \ell_{j_1})\right) \\
 &\qquad\qquad\qquad\qquad\qquad \times p_1^{-1} A(p_1, \ell_{j_2}) . \\
\end{split}
\end{equation*}
\end{lem}
\begin{proof}
The estimates stated in the lemma are direct consequences of the definition of the functions $h^{(l)}$, the bounds \eqref{eq:appA-1}--\eqref{eq:appA-3} obtained in Appendix~\ref{appendixA}, and the fact that $\varepsilon_\mu(\xi_3) / \sqrt{2 \omega_3(k)}$ are bounded functions of $k$ and $\xi_3$.
\end{proof}
\begin{prop}\label{prop:relative-bound}
Suppose that Hypothesis~\ref{hypothesis-1} holds. There exist $C_{F^{(j)}}$ and $\tilde{C}_{F^{(j)}}$ ($j=1,2$), defined as
\begin{equation}\label{eq:defCF}
 C_{F^{(1)}} := 2 \left\| \frac{F^{(1)}}{\sqrt{\omega_3}}\right\|^2
 + \left\| \frac{F^{(1)}}{\sqrt{p_2 \omega_3}}\right\|^2 , \quad
 C_{F^{(2)}} := \left\| \frac{F^{(2)}}{\sqrt{p_2 \omega_3}}\right\|^2 ,
\end{equation}
and
\begin{equation}\label{eq:deftCF}
\begin{split}
\tilde{C}_{F^{(1)}} & := \max\left( \left\| \frac{F^{(1)}}{\sqrt{p_2}}\right\|^2
, \left\| \frac{F^{(1)}}{\sqrt{p_2\omega_3}}\right\|^2 \right), \\
\tilde{C}_{F^{(2)}} & := \max\left( 2 \left\| \frac{F^{(2)}}{\sqrt{\omega_3}}\right\|^2
 + \left\| \frac{F^{(2)}}{\sqrt{p_2 \omega_3}}\right\|^2
 , 2 \left\| F^{(2)}\right\|^2
 + \left\| \frac{F^{(2)}}{\sqrt{p_2}}\right\|^2\right),
\end{split}
\end{equation}
such that, for all $\epsilon>0$ and all $\psi\in\D(H_0)$, we have that
\begin{equation}\nonumber
\begin{split}
 \|H_I \psi\|^2 \leq
 4 \sum_{j=1,2}\Big( C_{F^{(j)}} \|(H_0+1)\psi\|^2
 + \tilde{C}_{F^{(j)}} \Big( (1+\epsilon) \|(H_0+1)\psi \|^2 + \frac{1}{4\epsilon} \|\psi\|^2 \Big)\Big).
\end{split}
\end{equation}
\end{prop}
\begin{rema}
i) Lemma \ref{lem:kernel-bound-1} implies that the functions $F^{(j)}$ and $F^{(j)}/p_i$ are in $L^2$. Therefore, all constants defined in \eqref{eq:defCF} and \eqref{eq:deftCF} are finite.

ii) In the sequel, we use Proposition \ref{prop:relative-bound} for the kernels $F^{(j)}$ given by \eqref{eq:inter112c}, as well as for a localized version of these kernels, namely, for $\1_J F^{(j)}$ where $\1_J$ is the characteristic function of a set $J$. The result remains true in this case, replacing $F^{(j)}$ by $\1_J F^{(j)}$ in the definitions \eqref{eq:defCF} and \eqref{eq:deftCF}.

iii) One can indifferently replace $p_2$ by $p_1$ in the definitions of the constants $C_{F^{(j)}}$ and $\tilde{C}_{F^{(j)}}$.
\end{rema}
We establish Proposition \ref{prop:relative-bound} by combining the $N_\tau$ estimates of \cite{GlimmJaffe} and a strategy used in \cite{BDG, BDG2}. The arguments are fairly standard but we give details for the convenience of the reader.
\begin{proof}[Proof of Proposition~\ref{prop:relative-bound}]
For a.e.  $\xi_3\in\Sigma_3$, we define
\begin{equation*}
\begin{split}
 B^{(1)}(\xi_3)  & := - \int \overline{F^{(1)}(\xi_1, \xi_2, \xi_3)} b_+(\xi_1) b_-(\xi_2) d \xi_1 d \xi_2 , \\
 B^{(2)}(\xi_3)  & :=  \int F^{(2)}(\xi_1, \xi_2, \xi_3) b_+^*(\xi_1) b_-^*(\xi_2) d \xi_1 d \xi_2 .
\end{split}
\end{equation*}
Thus, according to \eqref{eq:inter11b}, we have
\begin{equation}
\begin{split}
 H_I^{(1)}  = \int B^{(1)}(\xi_3)^* \otimes a(\xi_3) d \xi_3  \quad \mathrm{and}\quad
 H_I^{(2)}  = \int B^{(2)}(\xi_3) \otimes a^*(\xi_3) d \xi_3  .
\end{split}
\end{equation}
In the rest of the proof, we establish a relative bound with respect to $H_0$ for each of the terms $\|H_I^{(1)}u\|$, $\|H_I^{(2)}u\|$, $\|(H_I^{(1)})^*u\|$ and $\|(H_I^{(2)})^*u\|$.

\noindent \textbf{Relative bound for $\| H_I^{(1)} u \|$}:
As in the proof of \cite[Lemma~5.1]{BDG2}, picking $u\in\D(H_0)$ and defining
\begin{equation*}
  \Phi(\xi_3) := \omega_3(k)^\frac12 \left( (\hD +1)^\frac12 \otimes a(\xi_3)\right) u,
\end{equation*}
we have
\begin{equation*}
 \int \| \Phi(\xi_3) \|^2 d\xi_3 = \big \| ( (\hD+1)^\frac12\otimes {\hZ0}^\frac12 )u \big \|^2,
\end{equation*}
and
\begin{align}
 & \left\| \int (B^{(1)}(\xi_3))^*\otimes a(\xi_3) d \xi_3 u \right\|^2 \notag \\
 &  \leq \left( \int \frac{1}{\omega_3(k)} \big \| (B^{(1)}(\xi_3)^* (\hD+1)^{-\frac12} \big \|^2_{\mathfrak{F}_D} d\xi_3 \right)
  \big  \| ( (\hD + 1)^\frac12 \otimes \hZ0^\frac12 )u \big \|^2 . \label{eq:prop3.4:1}
\end{align}
On the other hand, from Lemma~\ref{lem:bg4}, we obtain that
\begin{equation}\label{eq:prop3.4:2}
  \left\| (B^{(1)}(\xi_3))^*  (\hD+1)^{-\frac12} \psi \right\|^2
  \leq
  \left( 2 \| F^{(1)}(\cdot, \cdot, \xi_3)\|^2 + \|\frac{F^{(1)}(\cdot, \cdot, \xi_3)}{\sqrt{p_2}} \|^2\right)\|\psi\|^2.
\end{equation}
Thus, for  $C_{F^{(1)}} := 2 \left\| \frac{F^{(1)}}{\sqrt{\omega_3}}\right\|^2
 + \left\| \frac{F^{(1)}}{\sqrt{p_2 \omega_3} }\right\|^2$, the inequalities \eqref{eq:prop3.4:1} and \eqref{eq:prop3.4:2} give
\begin{align}
\big \| H_I^{(1)} u \big \| = \left\| \int (B^{(1)}(\xi_3))^* \otimes a(\xi_3) d \xi_3 u  \right\|^2
 & \leq C_{F^{(1)}} \left\| ( (\hD+1)^\frac12 \otimes \hZ0^\frac12 ) u \right\|^2 \notag \\
 & \leq  C_{F^{(1)}} \| (H_0+1) u \|^2 . \label{eq:prop3.4:3}
\end{align}

\noindent \textbf{Relative bound for $\| (H_I^{(2)})^* u\|$}:
Using again \cite[Lemma~5.1]{BDG2}, since
\begin{equation*}
 (B^{(2)} (\xi_3) )^* = - \int \overline{F^{(2)}(\xi_1, \xi_2, \xi_3)} b_+(\xi_1) b_-(\xi_2) d \xi_1 d \xi_2  ,
\end{equation*}
we obtain that
\begin{align}
 & \left\| \int (B^{(2)}(\xi_3))^* \otimes a(\xi_3) d\xi_3 u \right\|^2 \notag \\
 & \leq
 \left( \int \frac{1}{\omega_3(k)} \big \| (B^{(2)}(\xi_3))^* (\hD+1)^{-\frac12} \big \|^2_{\mathfrak{F}_D} d\xi_3 \right)
 \| ( (\hD+1)^\frac12 \otimes \hZ0^\frac12)u \|^2 . \label{eq:prop3.4:4}
\end{align}
Similarly to the above, using Lemma~\ref{lem:bg4} yields
\begin{equation*}
 \int \frac{1}{\omega_3(k)} \big \| (B^{(2)}(\xi_3))^* (\hD+1)^{-\frac12} \big \|^2_{\mathfrak{F}_D} d\xi_3
 \leq
 \int \frac{|F^{(2)}(\xi_2,\xi_2,\xi_3)|^2}{p_2 \omega_3(k)} d \xi_1 d \xi_2 d \xi_3 .
\end{equation*}
Therefore, for $C_{F^{(2)}} : = \left\| \frac{F^{(2)}}{\sqrt{ p_2 \omega_3}}\right\|^2 $,
we get with \eqref{eq:prop3.4:4},
\begin{equation}\label{eq:prop3.4:5}
  \big \| (H_I^{(2)})^* u \big \| = \left\| \int (B^{(2)}(\xi_3))^* \otimes a(\xi_3) d\xi_3 u \right\|^2  \leq C_{F^{(2)}} \| (H_0 +1) u \|^2 .
\end{equation}

\noindent \textbf{Relative bound for $\| (H_I^{(1)})^* u \|$}: We have
\begin{align*}
 \int B^{(1)}(\xi_3) \otimes a^*(\xi_3) d \xi_3 u = \int \Big( - \int & \overline{F^{(1)}(\xi_1,\xi_2,\xi3)} b_+(\xi_1) b_-(\xi_2) d \xi_1 d \xi_2 \Big) \\
 &  \otimes a^*(\xi_3) d\xi_3 u .
\end{align*}
Hence, using the canonical commutation relations, we obtain
\begin{equation}\label{eq:prop3.4:6}
\begin{split}
&  \left\| \int B^{(1)}(\xi_3) \otimes a^*(\xi_3) d\xi_3 u \right\|^2 \\
& = \int \la B^{(1)}(\xi_3)\otimes a(\xi_3') a^*(\xi_3)u, B^{(1)}(\xi_3')\otimes\1 u\ra d\xi_3d\xi_3' \\
& = \int \la B^{(1)}(\xi_3)\otimes \1 \delta(\xi_3' - \xi_3) u, B^{(1)}(\xi_3')\otimes\1 u\ra d\xi_3d\xi_3' \\
& \ \ + \int \la B^{(1)}(\xi_3)\otimes a^*(\xi_3) a(\xi_3') u, B^{(1)}(\xi_3')\otimes\1 u\ra d\xi_3d\xi_3' .
\end{split}
\end{equation}
The first term in the right hand side of \eqref{eq:prop3.4:6} is
\begin{equation}\label{eq:prop3.4:7}
\begin{split}
 &  \int \la B^{(1)}(\xi_3)\otimes \1 \delta(\xi_3' - \xi_3) u, B^{(1)}(\xi_3')\otimes\1 u\ra d\xi_3d\xi_3' \\
 & = \int \big \| B^{(1)}(\xi_3) \otimes \1 u \big \|^2d\xi_3 \\
 & = \int \big \| (B^{(1)}(\xi_3) (\hD+1)^{-\frac12}\otimes\1) ((\hD+1)^\frac12\otimes\1) u \big \|^2 d\xi_3 \\
 & \leq   \left(\int \big \| (B^{(1)}(\xi_3) (\hD+1)^{-\frac12}) \big \|^2 d \xi_3 \right)
\big \| ((\hD+1)^\frac12\otimes\1) u \big \|^2 \\
 & \leq C'_{F^{(1)}} \left( \epsilon \|(\hD+1)\otimes\1 u\|^2 + \frac{1}{4\epsilon} \|u\|^2 \right)  ,
 \end{split}
\end{equation}
with $C'_{F^{(1)}}:= \left\| \frac{F^{(1)}}{\sqrt{p_2}}\right\|^2$, and where we used Lemma~\ref{lem:bg4} in the last inequality.

The second term in the right hand side of \eqref{eq:prop3.4:6} can be rewritten as
\begin{equation}\label{eq:prop3.4:8}
\begin{split}
  & = \int \la B^{(1)}(\xi_3)\otimes a(\xi_3') u, B^{(1)}(\xi_3')\otimes  a(\xi_3) u\ra d\xi_3d\xi_3' \\
  & = \int \frac{1}{\sqrt{\omega_3(k)}} \frac{1}{\sqrt{\omega_3(k')}} \Big\la B^{(1)}(\xi_3) (\hD+1)^{-\frac12}\otimes\1 \Phi(\xi_3'), \\
  & \qquad \qquad \qquad \qquad \qquad \quad  B^{(1)}(\xi_3')  (\hD+1)^{-\frac12}\otimes\1 \Phi(\xi_3) \Big \ra d\xi_3d\xi_3' \\
  & \leq \left( \int \frac{1}{\sqrt{\omega_3(k)}} \big \| B^{(1)}(\xi_3) (\hD+1)^{-\frac12} \big \| \big \| \Phi(\xi_3) \big \| d\xi_3\right)^2 \\
  & \leq  \left( \int \frac{1}{\omega_3(k)} \| B^{(1)}(\xi_3) (\hD+1)^{-\frac12} \|^2 d\xi_3 \right)
  \int \|\Phi(\xi_3)\|^2 d\xi_3 \\
  & \leq  \left( \int \frac{1}{\omega_3(k)} \big \| B^{(1)}(\xi_3) (\hD+1)^{-\frac12} \big \|^2 d\xi_3 \right)
  \| (H_0 + 1) u \|^2 \\
  & \leq C''_{F^{(1)}} \| (H_0 + 1) u\|^2 , \\
 \end{split}
\end{equation}
where
\begin{equation*}
    C''_{F^{(1)}} = \left\| \frac{F^{(1)}}{\sqrt{p_2\omega_3}}\right\|^2 .
\end{equation*}
Thus, collecting \eqref{eq:prop3.4:7} and \eqref{eq:prop3.4:8} yields, for $\tilde{C}_{F^{(1)}} = \max
( C'_{F^{(1)}} , C''_{F^{(1)}}) $,
\begin{align}
\big \| (H_I^{(1)})^* u \big \|^2 &= \left\| \int B^{(1)}(\xi_3) \otimes a^*(\xi_3) d \xi_3 u \right\|^2  \notag \\
& \leq \tilde{C}_{F^{(1)}}  \left( (1+\epsilon) \| (H_0+1) u \|^2 + \frac{1}{4\epsilon} \|u\|^2 \right) . \label{eq:prop3.4:9}
\end{align}

\noindent \textbf{Relative bound for $H_I^{(2)}$}: With similar argument as above, we can prove
\begin{align}
\big \| H_I^{(2)} u \big \|^2 &= \left\| \int B^{(2)}(\xi_3) \otimes a^*(\xi_3) d\xi_3 \right\|^2 \notag \\
& \leq \tilde{C}_{F^{(2)}} \left( (1+\epsilon) \| (H_0+1) u\|^2 + \frac{1}{4\epsilon} \| u \|^2 \right), \label{eq:prop3.4:10}
\end{align}
where $\tilde{C}_{F^{(2)}} = \max\Big( 2 \left\| \frac{F^{(2)}}{\sqrt{\omega_3}}\right\|^2
 + \left\| \frac{F^{(2)}}{\sqrt{p_2 \omega_3}}\right\|^2
 , 2 \left\| F^{(2)}\right\|^2
 + \left\| \frac{F^{(2)}}{\sqrt{p_2}}\right\|^2\Big)$.

Collecting \eqref{eq:prop3.4:3}, \eqref{eq:prop3.4:5}, \eqref{eq:prop3.4:9} and \eqref{eq:prop3.4:10} concludes the proof of Proposition~\ref{prop:relative-bound}.
\end{proof}
%
%

\subsection{Proof of Theorem~\ref{thm:GS}}

To prove the existence of a ground state, we follow the strategy of \cite{BFS98} and \cite{BFP1,Pizzo} (see also \cite{AGG, BDG2}). The first step consists in proving the existence of a spectral gap for the Hamiltonian with neutrinos sharp infrared cutoffs.

Let $(\sigma_n)$ be a nonnegative decreasing sequence that tends to zero as $n$ tends to infinity.
Given $m<n$ two positive integers, let $\mathfrak{H}_{c,n} : = L^2([\sigma_n, \infty)\times\Gamma; \C^4)$ be the one neutrino space with momentum larger than $\sigma_n$, and
$\mathfrak{F}_{n} := \mathfrak{F}_{a}(\mathfrak{H}_{c,n})$ be its associated Fock space. Similarly, we define $\mathfrak{F}_{n}^m := \mathfrak{F}_{a}(L^2([\sigma_n, \sigma_m)\times\Gamma; \C^4) )$ and $\mathfrak{F}_\infty^n:=\mathfrak{F}_a(L^2([0, \sigma_n)\times\Gamma; \C^4)$. The Fock space $\mathfrak{F}_a$ for the neutrino is thus unitarily equivalent to $\mathfrak{F}_n \otimes
\mathfrak{F}_{n+1}^n \otimes \mathfrak{F}_{\infty}^{n+1}$. We then consider the total Fock space $\mathcal{H}_n$ with neutrino and antineutrino momenta larger than $\sigma_n$,
\begin{equation*}
 \mathcal{H}_n := \mathfrak{F}_n \otimes \mathfrak{F}_n \otimes \mathfrak{F}_{Z^0}.
\end{equation*}
The Hilbert space $\mathcal{H}_n$ identifies with a subspace of $\mathcal{H}$.

The neutrino infrared cutoff Hamiltonian is defined on $\mathcal{H}$ by
\begin{equation*}
 H_n : = H_{0} + g H_{I,n},
\end{equation*}
for all $n \in \mathbb{N}$, where the interaction part is
\begin{equation*}
 H_{I,n} := H_{I,n}^{(1)} + H_{I,n}^{(2)} + \mathrm{h.c.} ,
\end{equation*}
with
\begin{align*}
& H_{I,n}^{(1)}  := \int F^{(1)}(\xi_1, \xi_2,  \xi_3) \1_{\{ \sigma_n \leq p_1\}}(p_1)
 \1_{\{ \sigma_n \leq p_2\}}(p_2) b_+^*(\xi_1) b_-^*(\xi_2) a(\xi_3) d\xi_1d\xi_2 d\xi_3 , \\
&  H_{I,n}^{(2)}   : = \int F^{(2)}(\xi_1, \xi_2, \xi_3)
  \1_{\{ \sigma_n \leq p_1\}}(p_1)
 \1_{\{ \sigma_n \leq p_2\}}(p_2) b_+^*(\xi_1) b_-^*(\xi_2) a^*(\xi_3) d\xi_1d\xi_2 d\xi_3 .
\end{align*}
The restriction of $H_n$ to $\mathcal{H}_n$ is denoted by $K_n$. In other words,
\begin{equation*}
K_n : = H_{0,n} + g H_{I,n},
\end{equation*}
as an operator on $\mathcal{H}_n$, where the free part is here
\begin{equation}\nonumber
\begin{split}
 H_{0,n}  := & \int \omega(p_1) \1_{\{ \sigma_n \leq p_1\}}(p_1)
 b^*_+(\xi_1) b_+(\xi_1) d \xi_1 \\
 & + \int \omega(p_2) \1_{\{ \sigma_n \leq p_2\}}(p_2)
 b^*_-(\xi_2) b_-(\xi_2) d \xi_2 +\hZ0 .
\end{split}
\end{equation}
Observe that
\begin{equation*}
H_n = K_n \otimes \1_{\infty}^n \otimes \1_\infty^n + \1_n \otimes \check{H}^n_\infty ,
\end{equation*}
where $\1^n_\infty$ denotes the identity on $\mathfrak{F}_{\infty}^{n}$, $\1_n$ denotes the identity on $\mathcal{H}_{n}$ and
\begin{equation}\nonumber
\begin{split}
 \check{H}_\infty^n :=& \mathrm{d} \Gamma( \omega( p_1 ) ) \otimes \1^n_\infty + \1^n_\infty \otimes \mathrm{d} \Gamma ( \omega( p_2 ) ) \\
 =  & \int \1_{\{  p_1 \le \sigma_n\}}(p_1) \omega(p_1) b_+^*(\xi_1) b_+(\xi_1) d \xi_1
 \\
 & + \int \1_{\{  p_2 \le \sigma_n\}}(p_2) \omega(p_2) b_-^*(\xi_2) b_-(\xi_2) d \xi_2  ,\quad\mbox{on}\quad \mathfrak{F}_{\infty}^{n} \otimes \mathfrak{F}_{\infty}^{n} .
\end{split}
\end{equation}

To prove the spectral gap result, we also need the following  Hamiltonians,
\begin{equation*}
\tilde{K}_n := K_n\otimes \1_{n+1}^n \otimes\1_{n+1}^n + \1_n\otimes \check{H}_{n+1}^n \quad\mbox{on}\quad \mathcal{H}_n\otimes\left(\mathfrak{F}_{n+1}^n\otimes\mathfrak{F}_{n+1}^n\right) \simeq \mathcal{H}_{n+1},
\end{equation*}
where $\1_{n+1}^n\otimes\1_{n+1}^n$ denotes the identity on  $\mathfrak{F}_{n+1}^n\otimes \mathfrak{F}_{n+1}^n$, and with
\begin{equation}\nonumber
\begin{split}
 \check{H}_{n+1}^n : = & \int \1_{\{\sigma_{n+1}\leq p_1 < \sigma_n\}}(p_1) \omega(p_1) b_+^*(\xi_1) b_+(\xi_1) d \xi_1
 \\
 & + \int \1_{\{\sigma_{n+1}\leq p_2 < \sigma_n\}}(p_2) \omega(p_2) b_-^*(\xi_2) b_-(\xi_2) d \xi_2  ,\quad\mbox{on}\quad \mathfrak{F}_{n+1}^n\otimes \mathfrak{F}_{n+1}^n .
\end{split}
\end{equation}
We define
\begin{equation}\nonumber
\begin{split}
 H_{I,n+1}^{\ \ n} :=  H_{I,n+1}^{(1) \, n} +  H_{I,n+1}^{(2) \, n} + \mathrm{h.c.} ,
\end{split}
\end{equation}
where
\begin{equation}\nonumber
\begin{split}
 H_{I,n+1}^{(1)\, n} = &
  \int F^{(1)}(\xi_1, \xi_2, \xi_3)
 \big(
 \1_{\{\sigma_{n+1} \leq p_1 < \sigma_n, \sigma_n\leq p_2\}}(p_1, p_2) \\
 &
 + \1_{\{\sigma_{n+1} \leq p_1 < \sigma_n, \sigma_{n+1}\leq p_2<\sigma_n\}}(p_1, p_2)
 + \1_{\{\sigma_{n} < p_1 , \sigma_{n+1}\leq p_2 < \sigma_{n}\}}(p_1, p_2)
 \big) \\
 & b_+^*(\xi_1) b_-^*(\xi_2) a(\xi_3) d \xi_1 d \xi_2 d \xi_3 , \phantom{\int}
\end{split}
\end{equation}
and
\begin{equation}\nonumber
\begin{split}
 H_{I,n+1}^{(2)\, n} = &
  \int F^{(2)}(\xi_1, \xi_2, \xi_3)
 \big(
 \1_{\{\sigma_{n+1} \leq p_1 < \sigma_n, \sigma_n\leq p_2\}}(p_1, p_2) \\
 &
 + \1_{\{\sigma_{n+1} \leq p_1 < \sigma_n, \sigma_{n+1}\leq p_2<\sigma_n\}}(p_1, p_2)
 + \1_{\{\sigma_{n} < p_1 , \sigma_{n+1}\leq p_2 < \sigma_{n}\}}(p_1, p_2)
 \big) \\
 & b_+^*(\xi_1) b_-^*(\xi_2) a^*(\xi_3) d \xi_1 d \xi_2 d \xi_3 . \phantom{\int}
\end{split}
\end{equation}
With these definitions, we get
\begin{equation*}
K_{n+1} = \tilde{K}_n + g H_{I,n+1}^{\ \ n} ,
\end{equation*}
where the equality holds up to a unitary transformation.

Our first lemma gives a relative bound of $H_{I,n}$ w.r.t. $H_{0,n}$.
\begin{lem}\label{cor2:relative-bound}
Suppose that Hypothesis~\ref{hypothesis-1} holds. For a given $n \in \mathbb{N}$ and for $j\in\{1,2\}$, consider $C_{n,j}:=C_{\1_{\{p_1\geq \sigma_n; p_2\geq\sigma_n\}}F^{(j)}}$ and $\tilde{C}_{n,j}:=\tilde{C}_{\1_{\{p_1\geq \sigma_n; p_2\geq\sigma_n\}}F^{(j)}}$ obtained by replacing $F^{(j)}$ by $\1_{\{p_1\geq \sigma_n; p_2\geq\sigma_n\}}F^{(j)}$ in \eqref{eq:defCF} and \eqref{eq:deftCF} respectively. Then for all $\epsilon>0$ and all $\psi\in\D(H_{0,n})$, we have that
\begin{equation}\nonumber
\begin{split}
 & \|H_{I,n} \psi\|^2 \\
 & \leq
 4 \sum_{j=1,2}\Big( C_{n,j} \|(H_{0,n}+1)\psi\|^2
 + \tilde{C}_{n,j} \Big( (1+\epsilon) \|(H_{0,n}+1)\psi \|^2 + \frac{1}{4\epsilon} \|\psi\|^2 \Big)\Big) .
\end{split}
\end{equation}
\end{lem}
\begin{proof}
The proof is exactly the same as the one of Proposition~\ref{prop:relative-bound}, replacing $\mathcal{H}$ by $\mathcal{H}_n$, $H_0$ by $H_{0,n}$ and $H$ by $K_n$. 
\end{proof}
We will also need the following relative bound.
\begin{lem}\label{lm:acompleter}
Suppose that Hypothesis~\ref{hypothesis-1} holds. There exist $a$, $\tilde{a}$, $b$ and $\tilde{b}$ depending on the functions $F^{(j)}$ ($j=1,2$) such that, for all $n\geq 0$ and all $\psi\in\D(H_0)$, we have that
\begin{equation}\label{eq:bnd-a-b}
\begin{split}
 & \|H_{I,n}^{\ \ n+1} \psi\| \leq (\sigma_n - \sigma_{n+1}) \left( a  \|H_{0,n+1} \psi\| + b \|\psi\| \right)
\end{split}
\end{equation}
and
\begin{equation}\label{eq:form-bnd-a-b}
\begin{split}
 & \left| \la H_{I,n}^{\ \ n+1} \psi, \psi\ra\right|
  \leq  (\sigma_n - \sigma_{n+1}) \left( \tilde{a} \la H_{0,n+1} \psi, \psi\ra
  + \tilde{b}\|\psi\| \right).
\end{split}
\end{equation}
\end{lem}
\begin{proof}
A straightforward consequence of Proposition~\ref{prop:relative-bound} and its proof is that for
\begin{equation*}
\alpha_{n,j}:=C_{\1_{\{\sigma_n>p_1\geq \sigma_{n+1}; \sigma_n>p_2\geq\sigma_{n+1}\}}F^{(j)}} ,
\end{equation*}
and
\begin{equation*}
{\beta}_{n,j}:=\tilde{C}_{\1_{\{\sigma_n>p_1\geq \sigma_{n+1}; \sigma_n>p_2\geq\sigma_{n+1}\}}F^{(j)}} ,
\end{equation*}
as given by \eqref{eq:defCF} and \eqref{eq:deftCF} respectively, we have for all $\epsilon>0$ and all $\psi\in\D(H_{0,n})$,
\begin{align}
 \|H_{I,n}^{\ \ n+1} \psi\|^2  \leq 4 \sum_{j=1,2}\Big( & \alpha_{n,j} \|(H_{0,n+1}+1)\psi\|^2 \notag \\
& + {\beta}_{n,j} \Big( (1+\epsilon) \|(H_{0,n+1}+1)\psi \|^2 + \frac{1}{4\epsilon} \|\psi\|^2 \Big)\Big) . \label{eq:bnd-alpha-beta}
\end{align}
Now from Lemma~\ref{lem:kernel-bound-1}, we get that for all $i$ and $j$, the functions $F^{(j)}$ defined in \eqref{eq:inter112c},
and $F^{(j)}/\sqrt{p_i}$ are locally bounded and in $L^2$.
Since in addition the function $1/\omega_3$ is bounded, we derive from \eqref{eq:defCF} and \eqref{eq:deftCF}, applied to $\1_{\{\sigma_n>p_1\geq \sigma_{n+1}; \sigma_n>p_2\geq\sigma_{n+1}\}}F^{(j)}$ instead of $F^{(j)}$,  that there exists $c_j$, $\alpha$ and $\beta$ such that for all $n$,
\begin{equation}\label{eq:alpha_beta}
 \alpha_{n,j} \leq c_j \alpha (\sigma_n - \sigma_{n+1})^2
 \quad\mbox{and}\quad
 \beta_{n,j} \leq c_j \beta (\sigma_n - \sigma_{n+1})^2 .
\end{equation}
Equations \eqref{eq:bnd-alpha-beta} and \eqref{eq:alpha_beta} imply \eqref{eq:bnd-a-b}. Using \eqref{eq:bnd-a-b} and \cite[Theorem~X.18]{RS} yields \eqref{eq:form-bnd-a-b} and thus conclude the proof.
\end{proof}


Another related result, that will also be used in the proof of Theorem~\ref{thm:GS}, is as follows. We set
\begin{equation*}
g H_{I,\infty}^{\ \ n} := H - H_n ,
\end{equation*}
and we recall that
\begin{equation*}
 \check{H}_\infty^n = \mathrm{d} \Gamma ( \1_{[0,\sigma_n]}(p_1) p_1 ) + \mathrm{d} \Gamma ( \1_{[0,\sigma_n]}(p_2) p_2 ).
\end{equation*}
As before, the interaction term may be decomposed as
\begin{equation}\nonumber
\begin{split}
 H_{I,\infty}^{\ \ n} :=  H_{I,\infty}^{(1)\, n} +  H_{I,\infty}^{(2)\, n} + \mathrm{h.c.} ,
\end{split}
\end{equation}
where
\begin{equation}\nonumber
\begin{split}
 H_{I,\infty}^{(1)\, n} = &
  \int F^{(1)}(\xi_1, \xi_2, \xi_3)
 \big(
 \1_{\{ p_1 \le \sigma_n , \sigma_n\leq p_2\}}(p_1, p_2)  + \1_{\{ p_1 \le \sigma_n,  p_2 \le \sigma_n \}}(p_1, p_2) \\
 & \quad + \1_{\{ \sigma_{n} \le p_1 , p_2 \le \sigma_{n}\}}(p_1, p_2)
 \big)  b_+^*(\xi_1) b_-^*(\xi_2) a(\xi_3) d \xi_1 d \xi_2 d \xi_3 ,
\end{split}
\end{equation}
and
\begin{equation}\nonumber
\begin{split}
 H_{I , \infty }^{(2)\, n} = &
  \int F^{(2)}(\xi_1, \xi_2, \xi_3)
 \big(
 \1_{\{ p_1 \le \sigma_n, \sigma_n \leq p_2\}}(p_1, p_2)
 + \1_{\{ p_1 \le \sigma_n,   p_2 \le \sigma_n\}}(p_1, p_2) \\
& \quad + \1_{\{ \sigma_{n} \le p_1 , p_2 \le \sigma_{n}\}}(p_1, p_2)
 \big)  b_+^*(\xi_1) b_-^*(\xi_2) a^*(\xi_3) d \xi_1 d \xi_2 d \xi_3 .
\end{split}
\end{equation}
We have the following statement.
\begin{lem}\label{lm:rel_Ntau}
Suppose that Hypothesis \ref{hypothesis-1} holds. For $j = 1 , 2$, there exists $C_j > 0$ such that, for all $n \in \mathbb{N}$,
\begin{align}
\big \| (  \check{H}_\infty^n )^{- \frac12}  H_{I,\infty}^{(j) \, n} ( N_{ Z^0 } + 1 )^{-\frac12}  \big \| \le C_j \sigma_n . \label{eq:rel_Ntau}
\end{align}
\end{lem}
\begin{proof}
The estimate \eqref{eq:rel_Ntau} is a consequence of the $N_\tau$ estimates of \cite{GlimmJaffe}. First we prove \eqref{eq:rel_Ntau} for $j=1$. Considering for instance the term
\begin{equation}\nonumber
\begin{split}
 \mathrm{R}_1 := &
  \int F^{(1)}(\xi_1, \xi_2, \xi_3)
 \1_{\{ p_1 \le \sigma_n , \sigma_n\leq p_2\}}(p_1, p_2) b_+^*(\xi_1) b_-^*(\xi_2) a(\xi_3) d \xi_1 d \xi_2 d \xi_3 ,
\end{split}
\end{equation}
occurring in the expression of $H_{I,\infty}^{(1)\, n}$, we obtain by \cite[Proposition 1.2.3]{GlimmJaffe} that
\begin{equation*}
\big \| \mathrm{d} \Gamma ( \1_{[0,\sigma_n]}(p_1) p_1 )^{- \frac12}  \mathrm{R}_1 ( N_{ Z^0 } )^{-\frac12}  \big \| \le C \big \|Êp_1^{-\frac12} \1_{[0,\sigma_n]}(p_1) F^{(1)}(\xi_1, \xi_2, \xi_3) \big \|.
\end{equation*}
The estimates of Appendix \ref{appendixA} show that
\begin{equation*}
\big \|Êp_1^{-\frac12} \1_{[0,\sigma_n]}(p_1) F^{(1)}(\xi_1, \xi_2, \xi_3) \big \| \le C \sigma_n ,
\end{equation*}
and therefore, since in addition $\mathrm{d} \Gamma ( \1_{[0,\sigma_n]}(p_1) p_1 ) \le  \check{H}_\infty^n$, it follows that
\begin{align*}
\big \| (  \check{H}_\infty^n )^{- \frac12}  \mathrm{R}_1 ( N_{ Z^0 } + 1 )^{-\frac12}  \big \| \le C_1 \sigma_n .
\end{align*}
The other terms in $H_{I,\infty}^{(1)\, n}$ are treated similarly.

To prove \eqref{eq:rel_Ntau} for $j=2$, we modify the argument as follows. Consider for instance the term $\mathrm{R}_2$ in $H_{I,\infty}^{(2)\, n}$ defined by
\begin{equation}\nonumber
\begin{split}
 \mathrm{R}_2 := &
  \int F^{(1)}(\xi_1, \xi_2, \xi_3)
 \1_{\{ p_1 \le \sigma_n , \sigma_n\leq p_2\}}(p_1, p_2) b_+^*(\xi_1) b_-^*(\xi_2) a^*(\xi_3) d \xi_1 d \xi_2 d \xi_3 ,
\end{split}
\end{equation}
Applying again \cite[Proposition 1.2.3]{GlimmJaffe} yields
\begin{equation*}
\big \| \mathrm{d} \Gamma ( \1_{[0,\sigma_n]}(p_1) p_1 )^{- \frac12} ( N_{ Z^0 } )^{-\frac12}  \mathrm{R}_2  \big \| \le C \big \|Êp_1^{-\frac12} \1_{[0,\sigma_n]}(p_1) F^{(1)}(\xi_1, \xi_2, \xi_3) \big \| \le C \sigma_n.
\end{equation*}
Since $( N_{ Z^0 } )^{-1/2} \mathrm{R}_2 = \mathrm{R}_2 ( N_{ Z^0 } + 1 )^{-1/2}$ and since $\mathrm{d} \Gamma ( \1_{[0,\sigma_n]}(p_1) p_1 ) \le  \check{H}_\infty^n$, we obtain
\begin{align*}
\big \| (  \check{H}_\infty^n )^{- \frac12}  \mathrm{R}_2 ( N_{ Z^0 } + 1 )^{-\frac12}  \big \| \le C_2 \sigma_n .
\end{align*}
We can argue in the same way for the other terms in $H_{I,\infty}^{(2)\, n}$. This concludes the proof.
\end{proof}
In the sequel, we use the following specific definition for $(\sigma_n)$. Let
\begin{equation}
 \sigma_0 := \mz ,\quad \sigma_{n+1} := \gamma\sigma_n ,\quad \gamma=\frac14 .
\end{equation}
\begin{prop}[Spectral gap]\label{prop:spectral-gap}
Suppose that Hypothesis~\ref{hypothesis-1} holds. There exist $\cg<\infty$ and $g_0>0$ such that for all $ 0\leq g \leq g_0$ and  all $n\in \N$,  $E_n:=\inf\sigma(K_n)$ is a simple isolated eigenvalue, with
\begin{equation*}
   \inf\left( \sigma(K_{n}) \setminus \{E_{n}\}\right) - E_{n} \geq (1-\cg g)\sigma_n .
\end{equation*}
\end{prop}
\begin{proof}
First we claim that for all $n \in \N$, $E_{n+1} \le E_n$. Indeed, let $1>\epsilon>0$, and $\psi_n\in\mathcal{H}_n$ be an $\epsilon$-approximate ground state for the ground state energy $E_n$ of $H_n$, namely, $\|\psi_n\|=1$ and
\begin{equation*}
 E_n \leq \la \psi_n, K_n \psi_n\ra \leq E_n +\epsilon.
\end{equation*}
Let
\begin{equation*}
\tilde{\psi}_n := \psi_n\otimes (\Omega_{n+1}^n\otimes \Omega_{n+1}^n) \in \mathcal{H}_{n+1},
\end{equation*}
where $\Omega_{n+1}^n$ is the vacuum state in $\mathfrak{F}_{n+1}^n$. Then $\| \tilde{\psi}_{n}\|=1$ and
\begin{align}
 E_{n+1} & \leq \la \tilde{\psi}_n, K_{n+1} \tilde{\psi}_n \ra \notag \\
                 & = \la \tilde{\psi}_n , (K_n\otimes \1_{n+1}^n\otimes \1_{n+1}^n)\tilde{\psi}_n \rangle
                  + \la \tilde{\psi}_n, \1_n \otimes \check{H}_{n+1}^n \tilde{\psi}_n\ra
                  + g \la \tilde{\psi}_n,    H_{I,n+1}^{\ \ n}\tilde{\psi}_n\ra \notag \\
                  & \leq E_n +\epsilon , \label{eq:st1-1}
\end{align}
since by construction, $ \la \tilde{\psi}_n, \1_n \otimes \check{H}_{n+1}^n \tilde{\psi}_n\ra = \la \tilde{\psi}_n,    H_{I,n+1}^{\ \ n}\tilde{\psi}_n\ra=0$. Taking the limit $\epsilon$ to $0$ implies
\begin{equation}\label{eq:st1-2}
E_{n+1}\leq E_n .
\end{equation}

Now, using a min-max principle, we prove by induction the existence of a unique ground state for $H_n$ and the existence of a spectral gap. \\

\noindent \textbf{Basis}: For $n=0$, we have
\begin{equation*}
K_{n=0} = H_{0,n=0} + g H_{I,0}.
\end{equation*}
Moreover, $\sigma(H_{0,n=0}) \subset\left( \{0\} \cup [\sigma_0, \infty ) \right)$. For $\tilde{a}$ and $\tilde{b}$ are given by Lemma~\ref{lm:acompleter}, let
\begin{equation*}
 \cg := (\sigma_0\tilde{a} + 4 \tilde{b}) .
\end{equation*}
Then using the norm relative bounds of $H_{I,0}$ with respect to $H_{0, n=0}$ (see Lemma~\ref{cor2:relative-bound}), there exists $g_0>0$ such that for all $0 \le g \leq g_0$,
\begin{equation*}
 \inf\left( \sigma(K_{n=0}) \setminus \{E_{n=0}\}\right) - E_{n=0} \geq (1-\cg g)\sigma_0 ,
\end{equation*}
which proves existence of a ground state for $H_{n=0}$.

Now we have
\begin{equation}\nonumber
\begin{split}
\lambda_0 & := \sup_{\phi\in\mathcal{H}_0, \phi\neq 0}\quad \inf_{\psi\in\mathfrak{D}(K_{n=0}), \la\psi,\phi\ra=0, \|\psi\|=1} \la \psi, K_{n=0} \psi \ra \\
& \geq
\inf_{\psi\in\mathfrak{D}(K_{n=0}), \la\psi,\Omega_{n=0}\otimes\Omega_{n=0}\otimes\Omega_{Z^0}\ra=0, \|\psi\|=1} \la \psi, K_{n=0} \psi \ra ,
\end{split}
\end{equation}
where $\Omega_{n}$ is the vacuum state in $\gF_n$. Thus, from the relative bound in Lemma~\ref{cor2:relative-bound}, we obtain that there exists a constant c such that for all $g\geq 0$ small enough,
\begin{equation*}
 \lambda_0 > \sigma_0 - c g (\sigma_0 + 1) .
\end{equation*}
For all $g\geq 0$ small enough, we thus have
\begin{equation*}
 \lambda_0 >0  \geq E_{n=0},
\end{equation*}
which proves nondegeneracy of the ground state of $H_{n=0}$ and thus concludes the proof for the step $n=0$. \\

\noindent \textbf{Induction}: Assume for a given $n$ that $E_n$ is a simple isolated eigenvalue of $K_n$, and that the size of the gap is such that
\begin{equation}\label{eq:ind-1}
  \inf\left( \sigma(K_{n}) \setminus \{E_{n}\}\right) - E_{n} \geq (1-\cg g)\sigma_n .
\end{equation}

Recall that $\tilde{K}_n = K_n \otimes\1_{n+1}^n + \1_n\otimes \check{H}_{n+1}^n$. By the induction assumption \eqref{eq:ind-1}, $\{E_n\} \subset \sigma (K_n \otimes \1_{n+1}^n\otimes \1_{n+1}^n) \subset
\left(  \{ E_n \} \cup [E_n + (1-\cg g )\sigma_n, \infty)\right)$. Obviously, $\sigma(\1_n\otimes \check{H}_{n+1}^n) = \{ 0 \} \cup [\sigma_{n+1}, \infty)$.
Moreover, for all $0\leq g \leq 1/(2(\sigma_0\tilde{a} + 4\tilde{b}))$, we have
$\sigma_{n+1} = \gamma\sigma_n \leq (1-\cg g )\sigma_n$. Hence we obtain
\begin{equation}\label{eq:ind-2}
  \{E_n\} \subset \sigma(\tilde{K}_n) \subset (\{E_n\} \cup [E_n + \sigma_{n+1}, +\infty)) .
\end{equation}
We next apply min-max perturbation argument to estimate the maximum left shift of $[E_n + \sigma_{n+1}, \infty)$ under the perturbation $g H_{I,n+1}^{\ \ n}$. Let $\lambda_k : = \inf \left( \sigma(K_k) \setminus \{E_k\}\right)$. Then
\begin{equation*}
 \lambda_{n+1} = \sup_{\psi\in \mathcal{H}_{n+1},\psi\neq 0}\ \ \inf_{\la \phi,\psi\ra= 0, \phi\in\D(K_{n+1}), \|\phi\|=1} \ \la \phi, K_{n+1} \phi\ra .
\end{equation*}
Picking $\psi = \tilde{\psi}_{n+1}$ as the unique ground state of $\tilde{K}_n = K_n \otimes \1_{n+1}^n \otimes\1_{n+1}^n
+ \1_n\otimes\check{H}^n_{n+1}$ associated with the eigenvalue $E_n = \inf \sigma(\tilde{K}_n)$, we get
\begin{equation}\label{eq:ind-3}
 \lambda_{n+1} \geq \inf _{\la \phi,\tilde{\psi}_{n+1}\ra = 0, \phi\in\D(K_{n+1}), \|\phi\|=1} \ \la \phi, K_{n+1}. \phi\ra .
\end{equation}
For all $\phi\perp \tilde{\psi}_{n+1}$, $\phi\in\D(K_{n+1})$, $ \|\phi\|=1$, we have, using \eqref{eq:ind-2} and \eqref{eq:st1-2},
\begin{align}
\la\phi, \tilde{K}_n\phi\ra &= \left\la \phi, (K_n\otimes \1_{n+1}^n\otimes\1_{n+1}^n + \1_n\otimes \check{H}_{n+1}^n) \phi\right\ra \notag \\
& \geq E_n + \sigma_{n+1} \geq  E_{n+1} + \sigma_{n+1} .\label{eq:ind-4}
\end{align}
Therefore, using \eqref{eq:form-bnd-a-b} of Lemma~\ref{lm:acompleter}, we obtain
\begin{equation*}
\begin{split}
  & \lambda_{n+1} \geq \inf_{\la \phi,\tilde{\psi}_{n+1}\ra =0, \phi\in\D(K_{n+1}, \|\phi\|=1}
  \la \phi , (\tilde{K}_n + g H_{I, n+1}^{\ \ n}) \phi\ra \\
  & = \inf_{\la \phi,\tilde{\psi}_{n+1}\ra =0, \phi\in\D(K_{n+1}, \|\phi\|=1}
  \left(\la \phi , (1-\tilde{a}g\sigma_n)\tilde{K}_n \phi\ra
  +
  \la \phi , \tilde{a}\sigma_n g H_{I, n+1}^{\ \ n} \phi\ra
  \right) \\
  & \geq (1-\tilde{a}g\sigma_n) (E_{n+1} + \sigma_{n+1}) - \tilde{b} g\sigma_n \\
  & \geq E_{n+1} + \left(1 - g( \sigma_0 \tilde{a} + 4 \tilde{b})  \right)\sigma_{n+1},
\end{split}
\end{equation*}
where in the last inequality, we used $\sigma_{n+1} = \gamma\sigma_n = \frac14 \sigma_n$ and $E_n\leq 0$.

Thus
\begin{equation*}
 \inf\left( \sigma(K_{n+1})\setminus\{ E_{n+1}\}\right) - E_{n+1} \geq (1 - \cg g ) \sigma_{n+1},
\end{equation*}
and thus concludes the proof of the induction.
\end{proof}
The proof of Theorem~\ref{thm:GS} on the existence of a ground state is a consequence of the following proposition.
\begin{prop}\label{prop:cv_gs}
Suppose that Hypothesis~\ref{hypothesis-1} holds. There exists $g_0>0$ such that, for all $0 \le g \le g_0$, we have
\begin{itemize}
\item[{(i)}] For all $\psi\in\D (H_0)$,
 \begin{equation*}
 \lim_{n\to\infty}H_n \psi = H\psi .
 \end{equation*}
\item[{(ii)}] For all $n \in \mathbb{N}$, $H_n$ has a unique (up to a phase) normalized ground state  $\phi_n$ associated to the eigenvalue $E_n = \inf \sigma( K_n )$.
\item[{(iii)}]  There exists $\delta_g\in (0,1)$ such that, for all $n\in \mathbb{N}$,
\begin{equation*}
  \langle \phi_n, ( P_{\Omega_D} \otimes P_{\Omega_{Z^0}}  ) \phi_n \rangle \geq 1 - \delta_g ,
\end{equation*}
with $\lim_{g\to 0}\delta_g= 0$, where $P_{\Omega_{Z^0}}$ is the  projection onto the bosonic Fock vacuum $\Omega_{Z^0}$ and $P_{\Omega_D}$ is the projection onto the fermionic Fock vacuum state $\Omega_D$.
\end{itemize}
\end{prop}
Before giving the proof of this proposition, we need to state the following three lemmata.
\begin{lem}[Pull-through formula]\label{lem:A}
Suppose that Hypothesis~\ref{hypothesis-1} holds. For a.e. $\xi \in \Sigma$, let
\begin{equation*}
\begin{split}
 V_+(\xi) & :=	 - \sum_{j=1,2} \int F^{(j)}(\xi, \xi_2, \xi_3) b^*_-(\xi_2) a(\xi_3) d \xi_2 d \xi_3 , \\
 V_-(\xi)  & :=    - \sum_{j=1,2} \int F^{(j)} (\xi_1,\xi,\xi_3) b^*_+(\xi_1)a(\xi_3) d \xi_1 d \xi_3 .
\end{split}
\end{equation*}
 For all $n\in\N$, $\psi\in\D(H_{0})$ and a.e. $\xi \in \Sigma$, we have that
 \begin{equation}\nonumber
\begin{split}
 H_n b_+(\xi) \psi & = b_+(\xi)H_n\psi - \omega_1(k)b_+(\xi)\psi + g V_+(\xi)\psi , \\
 H_n b_-(\xi) \psi  & = b_-(\xi)H_n\psi  -  \omega_2(k)b_-(\xi)\psi  + g V_-(\xi) \psi .
\end{split}
\end{equation}
\end{lem}
\begin{proof}
The proof is a direct consequence of the anticanonical commutation relations for $b_\pm(\xi)$ and $b_\pm^*(\xi)$. Details of the proof can be found e.g. in \cite{AGG}.
\end{proof}
\begin{lem}\label{lem:B}
Suppose that Hypothesis~\ref{hypothesis-1} holds. For all $\psi\in\D(H_{0})$ and a.e. $\xi \in \Sigma$, we have that
\begin{equation}\label{eq:est-V1}
\begin{split}
 \| V_+ (\xi) \psi \| & \leq
 \sum_{j=1,2}  \| F^{(j)}(\xi,\cdot,\cdot)\|_{L^2( \Sigma\times\Sigma_3, d \xi_2 d \xi_3) }
  \|(N_{Z^0}+\1)^\frac12 \psi \|  , \\
 \| V_-(\xi) \psi \| & \leq \sum_{j=1,2}  \| F^{(j)}(\cdot,\xi,\cdot)\|_{L^2( \Sigma\times\Sigma_3, d \xi_1 d \xi_3) }
  \| (N_{Z^0} + \1) ^\frac12 \psi \|  .
\end{split}
\end{equation}
\end{lem}
\begin{proof}
It is a straightforward application of the $N_\tau$ estimates of \cite{GlimmJaffe}, using arguments similar to those employed in the proof of Proposition \ref{prop:relative-bound}. Details can be found in \cite[Proposition~2.3]{AGG}.
\end{proof}
\begin{lem}\label{lem:C}
Suppose that Hypothesis~\ref{hypothesis-1} holds. Let $\phi_n$ be a normalized ground state of $H_n$ for some $n \in \mathbb{N}$. Then we have that
\begin{equation}\label{eq:lemC-1}
\begin{split}
 \| N_+^\frac12 \phi_n \|  & \leq \frac{1}{\sqrt{\mz}} g
 \left( \left\|\frac{F^{(1)}}{p_1}\right\| + \left\| \frac{F^{(2)}}{p_1}\right\| \right)  \| (H_0 + \1)^\frac12 \phi_n \|  , \\
  \| N_-^\frac12 \phi_n \|  & \leq \frac{1}{\sqrt{\mz}} g
 \left( \left\|\frac{F^{(1)}}{p_2}\right\| + \left\| \frac{F^{(2)}}{p_2}\right\| \right)  \| (H_0 + \1)^\frac12 \phi_n \| . \\
\end{split}
\end{equation}
\end{lem}
\begin{proof}
Let $\phi_n$ be a normalized ground state of $H_n$ (assuming it exists). From Lemma~\ref{lem:A}, we obtain that
\begin{equation}\label{eq:N-bd-1}
 (H_n - E_n + \omega_1(p_1) ) b_+(\xi) \phi_n = g V_+(\xi) \phi_n ,
\end{equation}
for a.e. $\xi \in \Sigma$. Therefore, using that $H_n - E_n \geq 0$ and \eqref{eq:N-bd-1}, we get
\begin{equation*}
 \| b_+(\xi)\phi_n \|  = \| g (H_n-E_n + \omega_1(p_1) )^{-1} V_+(\xi)\phi_n \| \leq\frac{g}{p_1} \| V_+(\xi)\phi_n\| .
\end{equation*}
Together with the identity $\int \| b_+(\xi)\phi_n \|^2d\xi = \| N_+^\frac12 \phi_n\|^2$ and the first inequality in \eqref{eq:est-V1}, this concludes the proof of the first inequality in \eqref{eq:lemC-1}. The proof of the second inequality is similar.
\end{proof}
Equipped with the spectral gap result, namely Proposition~\ref{prop:spectral-gap}, and the above three lemmata, we can now prove  Proposition~\ref{prop:cv_gs}.
\begin{proof}[Proof of  Proposition~\ref{prop:cv_gs}]
According to the relative bounds of Proposition~\ref{prop:relative-bound} and Lemma~\ref{cor2:relative-bound}, there exists a constant $C>0$ and $\tilde{g_0}$, such that for $0 \le g \le \tilde{g_0}$, for all $\psi\in\D(H_0)$, and for all $n$,
\begin{equation*}
 \left\| \left( H - H_n\right) \psi \right\| \leq C g \sum_{j=1,2} \left\|  (1 - \chi_{\{\sigma_n\leq p_1\}}
 \chi_{\{\sigma_n\leq p_2\}})  F^{(j)} \right\|  \left( \| H_0 \psi\| + \|\psi\|) \right) .
\end{equation*}
Since $\displaystyle\lim_{n\to\infty}\left\|(1 - \chi_{\{\sigma_n\leq p_1\}}
 \chi_{\{\sigma_n\leq p_2\}})  F^{(j)} \right\| =0$, this proves $\rm{(i)}$.

 By Proposition \ref{prop:spectral-gap}, $E_n = \inf \mathrm{spec} ( K_n )$ is a simple, isolated eigenvalue of $K_n$. Since $H_n = K_n \otimes \1_{\infty}^n \otimes \1_\infty^n + \1_n \otimes \check{H}^n_\infty$ and since the spectrum of $\check{H}^n_\infty$ is composed of the simple eigenvalue $0$ and the semi-axis of absolutely continuous spectrum $[ 0 , \infty )$, we deduce that $E_n$ is a simple eigenvalue of $H_n$. Hence $\rm{(ii)}$ is proven.

 The proof of $\rm{(iii)}$ is as follows. Let $\phi_n$ be a normalized ground state of $H_n$. Then we have the identity
\begin{equation*}
\begin{split}
 0 & = \big ( ( P_{\Omega_D} \otimes P_{\Omega_{Z^0}}^\perp ) (H_n - E_n) \big ) \phi_n \\
   & = \big ( ( P_{\Omega_D} \otimes P_{\Omega_{Z^0}}^\perp ) \left( H_{Z^0} - E_n \right) \big ) \phi_n +  g \big ( ( P_{\Omega_D} \otimes P_{\Omega_{Z^0}}^\perp ) H_{I,n}  \big ) \phi_n .
\end{split}
\end{equation*}
Since $E_n \leq 0$, we obtain that
\begin{equation*}
\mz \la \phi_n, ( P_{\Omega_D} \otimes P_{\Omega_{Z^0}}^\perp )  \phi_n \ra + g \la \phi_n, ( P_{\Omega_{Z^0}}^\perp\otimes P_{\Omega_D} ) H_{I,n} \phi_n \ra \leq 0 ,
\end{equation*}
and therefore
\begin{equation}\label{eq:prop3.7-10}
 \la \phi_n, ( P_{\Omega_D} \otimes P_{\Omega_{Z^0}}^\perp ) \phi_n \ra
 \leq - \frac{g}{\mz} \la  \phi_n, ( P_{\Omega_D} \otimes P_{\Omega_{Z^0}}^\perp ) H_{I,n} \phi_n \ra .
\end{equation}
Proposition~\ref{prop:relative-bound} yields that there exists $C>0$ such that, for all $n$,
\begin{equation}\label{eq:prop3.7-11}
 \left| \la ( P_{\Omega_D} \otimes P_{\Omega_{Z^0}}^\perp ) H_{I,n}\phi_n, \phi_n\ra \right| \leq C .
\end{equation}
Equations \eqref{eq:prop3.7-10} and \eqref{eq:prop3.7-11} imply
\begin{equation}\label{eq:prop3.7-12}
  \la ( P_{\Omega_D} \otimes P_{\Omega_{Z^0}}^\perp ) \phi_n, \phi_n \ra  \leq C \frac{g}{\mz} .
\end{equation}
We also have that there exists $C>0$ such that, for all $n$,
\begin{equation}\label{eq:prop3.7-13}
\begin{split}
 \left| \left\la ( P_{\Omega_D}^\perp \otimes \1 ) \phi_n,  \phi_n \right\ra \right|
 & \leq \left( \| N_+^\frac12 P_{\Omega_D}^\perp \phi_n \| ^2+ \|N_-^\frac12 P_{\Omega_D}^\perp \phi_n \|^2\right) \\
 & \leq C\frac{g^2}{\mz} \sum_{i=1,2} \sum_{j=1,2} \|F^{(j)}/p_i\| ^2 \|(H_0+1)^\frac12 \phi_n\|^2 ,
\end{split}
\end{equation}
where in the last inequality, we used Lemma~\ref{lem:C}.
The identity
\begin{equation*}
  \la \phi_n, ( P_{\Omega_D} \otimes P_{\Omega_{Z^0}} ) \phi_n \ra
  = 1 -  \la \phi_n,  ( P_{\Omega_D} \otimes P_{\Omega_{Z^0}}^\perp ) \phi_n \ra
         -  \la \phi_n, ( P_{\Omega_D}^\perp \otimes \1 ) \phi_n \ra ,
\end{equation*}
together with Equations \eqref{eq:prop3.7-12}--\eqref{eq:prop3.7-13}, concludes the proof of ${\rm (iii)}$.
\end{proof}
\begin{proof}[Proof of  Theorem~\ref{thm:GS}]
Following the proof of \cite[Theorem II.8]{BFS98}, the existence of a ground state is now a consequence of Proposition~\ref{prop:cv_gs}. Since
$(\phi_n)$ is a bounded sequence, it converges weakly, up to a subsequence again denoted by $(\phi_n)$, to a vector $\phi$.
The property (iii) of Proposition~\ref{prop:cv_gs}, weak convergence of $(\phi_n)$ and compacity of $P_{\Omega_{Z^0}}\otimes P_{\Omega_D}$ yields $\phi\neq 0$.

From \eqref{eq:st1-2}, we know that $(E_n)$ is non increasing. Moreover, \cite[\S~V.4, Theorem~4.11]{Kato} and the relative bound of Lemma~\ref{cor2:relative-bound} implies that $(E_n)$ is a bounded sequence. Thus there exists $E\leq 0$ such that $E_n\to E$.

Lemma~\ref{cor2:relative-bound} implies that for all $\psi$ in a common core of $H_n$ ($n\geq 0$) and $H$, we have $H_n\psi \to H\psi$. Since we also have $E_n\to E$ and $\phi_n\rightharpoonup\phi$, we obtain (see e.g. \cite{BFS98} or \cite[Lemma ~4.2]{BDG2}) that $E$ is the ground state energy of $H$, and is an eigenvalue of $H$ with associated eigenfunction $\phi$.

To prove the uniqueness of the ground state, we follow again \cite{BFS98}. Suppose by contradiction that there exists $\phi'$ such that $H \phi' = E \phi'$, $\| \phi' \| = 1$ and $\langle \phi , \phi' \rangle = 0$. Since by Proposition \ref{prop:cv_gs}, $\phi_n$ is a unique normalized ground state of $H_n$, we can write
\begin{align}
 | \langle \phi , \phi' \rangle |^2 & = \lim_{n \to \infty}  | \langle \phi_n, \phi' \rangle |^2 \notag \\
&= \lim_{n\to\infty} \big \langle \phi' , \mathds{1}_{ \{ E_n \} }( H_n ) \phi' \big \rangle \notag \\
&= 1 - \lim_{n\to\infty} \big \langle \phi' , ( \1_{ \mathcal{H} } - \mathds{1}_{ \{ E_n \} }( H_n ) ) \phi' \big \rangle. \label{eq:ccc1}
\end{align}
Note that $\mathds{1}_{ \{ E_n \} }( H_n ) = \mathds{1}_{ \{ E_n \} }( K_n ) \otimes P_{ \Omega^n_\infty \otimes \Omega^n_\infty }$, where $P_{ \Omega^n_\infty \otimes \Omega^n_\infty }$ denotes the orthogonal projection onto the vacuum in $\mathfrak{F}^n_\infty \otimes \mathfrak{F}^n_\infty$. Decomposing
\begin{align}
\1_{ \mathcal{H}Ê} - \mathds{1}_{ \{ E_n \} }( H_n ) =& ( \1_n  - \mathds{1}_{ \{ E_n \} }( K_n ) ) \otimes P_{ \Omega^n_\infty \otimes \Omega^n_\infty } \notag \\
& + \1_n \otimes ( \1^n_\infty \otimes \1^n_\infty - P_{ \Omega^n_\infty \otimes \Omega^n_\infty } ), \label{eq:ccc2}
\end{align}
we want to estimate
\begin{equation}
\big \langle \phi' , ( ( \1_n  - \mathds{1}_{ \{ E_n \} }( K_n ) ) \otimes P_{ \Omega^n_\infty \otimes \Omega^n_\infty } ) \phi' \big \rangle , \label{eq:ccc3}
\end{equation}
and
\begin{equation}
\big \langle \phi' , ( \1_n \otimes ( \1^n_\infty \otimes \1^n_\infty - P_{ \Omega^n_\infty \otimes \Omega^n_\infty } ) ) \phi' \big \rangle . \label{eq:ccc4}
\end{equation}
To estimate \eqref{eq:ccc4}, we use Lemma \ref{lem:C}, from which it follows that
\begin{equation}
\big \langle \phi' , ( \1_n \otimes ( \1^n_\infty \otimes \1^n_\infty - P_{ \Omega^n_\infty \otimes \Omega^n_\infty } ) ) \phi' \big \rangle \le \big \langle \phi' , (N_+ + N_- ) \phi' \big \rangle \le C_0 g , \label{eq:ddd1}
\end{equation}
where $C_0$ denotes a positive constant depending on $\phi'$ and $F^{(j)}$, $j=1,2$.

To estimate \eqref{eq:ccc3}, we use Proposition \ref{lem:C}, which gives
\begin{align*}
& \big \langle \phi' , ( ( \1_n  - \mathds{1}_{ \{ E_n \} }( K_n ) ) \otimes P_{ \Omega^n_\infty \otimes \Omega^n_\infty } ) \phi' \big \rangle \notag \\
& \le \frac{1}{( 1 - C g ) \sigma_n}  \big \langle \phi' , ( ( K_n - E_n ) \otimes P_{ \Omega^n_\infty \otimes \Omega^n_\infty } ) \phi' \big \rangle \notag \\
& = \frac{1}{( 1 - C g ) \sigma_n}  \big \langle \phi' , ( H_n - E_n ) ( \1_n \otimes P_{ \Omega^n_\infty \otimes \Omega^n_\infty } ) \phi' \big \rangle \notag \\
& \le \frac{1}{( 1 - C g ) \sigma_n}  \big \langle \phi' , ( H_n - E_n ) \phi' \big \rangle .
\end{align*}
Since $H \phi' = E \phi'$, we obtain
\begin{align*}
& \big \langle \phi' , ( ( \1_n  - \mathds{1}_{ \{ E_n \} }( K_n ) ) \otimes P_{ \Omega^n_\infty \otimes \Omega^n_\infty } ) \phi' \big \rangle \notag \\
& \le \frac{E - E_n }{( 1 - C g ) \sigma_n}  + \frac{1}{( 1 - C g ) \sigma_n}  \big \langle \phi' , ( H_n - H ) \phi' \big \rangle.
\end{align*}
Let $\psi_n$ be a normalized ground state of $K_n$. Observe that
\begin{equation}\label{eq:E<En}
E \le \langle ( \psi_n \otimes \Omega^n_\infty \otimes \Omega^n_\infty ) , H ( \psi_n \otimes \Omega^n_\infty \otimes \Omega^n_\infty ) \rangle = \langle \psi_n , K_n \psi_n \rangle = E_n.
\end{equation}
Hence
\begin{align}
& \big \langle \phi' , ( ( \1_n  - \mathds{1}_{ \{ E_n \} }( K_n ) ) \otimes P_{ \Omega^n_\infty \otimes \Omega^n_\infty } ) \phi' \big \rangle \le \frac{1}{( 1 - C g ) \sigma_n}  \big \langle \phi' , ( H_n - H ) \phi' \big \rangle. \label{eq:ddd2}
\end{align}
Now, by Lemma \ref{lm:rel_Ntau}, we have that
\begin{align}
 \left \langle \phi' , \left( H - H_n\right) \phi' \right\| &\leq C g \sigma_n \| (  \check{H}_\infty^n )^{ \frac12} \phi' \| \| ( N_{Z^0} + 1 )^{\frac 12} \phi' \| \le C g \sigma_n. \label{eq:ddd3}
\end{align}

Combining \eqref{eq:ccc2}, \eqref{eq:ddd1}, \eqref{eq:ddd2} and \eqref{eq:ddd3}, we obtain that
\begin{equation*}
\big \langle \phi' , ( \1_{ \mathcal{H} } - \mathds{1}_{ \{ E_n \} }( H_n ) ) \phi' \big \rangle \le C'_0 g ,
\end{equation*}
for some positive constant $C'_0$ independent of $g$. By \eqref{eq:ccc1}, this implies that, for $g$ small enough, $\langle \phi , \phi' \rangle \neq 0$, which is a contradiction. Hence the theorem is proven.
\end{proof}

To conclude this section, we estimate the difference of the ground state energies $E$ and $E_n$. This estimate will be used several times in the sequel.
\begin{lem}\label{lm:rel_Ntau2}
Suppose that Hypothesis \ref{hypothesis-1} holds. There exist $g_0 > 0$ and $C > 0$ such that, for all $0 \le g  \le g_0$ and $n \in \mathbb{N}$,
\begin{align}
| E - E_n | \le C g \sigma_n^{2}. \label{eq:rel_Ntau2}
\end{align}
\end{lem}
\begin{proof}
It has been observed in the proof of the previous theorem that $E \le E_n$ (see \eqref{eq:E<En}). Hence it remains to verify that
\begin{equation*}
E_n \le E + C g \sigma_n^2 ,
\end{equation*}
for some positive constant $C$. Let $\phi$ be a ground state of $H$. We have that
\begin{equation}\label{eq:En<E}
E_n \le \langle \phi , H_n \phi \rangle = E + \langle \phi , (H_n - H) \phi \rangle .
\end{equation}
By Lemma \ref{lm:rel_Ntau}, this implies that
\begin{equation}\label{eq:E_n<E}
E_n - E \le C g \sigma_n \big \| (  \check{H}_\infty^n )^{ \frac12} \phi \big \| \big \| ( N_{Z_0}Ê+ 1 )^{\frac12} \phi \big \|.
\end{equation}
Since $N_{Z_0}$ is relatively $H_0$-bounded, we have that $\big \| ( N_{Z_0}Ê+ 1 )^{\frac12} \phi \big \| \le C$, and adapting Lemma \ref{lem:C} in a straightforward way, one verifies that
\begin{align*}
\big \| (  \check{H}_\infty^n )^{\frac12} \phi \big \|  & \leq C | g | \sum_{j=1,2} \sum_{ l = 1 , 2 }  \left\| \1_{ [ 0 , \sigma_n ] }( p_l ) \frac{F^{(j)}}{p_l^{\frac12}}\right\| \big \| (H_0 + \1)^\frac12 \phi \big \| \le C \sigma_n ,
\end{align*}
the second inequality being a consequence of the estimates of Appendix \ref{appendixA}. Together with \eqref{eq:E_n<E}, this concludes the proof of the lemma.
\end{proof}
%
%


\subsection{Proof of Theorem~\ref{thm:es}}
To prove Theorem \ref{thm:es}, we follow the strategy of \cite{DG} and \cite{Takaesu2014}. We recall the following standard notations that will be used in this section: for any $g \in \mathfrak{h}_c$,
\begin{equation*}
b_{\pm}^*(g) := \int_\Sigma g( \xi ) b_{\pm}^*(\xi) d \xi  , \quad b_{\pm}(g) := \int_\Sigma \overline{g( \xi )} b_{\pm}(\xi) d \xi ,
\end{equation*}
and, for any $f \in L^2( \Sigma_3 )$,
\begin{equation*}
 a^*(f) = \int_{\Sigma_3} {f(\xi_3)} a^*(\xi_3) d\xi_3, \quad a(f) = \int_{\Sigma_3} \overline{f(\xi_3)} a(\xi_3) d\xi_3.
\end{equation*}
We begin with the following technical lemma.
\begin{lem}\label{lem:ess-spectrum-comm-estimates}
Suppose that Hypothesis~\ref{hypothesis-1} holds. Let $(f_n)$ and $(g_n)$ be respectively two sequences of elements in $\D(p_1)$ and $\D(p_2)$ such that $\wl f_n = 0$ and $\wl g_n=0$. Then, for all $\psi\in\D(H)$,
\begin{equation}\nonumber
\begin{split}
   \lim_{n\to\infty} [H_I,  b_+(f_n) + b_-(g_n) + b_+^*(f_n) + b_-^*(g_n)]\psi = 0 .
\end{split}
\end{equation}
\end{lem}
\begin{proof}
Using the canonical anticommutation relations for the creation and annihilation operators of the neutrinos, we obtain that
\begin{equation*}
[H_I^{(j)}, b_+^*(f_n)] = [H_I^{(j)}, b_-^*(g_n)] = 0,
\end{equation*}
for $j=1,2$. Taking the adjoints also shows that 
\begin{equation*}
[(H_I^{(j)})^*, b_+(f_n)] = [(H_I^{(j)})^*, b_-(g_n)] = 0.
\end{equation*}

We next prove that
\begin{equation}\label{eq:lem-esce-1}
\lim_{n\to\infty} [H_I^{(1)}, b_+(f_n) ]\psi = 0 .
\end{equation}
Using the expressions \eqref{eq:inter11b}--\eqref{eq:inter2244} of the interaction and the canonical commutation relations, a direct computation gives
\begin{equation*}
\begin{split}
 &  [ H_I^{(1)}, b_+(f_n)] \psi  = - \int \big \langle f_n , F^{(1)}(\cdot ,\xi_2,\xi_3) \big \rangle_{ \mathfrak{h}_c } b_-^*(\xi_2) a(\xi_3) \psi  d \xi_2 d \xi_3 .
\end{split}
\end{equation*}
The $N_\tau$ estimates of \cite{GlimmJaffe} then imply that
\begin{equation*}
\begin{split}
 \big \| [ H_I^{(1)}, b_+(f_n)] \psi \big \| &\le \Big ( \int \big | \big \langle F^{(1)}(\cdot ,\xi_2,\xi_3) , f_n \big \rangle_{ \mathfrak{h}_c } \big |^2 d \xi_2 d \xi_3 \Big )^{\frac12} \| N_{Z^0}^{\frac12} \psi \| .
\end{split}
\end{equation*}
Since $\wl f_n = 0$, we have that $\langle F^{(1)}(\cdot ,\xi_2,\xi_3) , f_n \big \rangle_{ \mathfrak{h}_c } \to 0$ for a.e. $\xi_2$, $\xi_3$. Moreover, since $(f_n)$ is weakly convergent, it is bounded and therefore, using in addition that $F \in L^2( d\xi_1 d \xi_2 d\xi_3 )$ and that $\| N_{Z^0}^{\frac12} \psi \| < \infty$ (since $\psi \in \D( H)$), we can apply Lebesgue's dominated convergence theorem to obtain that $[ H_I^{(1)}, b_+(f_n)] \psi \to 0$ as $n \to \infty$.

The proofs of $\lim_{n\to\infty} [H_I^{(1)}, b_-(g_n)]\psi = 0$, $\lim_{n\to\infty} [H_I^{(2)}, b_+(f_n)]\psi = 0$ and $\lim_{n\to\infty} [H_I^{(2)}, b_-(g_n)]\psi = 0$ are similar. To conclude the proof of Lemma~\ref{lem:ess-spectrum-comm-estimates}, we are thus left with showing that for $j=1,2$
\begin{equation*}
 \lim_{n\to\infty} [(H_I^{(j)})^*, b_+^*(f_n)]\psi  =
 \lim_{n\to\infty}  [(H_I^{(j)})^*, b_-^*(g_n)]\psi  = 0.
\end{equation*}
We can proceed similarly, writing for instance for the first term
\begin{equation*}
\begin{split}
 \big \| [ (H_I^{(1)})^* , b_+^*(f_n)] \psi \| &= \Big \| \int \big \langle  F^{(1)}(\cdot ,\xi_2,\xi_3) , f_n \big \rangle_{ \mathfrak{h}_c } b_-(\xi_2) a^*(\xi_3) \psi  d \xi_2 d \xi_3 \Big \| \\
&\le \Big ( \int \big | \big \langle F^{(1)}(\cdot ,\xi_2,\xi_3) , f_n \big \rangle_{ \mathfrak{h}_c } \big |^2 d \xi_2 d \xi_3 \Big )^{\frac12} \| (N_{Z^0}+ 1 )^{\frac12} \psi \| ,
\end{split}
\end{equation*}
and then applying Lebesgue's dominated convergence theorem.
\end{proof}
\begin{proof}[Proof of Theorem~\ref{thm:es}]
We adapt the proof of \cite{Takaesu2014} to our case (see also \cite{DG} and \cite{Arai}). For $\omega$ the multiplication operator by $\omega(p) = p$ , and $\gH_c$ given by \eqref{def:H_c}, we set
\begin{equation*}
 K := \omega\oplus\omega , \quad \mbox{ acting on } \gH_c\oplus \gH_c  .
\end{equation*}
The operator $\mathrm{d} \Gamma(\omega)\otimes\1 + \1\otimes \mathrm{d}\Gamma(\omega)$ on $\gF_D = \gF_a(\gH_c)\otimes\gF_a(\gH_c)$ is unitarily equivalent to the operator $\mathrm{d}\Gamma(K)$ on $\gF_a(\gH_c\oplus\gH_c)$. Hence, in the sequel of this proof, by abuse of notation, we write $\mathrm{d}\Gamma(K)$ for $\mathrm{d}\Gamma(\omega)\otimes\1 + \1\otimes \mathrm{d}\Gamma(\omega)$. Similarly, instead of using the notation $b^\sharp((f,g))$ for the fermionic creation and annihiliation operators on $\gF_a(\gH_c\oplus\gH_c)$, we use instead the unitarily equivalent operators $b^\sharp(f,g) := \big(b_+^\sharp(f) + b_-^\sharp(g)\big)\otimes\1$ acting on $\cH = \gF_D \otimes\gF_{Z^0}= \gF_a(\gH_c)\otimes\gF_a(\gH_c)\otimes\gF_{Z^0}$.

Since $\sigma_{\mathrm{ess}}(\omega) = [0,\infty)$, for any $\lambda\geq 0$, we can pick two Weyl sequences $(f_n)$ and $(g_n)$ in $\gH_c$, such that for $n\geq 1$, $\|f_n\| = \|g_n\|=1$, $f_n\in\D(p_1)$, $g_n\in\D(p_2)$, and $\wl_{n\to\infty} f_n =
\wl_{n\to\infty} g_n = 0$,  $\lim_{n\to\infty}(p_1-\lambda)f_n = \lim_{n\to\infty}(p_2-\lambda)g_n = 0$. Thus, $\big(\frac{1}{\sqrt{2}} (f_n, g_n)\big)$ is a Weyl sequence for $K$ and $\lambda$.

For any $\epsilon>0$, there exists a normalized state $\phi_\epsilon\in \mbox{Ran} P_H([E, E+\epsilon)$, where $P_H(.)$ is the spectral measure for $H$ and, recall, $E = \inf\sigma(H)$. We define
\begin{equation*}
 \cH\ni\psi_{n,\epsilon} := \frac{1}{\sqrt{2}} \left( b_+(f_n) + b_+^*(f_n) + b_-(g_n) + b_-^*(g_n) \right) \otimes \1 \ \phi_\epsilon .
\end{equation*}
Using the anticommutation relations for $b_\pm^\sharp$, a straightforward computation shows that
\begin{equation*}
  \| \psi_{n,\epsilon} \|^2 = \frac12 (\phi_\epsilon, (\|f_n\|^2 + \|g_n\|^2) \phi_\epsilon) = 1 .
\end{equation*}
We will now show that one can extract from the family $(\psi_{n,\epsilon})$ a subsequence $(\psi_{n_k, \epsilon_k})$ which is a Weyl sequence for $H$ and $E+\lambda$.

Let $\zeta$ and $\phi$ be in $\D(H)$, and pick $(f,g)\in\D(K)$. Then
\begin{equation}\label{eq:commut-ess-1}
\begin{split}
 & \left\la \zeta, \left[ H, \left( b(f,g) + b^*(f, g) \right) \right] \phi \right\ra  \\
 & = \la H\zeta, \left( b(f,g) + b^*(f, g) \right)\phi \ra
     \ -\ \la \left( b(f,g) + b^*(f, g) \right)\zeta, H\phi \ra \\
 & = \left\la \zeta,  \left[ H_0, \left( b(f,g) + b^*(f, g) \right ) \right] \phi \right\ra
 \ + \ g \left\la \zeta,  \left[ H_I, \left( b(f,g) + b^*(f, g) \right) \right] \phi \right\ra \\
 & = \left\la \zeta,  \left( b^*(K(f,g)) - b(K(f, g)) \right)  \phi \right\ra
 \ + \ g \left\la \zeta,  \left[ H_I, \left( b(f,g) + b^*(f, g) \right) \right] \phi \right\ra .
\end{split}
\end{equation}
For the first term in the second equality, we used the commutation relations 
\begin{equation*}
[\mathrm{d}\Gamma(h), b(w)] = - b(h^* w), \quad [ \mathrm{d} \Gamma(h), b^*(w)] = b^*(h w),
\end{equation*}
for any one-particle state $h$. Replacing $\phi$ by $\phi_\epsilon$, $g$ by $g_n$ and $f$ by $f_n$ in \eqref{eq:commut-ess-1} yields
\begin{equation}\nonumber
\begin{split}
 & \sqrt{2} \la H\zeta, \psi_{n,\epsilon} \ra \\
  & =  \la \zeta, \left( b((f_n,g_n)) + b^*((f_n,g_n))\right) H \phi_\epsilon \ra
    + \la \zeta, (b^*(K(f_n,g_n)) - b(K(f_n,g_n))) \phi_\epsilon \ra \\
    & \quad + g \la \zeta, [H_I, b(f_n,g_n) + b^*(f_n,g_n)] \phi_\epsilon \ra ,
\end{split}
\end{equation}
and therefore
\begin{equation}\nonumber
\begin{split}
 & \sqrt{2} H\psi_{n,\epsilon} \\
  & =  \left( b((f_n,g_n)) + b^*((f_n,g_n))\right) H \phi_\epsilon
    + (b^*(K(f_n,g_n)) - b(K(f_n,g_n))) \phi_\epsilon  \\
    & \quad + g [H_I, b(f_n,g_n) + b^*(f_n,g_n)] \phi_\epsilon .
\end{split}
\end{equation}
Since $\lambda$ is real, we obtain
\begin{equation}\nonumber
\begin{split}
 & \sqrt{2} (H - \lambda - E) \psi_{n,\epsilon} \\
  & =  \left( b((f_n,g_n)) + b^*((f_n,g_n))\right) (H - E )  \phi_\epsilon
    + b^*((K-\lambda)(f_n,g_n))\phi_\epsilon \\
    & \quad - b((K+\lambda)(f_n,g_n)) \phi_\epsilon
      + g [H_I, b(f_n,g_n) + b^*(f_n,g_n)] \phi_\epsilon ,
\end{split}
\end{equation}
and hence
\begin{equation}\nonumber
\begin{split}
 & \|(H - \lambda - E) \psi_{n,\epsilon}\| \\
  & \leq \frac{1}{\sqrt{2}}
  \| \left(b_+(f_n) + b_-(g_n) + b_+^*(f_n) + b_-^*(g_n) \right) (H - E )  \phi_\epsilon \| \\
  & \quad  + \frac{1}{\sqrt{2}}  \| (b_+^*((p_1-\lambda)f_n) + b_-^*((p_2-\lambda)g_n) )\phi_\epsilon\| \\
    & \quad + \frac{1}{\sqrt{2}} \| (b_+((p_1+\lambda)f_n) + b_-((p_2+\lambda)g_n) )\phi_\epsilon\| \\
    & \quad + \frac{g}{\sqrt{2}}  \left\{
    \| [H_I, b_+(f_n) + b_-(g_n)]\phi_\epsilon\|
    +
    \| [H_I, b_+^*(f_n) + b_-^*(g_n)]\phi_\epsilon \| 
    \right\} .
\end{split}
\end{equation}
The latter estimate implies that
\begin{align}
 & \|(H - \lambda - E) \psi_{n,\epsilon}\| \notag \\
  & \leq 2\sqrt{2}  \|(H-E)  \phi_\epsilon \| \notag \\
  & \quad  + \sqrt{2}  (\| (p_1-\lambda)f_n\| + \|(p_2-\lambda)g_n\| ) \|\phi_\epsilon\| \notag \\
&\quad + \sqrt{2} \lambda  (\| b_+(f_n) \phi_\epsilon\| + \| b_-(g_n)\phi_\epsilon\|) \notag \\
    & \quad + \frac{g}{\sqrt{2}}   \left\{
    \| [H_I^{(1)}, b_+(f_n) + b_-(g_n)]\phi_\epsilon\|
    +
    [H_I^{(1)*}, b_+^*(f_n) + b_-^*(g_n)]\phi_\epsilon \|
    \right\} \notag \\
     & \quad + \frac{g}{\sqrt{2}}   \left\{
    \| [H_I^{(2)}, b_+(f_n) + b_-(g_n)]\phi_\epsilon\|
    +
    [H_I^{(2)*}, b_+^*(f_n) + b_-^*(g_n)]\phi_\epsilon \|
    \right\} . \label{eq:commut-ess-2}
\end{align}
By assumption on $\phi_\epsilon$, the first term on the right hand side is smaller than $\epsilon$. Moreover, by construction, since $(f_n)_{n\geq 1}$ (respectively $(g_n)_{n\geq 1}$) is a Weyl sequence for $p_1$ associated to $\lambda$ (respectively for $p_2$ associated to $\lambda$), we also have that $\lim_{n\to\infty}\|(p_1 - \lambda)f_n\|= 0$, and $\lim_{n\to\infty}\|(p_2 - \lambda)g_n\|=0$. According to \cite[Lemma~2.1]{Takaesu2014}, since $\wl g_n = \wl f_n=0$, we get $\lim_{n\to\infty}\|b_+(f_n)\phi_\epsilon\| = \lim_{n\to\infty}\|b_-(g_n)\phi_\epsilon \| =0$. Finally, by Lemma~\ref{lem:ess-spectrum-comm-estimates}, the last four terms in the right hand side of \eqref{eq:commut-ess-2} involving the interaction $H_I$ tend to zero as $n$ tends to infinity. This shows that
\begin{equation*}
 \lim_{\epsilon\to 0} \limsup_{n\to\infty} \|(H-\inf\sigma(H)-\lambda)\psi_{n,\epsilon}\| = 0 ,
\end{equation*}
and thus, by extraction of a subsequence $(n_j,\epsilon_j)_{j\in\N}$ of $(n,\epsilon)_{n\geq 1, \epsilon>0}$, we get
\begin{equation}\label{eq:commut-ess-3}
 \lim_{j\to\infty}\| (H-\inf\sigma(H)-\lambda)\psi_{n_j,\epsilon_j} \| =0.
\end{equation}
Moreover, using again \cite[Lemma~2.1]{Takaesu2014} and the fact that $\wl g_n = \wl f_n=0$, we obtain
\begin{equation}\label{eq:commut-ess-4}
 \wl_{j\to\infty} \psi_{n_j,\epsilon_j} =0.
\end{equation}
Since $\lambda$ is an arbitrary nonnegative real number, \eqref{eq:commut-ess-3} and \eqref{eq:commut-ess-4} conclude the proof.
\end{proof}


\section{Local Decay}\label{section:results}

In this section, we prove Theorem \ref{thm:return}. We follow the strategy explained in Section \ref{S2}. In a first subsection, we introduce the conjugate operator $A_{\sigma_n}$ that is used throughout the proof. The low energy parameter $\sigma_n$ corresponds to the distance from the ground state energy $E$ to the spectral interval $I_{\sigma_n}$ on which the Mourre estimate will be proven. In Section \ref{subsec:regul}, we verify that $H$ is sufficiently regular w.r.t. $A_{\sigma_n}$ in a suitable sense, and we establish uniform bounds on the commutators between $H$ and $A_{\sigma_n}$. Section \ref{subsec:Mourre} is devoted to the proof of the Mourre estimate for $H$ in the spectral interval $I_{\sigma_n}$. Finally, in Section \ref{subsec:unif}, we obtain uniformity in $\sigma_n$ of the local energy decay, and we derive the property of relaxation to the ground state stated in Theorem \ref{thm:return}.

For simplicity of exposition, we prove Theorem \ref{thm:return} with $m_{Z^0}$ replaced by $m_{Z^0} / 3$, but it is not difficult to modify the definitions of $I_{\sigma_n}$ and the function $\varphi$ in \eqref{eq:defvarphi1} below to obtain the result as stated in Theorem \ref{thm:return}.


\subsection{The conjugate operator}\label{subsec:conjugate}

Recall from Proposition \ref{prop:spectral-gap} that there exists a positive constant $C_{\mathrm{gap}}$ such that
$
   \inf\left( \sigma(K_{n}) \setminus \{E_{n}\}\right) - E_{n} \geq (1-C_{\mathrm{gap}} g)\sigma_n ,
$
for all $n \in \N$ and $0 \le g \le g_0$, where $K_n$ denotes the restriction of the infrared cutoff Hamiltonian $H_n$ to the Hilbert space $\mathcal{H}_n = \mathfrak{F}_n \otimes \mathfrak{F}_n \otimes \mathfrak{F}_{Z^0}$. For shortness we set
\begin{equation*}
\rho := 1 - C_{\mathrm{gap}}  g_0 ,
\end{equation*}
where $g_0$ is given by Proposition \ref{prop:spectral-gap}. This insures that the spectrum of $K_n$ has a gap of size at least $\rho \sigma_n$ above the ground state energy, for any $n \in \mathbb{N}$ and for any $g$ small enough.

Let $n \in \mathbb{N}$. We want to prove a Mourre estimate for $H$ on the interval
\begin{equation*}
I_{\sigma_n} := [ E + \rho \sigma_{n+1} / 4 , E + \rho \sigma_n / 3 ].
\end{equation*}
Here we recall that $\gamma = 1/4$ and that $\sigma_{n+1} = \gamma \sigma_n$.

Let $\varphi \in C^\infty_0(\R ; [0,1])$ be such that, for all $x\in\R$,
\begin{equation}\label{eq:defvarphi1}
\varphi (x)
:=\begin{cases}
1, &\text{if } x\in [ \rho \gamma / 4 , \rho / 3 ] , \\
0, & \text{if } x\in\R\setminus ( \rho \gamma / 5 , \rho / 2 ),
\end{cases}
\end{equation}
and let, for all $n \in \mathbb{N}$,
\begin{equation}\label{eq:defvarphi2}
\varphi_{\sigma_n}( x ) = \varphi ( \sigma_n^{-1} x ).
\end{equation}
Notice that $\varphi_{\sigma_n} \in C^\infty_0(\R ; [0,1])$ and that
\begin{equation*}
\varphi_{\sigma_n} (x)
:=\begin{cases}
1, & \text{if } x\in [ \rho \sigma_{n+1} / 4 , \rho \sigma_n / 3 ] , \\
0, & \text{if } x\in\R\setminus ( \rho \sigma_{n+1} / 5 , \rho \sigma_n / 2 ).
\end{cases}
\end{equation*}
The conjugate operator that we choose, acting on $\mathcal{H} = \mathfrak{F}_a \otimes \mathfrak{F}_a \otimes \mathfrak{F}_{Z^0}$, is given by
\begin{equation*}
A_{\sigma_n} := \mathrm{d}\Gamma(a_{1,[0,\sigma_n/2]}) \otimes \1_a \otimes \1_{Z^0} + \1_a \otimes \mathrm{d}\Gamma(a_{2,[0,\sigma_n/2]}) \otimes \1_{Z^0} ,
\end{equation*}
where $\1_a$, respectively $\1_{Z^0}$, stands for the identity in $\mathfrak{F}_a$, respectively in $\mathfrak{F}_{Z^0}$, and, for $j=1,2$,
\begin{align*}
& a_{j,[0,\sigma_n/2]} := \chi_{[0,\sigma_n/2]}(p_j) a_j \chi_{[0,\sigma_n/2]}(p_j), \qquad a_j = \frac{i}{2} (p_j \partial_{p_j} + \partial_{p_j}p_j) .
\end{align*}
The function $\chi_{[0,\sigma_n/2]}$ satisfies $\chi_{[0,\sigma_n/2]}\in C^\infty( [ 0 , \infty ) ; [0,1])$ and
\begin{equation*}
\chi_{[0,\sigma_n/2]}(x) =
\begin{cases}
1, & \text{if } x\in [0,\sigma_n/2],\\
0, & \text{if } x\in [\sigma_n,\infty).
\end{cases}
\end{equation*}
Note that $A_{\sigma_n}$ defines a self-adjoint operator for all $n \in \mathbb{N}$.

\subsection{Regularity of the Hamiltonian w.r.t. the conjugate operator}\label{subsec:regul}

The first proposition of this section will show that the total Hamiltonian $H$ is sufficiently regular with respect to $A_{\sigma_n}$ for the Mourre theory to be applied. We recall that, given a bounded self-adjoint operator $B$ and another self-adjoint operator $A$ in a separable Hilbert space $\mathcal{H}$, $B$ is said to belong to the class $C^n(A)$ if and only if, for all $\psi \in \mathcal{H}$, the map
\begin{align}
s \mapsto e^{- i t A} B e^{i t A} \psi,
\end{align}
is of class $C^n(\R)$. Equivalently, $B \in C^n (A)$ if and only if, for all $k \in \{ 1 , \dots ,  n \}$, the iterated commutators $\mathrm{ad}_A^k(B)$ originally defined as quadratic forms on $\mathfrak{D}(A) \times \mathfrak{D}(A)$ (where $\mathrm{ad}_A^0(B) = B$ and $\mathrm{ad}_A^k(B) = [ \mathrm{ad}_A^{k-1}(B) , A ]$) extend by continuity to bounded quadratic forms on $\mathcal{H} \times \mathcal{H}$.

In the case where $B$ is not bounded, $B$ is said to belong to $C^n(A)$ if and only if the resolvent $(B-z)^{-1}$ belongs to $C^n(A)$ for all $z \in \C \setminus \mathrm{spec}( B ) $. Given $B \in C^n(A)$, one verifies using the functional calculus that $\varphi(B) \in C^n(A)$ for all $\varphi \in C_0^\infty( \R ; \R )$.

We also remind the reader that if $B \in C^1(A)$, then the set $\mathfrak{D}(B) \cap \mathfrak{D}(A)$ is a core for $B$, and the quadratic form $[B,A]$ originally defined on $( \mathfrak{D}(B) \cap \mathfrak{D}(A) ) \times ( \mathfrak{D}(B) \cap \mathfrak{D}(A) )$ extends by continuity to a bounded quadratic form on $\mathfrak{D}(B) \times \mathfrak{D}(B)$. Furthermore, the fact that $B \in C^1(A)$ implies that $(B-z)^{-1} \mathfrak{D}( A ) \subset \mathfrak{D}( A )$ for all $z \in \C \setminus \mathrm{spec}( B )$.

The next lemma will be used in the proof of Proposition \ref{lm:regularity}. We recall that $H_I^{(l)}( -i a_{j,[0,\sigma_n/2]} F^{(l)} )$ denotes the operator \eqref{eq:inter11b} or \eqref{eq:inter2244} with $F^{(l)}$ replaced by $-ia_{j,[0,\sigma_n/2]} F^{(l)}$, and likewise for $H_I^{(l)}\big ( a_{j,[0,\sigma_n/2]} a_{j',[0,\sigma_n/2]} F^{(l)} \big )$.
\begin{lem}\label{lm:acompleter2}
Suppose that Hypotheses~\ref{hypothesis-1} and \ref{hypothesis-2} hold. There exists $C > 0$ such that, for all $n \in \mathbb{N}$, $j=1,2$, $j'=1,2$ and $l=1,2$
\begin{align}
\big \| H_I^{(l)}\big ( -i a_{j,[0,\sigma_n/2]} F^{(l)} \big )^\sharp \psi \big \| \leq C \sigma_n \|(H_0+1)\psi\| , \label{eq:mkh1}
\end{align}
and
\begin{align}
\big \| H_I^{(l)}\big ( a_{j,[0,\sigma_n/2]} a_{j',[0,\sigma_n/2]} F^{(l)} \big )^\sharp \psi \big \| \leq C \sigma_n \|(H_0+1)\psi\| , \label{eq:mkh2}
\end{align}
where the notation $H_I^{(l)}\big ( - i a_{j,[0,\sigma_n/2]} F^{(l)} \big )^\sharp$ stands for $H_I^{(l)}\big ( - i a_{j,[0,\sigma_n/2]} F^{(l)} \big )$ or $H_I^{(l)} \big ( - i a_{j,[0,\sigma_n/2]} F^{(l)} \big )^*$, and likewise for $H_I^{(l)}\big ( a_{j,[0,\sigma_n/2]} a_{j',[0,\sigma_n/2]} F^{(l)} \big )^\sharp$.

Furthermore,
\begin{align}
\big \| (  \check{H}_\infty^n )^{- \frac12}  H_I^{(l)}\big ( -i a_{j,[0,\sigma_n/2]} F^{(l)} \big ) ( N_{ Z^0 } + 1 )^{-\frac12}  \big \| \le C_j \sigma_n . \label{eq:rel_Ntau_31}
\end{align}
\end{lem}
\begin{proof}
Replacing $F^{(j)}$ by $a_{j,[0,\sigma_n/2]} F^{(l)}$ in the proof of Proposition \ref{prop:relative-bound}, we see that
\begin{align}
\big \| H_I^{(l)}\big ( -i a_{j,[0,\sigma_n/2]} F^{(l)} \big )^\sharp \psi \big \| \leq C M_{\sigma_n} \|(H_0+1)\psi\| ,
\end{align}
where $M_{\sigma_n}$ denotes a positive constant depending on $\sigma_n$ and smaller that the maximum of the $L^2$-norms of $ a_{j,[0,\sigma_n/2]} F^{(l)}$, $\omega_3^{-1/2} a_{j,[0,\sigma_n/2]} F^{(l)}$, $p_2^{-1/2} a_{j,[0,\sigma_n/2]} F^{(l)}$ and $(p_2 \omega_3)^{-1/2} a_{j,[0,\sigma_n/2]} F^{(l)}$. Recalling that $F^{(l)} = h^{(l)} G^{(l)}$ with $h^{(l)}$ satisfying the estimates of Lemma \ref{lem:kernel-bound-1} and $G^{(l)}$ satisfying Hypothesis \ref{hypothesis-2}, a direct computation shows that the norms of these functions are bounded by $C \sigma_n$. This proves \eqref{eq:mkh1}. The proof of \eqref{eq:mkh2} is similar.

To establish \eqref{eq:rel_Ntau_31}, it suffices to adapt the proof of Lemma \ref{lm:rel_Ntau} in a straightforward way.
\end{proof}
We are now ready to verify that $H$ is regular with respect to $A_{\sigma_n}$, in the sense that $H \in C^2( A_{ \sigma_n } )$. We also prove that the first and second commutator of $H$ w.r.t. $A_{\sigma_n}$ extend to uniformly bounded operators from $\mathfrak{D}(H)$ to $\mathcal{H}$. These properties will be used in Section \ref{subsec:unif} below.
\begin{prop}\label{lm:regularity}
Suppose that Hypotheses~\ref{hypothesis-1} and \ref{hypothesis-2} hold. There exists $g_0 > 0$ such that, for all $0 \le g \le g_0$ and $n \in \mathbb{N}$, $H$ is of class $C^2 ( A_{ \sigma_n } )$. Moreover there exists $C > 0$ such that, for all $|g| \le g_0$ and $n \in \mathbb{N}$,
\begin{align}
& \big \|Ê[ H , i A_{ \sigma_n } ]( H_0 + i )^{-1}Ê\big \|Ê\le C , \label{eq:unifcomm1} \\
& \big \|Ê[ [ H , i A_{ \sigma_n } ] , i A_{Ê\sigma_n } ] ( H_0 + i )^{-1}Ê\big \|Ê\le C. \label{eq:unifcomm2}
\end{align}
\end{prop}
\begin{proof}
Using the method of \cite[Proposition 9]{FGS1}, it is not difficult to verify that, for all $n \in \mathbb{N}$ and $t \in \mathbb{R}$,
\begin{equation*}
e^{i t A_{\sigma_n} } \mathfrak{D}( H_0 ) \subset \mathfrak{D}( H_0 ).
\end{equation*}
By \cite[Theorem 6.3.4]{ABG} (see also \cite{GG}) and since $\mathfrak{D}( H ) = \mathfrak{D}( H_0 )$, in order to justify that $H \in C^1( A_{ \sigma_n } )$, it then suffices to prove that
\begin{equation*}
\big | \langle H \phi , A_{\sigma_n} \phi \rangle - \langle A_{\sigma_n} \phi , H \phi \rangle \big | \le C \big( \| H \phi \big \|^2 + \| \phi \|^2 \big) ,
\end{equation*}
for all $\phi \in \mathfrak{D} ( H ) \cap \mathfrak{D} ( A_{ \sigma_n } )$. A direct computation gives
\begin{align}
[ H , i A_{\sigma_n} ] =& \mathrm{d} \Gamma( \chi_{[0,\sigma_n/2]}^2( p_1 ) p_1 ) \otimes \1_{\mathfrak{F}_a} \otimes \1_{\mathfrak{F}_{Z^0}} + \1_{\mathfrak{F}_a} \otimes \mathrm{d} \Gamma( \chi_{[0,\sigma_n/2]}^2( p_2 ) p_2 ) \otimes \1_{\mathfrak{F}_{Z^0}} \notag \\
&+ g\sum_{l=1,2} \sum_{j=1,2} \big ( H_I^{(l)}\big ( - i a_{j,[0,\sigma_n/2]} F^{(l)} \big ) + H_I^{(l)} \big ( - i a_{j,[0,\sigma_n/2]} F^{(l)} \big )^* \big ). \label{eq:efg1}
\end{align}
As usual, all the commutators are to be understood in the sense of quadratic forms.

Obviously, the operator $\mathrm{d} \Gamma( \chi_{[0,\sigma_n/2]}^2( p_1 ) p_1 )$ is relatively bounded w.r.t. $\mathrm{d} \Gamma( p_1 )$ and hence also w.r.t. $H$. Moreover,
\begin{equation*}
\big \|Ê\mathrm{d} \Gamma( \chi_{[0,\sigma_n/2]}^2( p_1 ) p_1 )( H_0 + i )^{-1}Ê\big \|Ê\le C,
\end{equation*}
uniformly in $n \in \mathbb{N}$. The same holds for $\mathrm{d} \Gamma( \chi_{[0,\sigma_n/2]}^2( p_2 ) p_2 )$. According to Lemma \ref{lm:acompleter2}, we have in addition that
\begin{equation*}
\big \|ÊH_I^{(l)}\big ( - i a_{j,[0,\sigma_n/2]} F^{(l)} \big )^\sharp( H_0 + i )^{-1}Ê\big \|Ê\le C \sigma_n,
\end{equation*}
uniformly in $n \in \mathbb{N}$. Hence $H \in C^1( A_{ \sigma_n } )$ and \eqref{eq:unifcomm1} is indeed satisfied.

Similarly, in order to prove that $H \in C^2( A_{ \sigma_n } )$, it suffices to verify that the second commutator $[[H,iA_{\sigma_n}],iA_{\sigma_n}]$ extends to a relatively $H$-bounded operator. As in \eqref{eq:efg1}, we compute
\begin{align}
[ [ H , i A_{\sigma_n} ] , i A_{Ê\sigma_n }Ê] =& \mathrm{d} \Gamma \big ( [ \chi_{[0,\sigma_n/2]}^2( p_1 ) p_1 , i a_{1,[0,\sigma_n/2]} ] \big ) \otimes \1_{\mathfrak{F}_a} \otimes \1_{\mathfrak{F}_{Z^0}} \notag \\
&+ \1_{\mathfrak{F}_a} \otimes \mathrm{d} \Gamma \big ( [ \chi_{[0,\sigma_n/2]}^2( p_2 ) p_2 , i a_{2,[0,\sigma_n/2]} ] \big ) \otimes \1_{\mathfrak{F}_{Z^0}} \notag \\
&- g \sum_{l=1,2} \sum_{j=1,2 } \sum_{ j'=1,2} \big ( H_I^{(l)}\big ( a_{j,[0,\sigma_n/2]} a_{j',[0,\sigma_n/2]} F^{(l)} \big ) \notag \\
& \qquad \qquad \qquad \qquad \quad + H_I^{(l)}\big ( a_{j,[0,\sigma_n/2]} a_{j',[0,\sigma_n/2]} F^{(l)} \big )^* \big ). \label{eq:efg10}
\end{align}
We have that
\begin{equation*}
[Ê\chi_{[0,\sigma_n/2]}^2( p_1 ) p_1 , i a_{1,[0,\sigma_n/2]} ] = \chi_{[0,\sigma_n/2]}^4( p_1 ) p_1 + 2 \chi_{[0,\sigma_n/2]}^3( p_1 ) ( \partial_{p_1} \chi_{[0,\sigma_n/2]} )( p_1 ) p_1^2 ,
\end{equation*}
and therefore, since, in particular, $| p_1 ( \partial_{p_1} \chi_{[0,\sigma_n/2]} )( p_1 ) | \le C$ uniformly in $n \in \mathbb{N}$, we can conclude as before that
\begin{equation*}
\big \|Ê\mathrm{d} \Gamma \big ( [ \chi_{[0,\sigma_n/2]}^2( p_1 ) p_1 , i a_{1,[0,\sigma_n/2]} ] \big )( H_0 + i )^{-1}Ê\big \|Ê\le C.
\end{equation*}
The same holds for $\mathrm{d} \Gamma \big ( [ \chi_{[0,\sigma_n/2]}^2( p_2 ) p_2 , i a_{2,[0,\sigma_n/2]} ] \big )$ and Lemma \ref{lm:acompleter2} gives
\begin{equation*}
\big \|ÊH_I^{(l)}\big ( a_{j,[0,\sigma_n/2]} a_{j',[0,\sigma_n/2]} F^{(l)} \big )^\sharp( H_0 + i )^{-1}Ê\big \|Ê\le C \sigma_n,
\end{equation*}
uniformly in $n \in \mathbb{N}$. This concludes the proof.
\end{proof}
%
%

\subsection{The Mourre estimate and its consequences}\label{subsec:Mourre}

Our next task is to prove a Mourre estimate for $H$ in any spectral interval of size $\sigma_n$ located at a distance of order $\sigma_n$ from the ground state energy $E$. We follow the general approach of \cite{FGS1} and \cite{ABFG} with some noticeable differences due to the structure of the interaction Hamiltonian and the different infrared behavior.
\begin{theo}
\label{thm:mourre}
Suppose that Hypotheses~\ref{hypothesis-1} and \ref{hypothesis-2} hold. There exists $g_0>0$ such that, for all $0 \le g  \le g_0$ and $n \in \mathbb{N}$,
\begin{align*}
& \1_{[ \rho \sigma_{n+1} / 4 , \rho \sigma_n / 3 ]} ( H - E ) [ H , i A_{\sigma_n} ] \1_{[ \rho \sigma_{n+1} / 4 , \rho \sigma_n / 3 ]} ( H - E ) \\
&\ge \frac{ \rho \gamma \sigma_{n} }{ 6 } \1_{[ \rho \sigma_{n+1} / 4 , \rho \sigma_n / 3 ]} ( H - E ).
\end{align*}
\end{theo}

\begin{proof}
We decompose
\begin{align}
& \varphi_{\sigma_n} ( H - E ) [ H , i A_{\sigma_n} ] \varphi_{\sigma_n}( H - E ) \\
&= \varphi_{\sigma_n} ( H_n - E_n ) [ H , i A_{\sigma_n} ] \varphi_{\sigma_n}( H_n - E_n ) \label{eq:comm_1} \\
&\quad + ( \varphi_{\sigma_n} ( H - E ) - \varphi_{\sigma_n} ( H_n - E_n ) ) [ H , i A_{\sigma_n} ] \varphi_{\sigma_n}( H_n - E_n ) \label{eq:comm_2} \\
&\quad + \varphi_{\sigma_n} ( H - E ) [ H , i A_{\sigma_n} ] ( \varphi_{\sigma_n}( H - E ) - \varphi_{\sigma_n}( H_n - E_n ) ). \label{eq:comm_3}
\end{align}
Since the support of $\varphi_{ \sigma_n }$ is included into $( 0 , \rho \sigma_n )$, it should be noticed that, by Proposition \ref{prop:spectral-gap}, the following identity holds
\begin{equation}
\varphi_{\sigma_n} ( H_n - E_n ) = \1_{ \{ E_n \} }( K_n ) \otimes \varphi_{\sigma_n}( \check{H}_\infty^n ), \label{eq:cba1}
\end{equation}
the unitary equivalence $\mathcal{H} = \mathcal{H}_n \otimes \mathfrak{F}^n_\infty \otimes \mathfrak{F}^n_\infty$ being implicite, as above. Here we recall the notation $\check{H}_\infty^n = \mathrm{d} \Gamma( p_1 ) \otimes \1^n_\infty + \1^n_\infty \otimes \mathrm{d} \Gamma ( p_2 )$ on $\mathfrak{F}^n_\infty \otimes \mathfrak{F}^n_\infty$.\\

\noindent \textbf{Estimate of \eqref{eq:comm_2}}. We begin with estimating \eqref{eq:comm_2}. Using the Helffer-Sj{\"o}strand functional calculus, we represent
\begin{align}
& \varphi_{\sigma_n}( H - E ) - \varphi_{\sigma_n} ( H_n - E_n )  \notag \\
&= \sigma_n \int [ H_n - E_n - z \sigma_n ]^{-1} ( H_n - H + E - E_n ) [ H - E - z \sigma_n ]^{-1} d \tilde \varphi (z), \label{eq:almost1}
\end{align}
where
\begin{equation}
d \tilde \varphi(z) := - \frac{1}{ \pi } \frac{ \partial \tilde \varphi }{ \partial \bar z }(z) d x d y, \label{eq:almost2}
\end{equation}
and $\tilde \varphi$ denotes an almost analytic extension of $\varphi$ satisfying
\begin{equation}
\mathrm{supp}( \tilde \varphi ) \subset \{ z = x + i y , x \in \mathrm{supp}( \varphi ), | y | \le 1 \}, \label{eq:almost3}
\end{equation}
and
\begin{equation}
\left | \frac{ \partial \tilde \varphi }{ \partial \bar z }( x + i y ) \right | \le C y^4.  \label{eq:almost4}
\end{equation}
By Lemma \ref{lm:rel_Ntau2}, we have that $| E - E_n |Ê\le C g \sigma_n^2$. From the resolvent estimates
\begin{align*}
\big \| [ H_n - E_n - z \sigma_n ]^{-1} \big \| \le \sigma_n^{-1} | \mathrm{Im}(z) |^{-1} , \quad\big \|Ê[ H - E - z \sigma_n ]^{-1} \big \|Ê\le \sigma_n^{-1} | \mathrm{Im}(z) |^{-1} , 
\end{align*}
we obtain that
\begin{equation}
\big \|Ê[ H_n - E_n - z \sigma_n ]^{-1} ( E - E_n ) [ H - E - z \sigma_n ]^{-1} \big \| \le C g | \mathrm{Im}( z ) |^{-2}. \label{eq:qsd1}
\end{equation}
Now, recall the notations $H - H_n = g( H_{I,\infty}^{(1)\, n} + H_{I,\infty}^{(2)\, n} + ( H_{I,\infty}^{(1)\, n})^* + ( H_{I,\infty}^{(2)\, n})^*)$. Using Lemma \ref{lm:rel_Ntau}, we can estimate
\begin{align}
& \big \|Ê[ H_n - E_n - z \sigma_n ]^{-1}g H_{I,\infty}^{(j) \, n} [ H - E - z \sigma_n ]^{-1} \big \| \notag \\
& \le C g \sigma_n \big \|Ê[ H_n - E_n - z \sigma_n ]^{-1} (  \check{H}_\infty^n )^{ \frac12} \big \|Ê\big \|Ê( N_{ Z^0 } + 1 )^{\frac12} [ H - E - z \sigma_n ]^{-1} \big \| , \label{eq:nbv1}
\end{align}
where
\begin{equation*}
 \check{H}_\infty^n = \mathrm{d} \Gamma ( \1_{[0,\sigma_n]}(p_1) p_1 ) + \mathrm{d} \Gamma ( \1_{[0,\sigma_n]}(p_2) p_2 ).
\end{equation*}
Since
\begin{align}
 \1_n \otimes \check{H}_\infty^n \le ( K_n - E_n ) \otimes \1^n_\infty \otimes \1^n_\infty + \1_n \otimes \check{H}_\infty^n  = H_n - E_n, \label{eq:tyu1}
\end{align}
by positivity of $K_n - E_n$, we deduce that
\begin{equation}
\big \|Ê[ H_n - E_n - z \sigma_n ]^{-1} (  \check{H}_\infty^n )^{ \frac12} \big \| \le C | \mathrm{Im}(z) |^{-1} \sigma_n^{-\frac12}. \label{eq:nbv2}
\end{equation}
Moreover, since $N_{Z_0}$ is relatively bounded w.r.t. $H_0$ and hence w.r.t. $H$, it follows that
\begin{equation}
\big \|Ê( N_{ Z^0 } + 1 )^{\frac12} [ H - E - z \sigma_n ]^{-1} \big \| \le C | \mathrm{Im}(z) |^{-1} \sigma_n^{-1}. \label{eq:nbv3}
\end{equation}
Estimates \eqref{eq:nbv1}, \eqref{eq:nbv2} and \eqref{eq:nbv3} imply that
\begin{align}
& \big \|Ê[ H_n - E_n - z \sigma_n ]^{-1}g H_{I,\infty}^{(j) \, n} [ H - E - z \sigma_n ]^{-1} \big \| \le C g \sigma_n^{-\frac12} | \mathrm{Im}(z) |^{-2} . \label{eq:sdf1}
\end{align}
To estimate
\begin{align}
& \big \|Ê[ H_n - E_n - z \sigma_n ]^{-1} (H_{I,\infty}^{(j), n})^* [ H - E - z \sigma_n ]^{-1} \big \| ,
\end{align}
we can proceed similarly, using in addition the second resolvent equation, namely
\begin{align*}
&\big \| (  \check{H}_\infty^n )^{ \frac12} [ H - E - z \sigma_n ]^{-1} \big \| \\
& \le \big \| (  \check{H}_\infty^n )^{ \frac12} [ H_n - E_n - z \sigma_n ]^{-1} \big \| \\
&\quad + \big \| (  \check{H}_\infty^n )^{ \frac12} [ H_n - E_n - z \sigma_n ]^{-1} ( H_n - H + E - E_n ) [ H - E - z \sigma_n ]^{-1} \big \| \\
&\le C \sigma_n^{-\frac12} \big ( |Ê\mathrm{Im}(z) |^{-1}Ê+ |Ê\mathrm{Im}(z) |^{-2} \big ).
\end{align*}
In the last inequality, we have used \eqref{eq:nbv2} and Lemmas \ref{lm:rel_Ntau2} and \ref{lm:acompleter} which imply that
\begin{align}
\left \| ( H_n - H + E - E_n ) [ H - E - z \sigma_n ]^{-1} \right \| \le \frac{ C g }{ | \mathrm{Im} z | }. \label{eq:vvv3}
\end{align}
This shows that
\begin{align}
& \big \|Ê[ H_n - E_n - z \sigma_n ]^{-1} g( H_{I,\infty}^{(j) \, n} )^* [ H - E - z \sigma_n ]^{-1} \big \| \notag \\
& \le C g \sigma_n \big \|Ê[ H_n - E_n - z \sigma_n ]^{-1} (N_{Z^0} + 1)^{\frac12} \big \|Ê\big \|Ê(  \check{H}_\infty^n )^{ \frac12} [ H - E - z \sigma_n ]^{-1} \big \| \notag \\
& \le C g \sigma_n^{-\frac12} \big ( | \mathrm{Im}(z) |^{-2} + | \mathrm{Im}(z) |^{-3} \big ) . \label{eq:sdf2}
\end{align}
Combining \eqref{eq:sdf1} and \eqref{eq:sdf2}, we obtain that
\begin{align}
& \big \|Ê[ H_n - E_n - z \sigma_n ]^{-1} ( H_n - H ) [ H - E - z \sigma_n ]^{-1} \big \| \notag \\
& \le C g \sigma_n^{-\frac12} \big ( | \mathrm{Im}( z ) |^{-2} + | \mathrm{Im}( z ) |^{-3} \big ). \label{eq:qsd2}
\end{align}
Summarizing, Equations \eqref{eq:qsd1} and \eqref{eq:qsd2} together with the properties of the almost analytic extension $\tilde \varphi$ yield
\begin{equation}
\| \varphi_{\sigma_n}( H - E ) - \varphi_{\sigma_n} ( H_n - E_n ) \| \le C g \sigma_n^{\frac12}. \label{eq:dcb1}
\end{equation}

Now we estimate $\big \|Ê[ H , i A_{\sigma_n} ] \varphi_{\sigma_n}( H_n - E_n ) \big \|$ with $[ H , i A_{ \sigma_n } ]$ given by \eqref{eq:efg1}. We have that
\begin{align}
& \| \mathrm{d} \Gamma( \chi_{[0,\sigma_n/2]}^2( p_1 ) p_1 ) \varphi_{\sigma_n} ( H_n - E_n ) \| \notag \\
& = \big \| \1_{ \{ E_n \} }( K_n ) \otimes \big ( \mathrm{d} \Gamma( \chi_{[0,\sigma_n/2]}^2( p_1 ) p_1 ) \varphi_{\sigma_n} ( \check{H}_\infty^n ) \big ) \big \| \notag \\
& \le \big \| \mathrm{d} \Gamma( p_1 ) \varphi_{\sigma_n} ( \check{H}_\infty^n ) \big \| \le \rho \sigma_n / 2. \label{eq:arg1}
\end{align}
Likewise,
\begin{align}
& \| \mathrm{d} \Gamma( \chi_{[0,\sigma_n/2]}^2( p_2 ) p_2 ) \varphi_{\sigma_n} ( H_n - E_n ) \| \le \rho \sigma_n / 2. \label{eq:arg1bla}
\end{align}
From Lemma \ref{lm:acompleter2}, we deduce that
\begin{align}
& \left \| g H_I^{(l)} \big ( - i a_{j,[0,\sigma_n/2]} F^{(l)} \big )^\sharp \varphi_{\sigma_n} ( H_n - E_n ) \right \| \notag \\
& \le C gÊ\sigma_n \big \| ( H_0 + i ) \varphi_{\sigma_n} ( H_n - E_n ) \big \| \le C g \sigma_n. \label{eq:estim_BG_3}
\end{align}
Estimates \eqref{eq:arg1}, \eqref{eq:arg1bla} and \eqref{eq:estim_BG_3} imply that
\begin{equation*}
\big \|Ê[ H , i A_{\sigma_n} ] \varphi_{\sigma_n}( H_n - E_n ) \big \| \le C \sigma_n ,
\end{equation*}
which, together with \eqref{eq:dcb1}, yields
\begin{equation}
\| \eqref{eq:comm_2} \| \le C g \sigma_n^{\frac32}. \label{eq:edc1}
\end{equation}
$\quad$\\
\noindent \textbf{Estimate of \eqref{eq:comm_3}}.
The argument is analogous to the one we used to estimate \eqref{eq:comm_2}, except that we cannot argue in the same way as in \eqref{eq:arg1} and \eqref{eq:arg1bla} to handle the terms $\mathrm{d} \Gamma( \chi_{[0,\sigma_n/2]}^2( p_1 ) p_1 )$ and $\mathrm{d} \Gamma( \chi_{[0,\sigma_n/2]}^2( p_2 ) p_2 )$, because $\varphi_{\sigma_n}( H_n - E_n )$ is replaced by $\varphi_{\sigma_n}( H - E )$. Hence we modify our estimates as follows. We write
\begin{align}
& \mathrm{d} \Gamma( \chi_{[0,\sigma_n/2]}^2( p_1 ) p_1 ) ( \varphi_{\sigma_n}( H - E ) - \varphi_{\sigma_n} ( H_n - E_n ) ) \notag \\
& = \sigma_n \int \mathrm{d} \Gamma( \chi_{[0,\sigma_n/2]}^2( p_1 ) p_1 ) [ H_n - E_n - z \sigma_n ]^{-1} \notag \\
& \phantom{ = \sigma \int } ( H_n - H + E - E_n ) [ H - E - z \sigma_n ]^{-1} {\rm d} \tilde \varphi (z) , \label{eq:vvv1}
\end{align}
and estimate
\begin{align}
& \left \| \mathrm{d} \Gamma( \chi_{[0,\sigma_n/2]}^2( p_1 ) p_1 ) [ H_n - E_n - z \sigma_n ]^{-1} \right \| \notag \\
& \le \left \| (\1_n \otimes \mathrm{d} \Gamma( p_1 ) \otimes \1^n_\infty) [ H_n - E_n - z \sigma_n ]^{-1} \right \| \notag \\
& \le \left \| ( H_n - E_n ) [ H_n - E_n - z \sigma_n ]^{-1} \right \| \le C. \label{eq:vvv2}
\end{align}
The second inequality follows from $\1_n \otimes \mathrm{d} \Gamma( p_1 ) \otimes \1^n_\infty \le \1_n \otimes \check{H}_\infty^n$ and \eqref{eq:tyu1}. Combining \eqref{eq:vvv1}, \eqref{eq:vvv2} and \eqref{eq:vvv3}, we obtain that
\begin{align*}
\left \| \mathrm{d} \Gamma( \chi_{[0,\sigma_n/2]}^2( p_1 ) p_1 ) ( \varphi_{\sigma_n}( H - E ) - \varphi_{\sigma_n} ( H_n - E_n ) ) \right \| \le C g \sigma_n ,
\end{align*}
which, together with \eqref{eq:arg1}, yields
\begin{align*}
\left \| \mathrm{d} \Gamma( \chi_{[0,\sigma_n/2]}^2( p_1 ) p_1 ) \varphi_{\sigma_n}( H - E ) \right \| \le C \sigma_n.
\end{align*}
Proceeding analogously, we also have that
\begin{align*}
\left \| \mathrm{d} \Gamma( \chi_{[0,\sigma_n/2]}^2( p_2 ) p_2 ) \varphi_{\sigma_n}( H - E ) \right \| \le C \sigma_n ,
\end{align*}
and the terms $g H_I^{(l)}\big ( - i a_{j,[0,\sigma_n/2]} F^{(l)} \big )^\sharp$ in the commutator \eqref{eq:efg1} can be treated exactly in the same way as in \eqref{eq:estim_BG_3}, which gives
\begin{align*}
\left \| g H_I^{(l)}\big ( - i a_{j,[0,\sigma_n/2]} F^{(l)} \big )^\sharp \varphi_{\sigma_n} ( H - E ) \right \| \le C g \sigma_n.
\end{align*}
The last three estimates yield
\begin{equation}
\big \|Ê[ H , i A_{\sigma_n} ] \varphi_{\sigma_n}( H - E ) \big \| \le C \sigma_n , \label{jhg1}
\end{equation}
and hence, combining \eqref{eq:dcb1} and \eqref{jhg1}, we obtain that
\begin{equation}
\| \eqref{eq:comm_3} \| \le C g \sigma_n^{\frac32}. \label{eq:edc2}
\end{equation}
$\quad$ \\
\noindent \textbf{Lower bound for \eqref{eq:comm_1}}.
The properties of the supports $\mathrm{supp}( \varphi_{\sigma_n} )$ and $\mathrm{supp}( \chi_{ [0, \sigma_n/2 ] } )$ together with \eqref{eq:cba1} imply that
\begin{align*}
& \varphi_{\sigma_n} ( H_n - E_n ) \big ( \mathrm{d} \Gamma( \chi_{[0,\sigma_n/2]}^2( p_1 ) p_1 ) + \mathrm{d} \Gamma( \chi_{[0,\sigma_n/2]}^2( p_2 ) p_2 ) \big ) \varphi_{\sigma_n}( H_n - E_n ) \\
&= \1_{ \{ E_n \} }( K_n ) \otimes \left ( \varphi_{\sigma_n}( \check{H}_\infty^n ) \check{H}_\infty^n  \varphi_{\sigma_n}( \check{H}_\infty^n ) \right ).
\end{align*}
Consequently,
\begin{align*}
&  \varphi_{\sigma_n} ( H_n - E_n ) \big ( \mathrm{d} \Gamma( \chi_{[0,\sigma_n/2]}^2( p_1 ) p_1 ) + \mathrm{d} \Gamma( \chi_{[0,\sigma_n/2]}^2( p_2 ) p_2 ) \big ) \varphi_{\sigma_n}( H_n - E_n ) \\
& \ge \frac{ \rho \sigma_{n+1} }{ 5 } \1_{ \{ E_n \} }( K_n ) \otimes \varphi_{\sigma_n}( \check{H}_\infty^n )^2 \\
& =  \frac{ \rho \gamma \sigma_{ n } }{ 5 } \varphi_{\sigma_n} ( H_n - E_n )^2 \\
& \ge \frac{ \rho \gamma \sigma_{ n } }{ 5 } \varphi_{\sigma_n} ( H - E )^2 - C g \sigma_n^{\frac32} ,
\end{align*}
the last inequality being a consequence of \eqref{eq:dcb1}. Together with \eqref{eq:estim_BG_3}, this shows that
\begin{align}
\varphi_{\sigma_n} ( H_n - E_n ) [ H , i A_{\sigma_n} ] \varphi_{\sigma_n}( H_n - E_n ) \ge \frac{ \rho \gamma \sigma_n }{ 5 } \varphi_{\sigma_n} ( H - E )^2 - C g \sigma_n. \label{eq:edc3}
\end{align}
$\quad$ \\
\noindent \textbf{Conclusion}. Putting together \eqref{eq:edc1}, \eqref{eq:edc2} and \eqref{eq:edc3}, we finally obtain that
\begin{align}
\varphi_{\sigma_n} ( H - E ) [ H , i A_{\sigma_n} ] \varphi_{\sigma_n}( H - E ) \ge \frac{ \rho \gamma \sigma_n }{ 5 } \varphi_{\sigma_n} ( H - E )^2 - C g \sigma_n.
\end{align}
The proof is concluded by multiplying the previous inequality on the left and on the right by $\1_{[ \rho \sigma_{n+1} / 4 , \rho \sigma_n / 3 ]} ( H - E )$ and by choosing $g$ small enough.
\end{proof}

As a usual consequence of the Mourre estimate of Theorem \ref{thm:mourre}  together with the regularity of the Hamiltonian $H$ with respect to the conjugate operator $A$ (see Proposition \ref{lm:regularity}), we obtain the following limiting absorption principle and local decay property.
\begin{theo}
\label{thm:LAP}
Suppose that Hypotheses~\ref{hypothesis-1} and \ref{hypothesis-2} hold. There exists $g_0 > 0$ such that, for all $1/2 < s \le 1$, there exists $C_s > 0$ such that, for all $0 \le gÊ\le g_0$ and $n \in \mathbb{N}$,
\begin{align}
\sup_{ z \in J_n } \big \|Ê\langle A_{ \sigma_n } \rangle^{-s} ( H - z )^{-1}Ê\langle A_{ \sigma_n } \rangle^{-s} \big \| \le C_s \sigma_n^{-1} , \label{eq:LAP}
\end{align}
where $J_n = \{ z \in \mathbb{C} , \mathrm{Re}( z ) \in [ \gamma \rho \sigma_n / 4 , \rho \sigma_n / 3 ] , 0 < |Ê\mathrm{Im}( z ) |Ê\le 1 \}$. Furthermore, with $\varphi_{ \sigma_n }$ defined by \eqref{eq:defvarphi1}--\eqref{eq:defvarphi2}, we have the following: there exists $g_0 > 0$ such that, for all $0 \le s < 1$, there exists $C_s > 0$ such that, for all $0 \le gÊ\le g_0$ and $n \in \mathbb{N}$,
\begin{equation*}
\big \|Ê\langle A_{ \sigma_n } \rangle^{-s} e^{ - i t H }Ê\varphi_{ \sigma_n }( H -  E )Ê\langle A_{ \sigma_n } \rangle^{-s} \big \| \le C_s \langle \sigma_n t \rangle^{-s}.
\end{equation*}
\end{theo}
\begin{proof}
The limiting absorption principle on the spectral interval $[ \gamma \rho \sigma_n / 4 , \rho \sigma_n / 3 ]$
\begin{align*}
\sup_{ z \in J_n } \big \|Ê\langle A_{ \sigma_n } \rangle^{-s} ( H - z )^{-1}Ê\langle A_{ \sigma_n } \rangle^{-s} \big \| < \infty ,
\end{align*}
is a standard consequence of having a Mourre estimate on that spectral interval (Theorem \ref{thm:mourre}) together with the property that $H \in C^2( A_{ \sigma_n } )$ (Proposition \ref{lm:regularity}), see e.g. \cite[Theorem 7.4.1]{ABG}.

The bound of order $\mathcal{O}( \sigma_n^{-1} )$ in \eqref{eq:LAP} is a consequence of the facts that the positive constant appearing in the Mourre estimate of Theorem \ref{thm:mourre} is proportional to $\sigma_n$ and that the commutators $[ H , i A_{ \sigma_n} ]$ and $[ [ H , i A_{\sigma_n} ] , i A_{ \sigma_n } ]$ are relatively bounded w.r.t. $H$ uniformly in $\sigma_n$ (Proposition \ref{lm:regularity}). We refer to \cite[Section 2]{BF} for a detailed explanation,  that can be adapted without change to our context, leading to the bound of order $\mathcal{O}( \sigma_n^{-1} )$ in \eqref{eq:LAP}.

Similarly, it follows from \cite{HSS} that Theorem \ref{thm:mourre} together with Proposition \ref{lm:regularity} imply that
\begin{equation*}
\big \|Ê\langle A_{ \sigma_n } \rangle^{-s} e^{ - i t H }Ê\varphi_{ \sigma_n }( H -  E )Ê\langle A_{ \sigma_n } \rangle^{-s} \big \| \le C_{s , \sigma_n} \langle t \rangle^{-s} ,
\end{equation*}
for all $0 \le s < 1$, for some positive constant $C_{s , \sigma_n}$ depending on $s$ and $\sigma_n$. Again, the fact that $C_{s , \sigma_n}$ is of order $\mathcal{O}( \sigma_n^{-s} )$ can be proven exactly as in \cite[Section 2]{BF}. We do not give the details.
\end{proof}
%
%


\subsection{Uniform local decay}\label{subsec:unif}
The last step of the proof of Theorem \ref{thm:return} consists in deducing from Theorem \ref{thm:LAP} the uniform local decay estimate \eqref{eq:uniflocdec}. To this end, we follow \cite{BF} with some key modifications. More precisely, the structure of the Hamiltonian of the present paper being different from the one in \cite{BF}, most of the commutators entering the proof will be estimated with different tools, and, more importantly, we will have to accommodate the method to the more singular infrared behavior. We will emphasize the differences below, but the results that can be straightforwardly adapted to our context will be given without proof.

We begin with a few preliminary technical results that will be useful in the proof of Theorem \ref{thm:return}. The first lemma controls commutators between powers of $A_{\sigma_n}$ and functions of $H$.
\begin{lem}\label{lm:lmunif}
Suppose that Hypotheses~\ref{hypothesis-1} and \ref{hypothesis-2} hold. There exists $g_0 > 0$ such that, for all $0 \le s \le 1$, there exists $C_s >0$ such that, for all $0 \le g \le g_0$ and $n \in \mathbb{N}$,
\begin{equation*}
\big \|Ê\langle A_{Ê\sigma_n }Ê\rangle^s \varphi_{ \sigma_n }( H - E ) \langle A_{ \sigma_n } \rangle^{-s} \big \|Ê\le C_s .
\end{equation*}
\end{lem}
\begin{proof}
Since the result is obvious for $s=0$, using an interpolation argument, it suffices to prove that
\begin{equation*}
\big \| A_{Ê\sigma_n } \varphi_{ \sigma_n }( H - E ) ( A_{ \sigma_n } + i )^{-1} \big \|Ê\le C ,
\end{equation*}
for some $C > 0$. Moreover, commuting $A_{ \sigma_n }$ through $\varphi_{ \sigma_n }( H - E )$, we see that it suffices in fact to estimate the commutator
\begin{equation*}
Ê\big [Ê\varphi_{ \sigma_n }( H - E ) , i A_{ \sigma_n } \big ].
\end{equation*}
Let $\psi \in C^\infty_0(\R ; [0,1])$ be such that, for all $x\in\R$,
\begin{equation}\label{eq:defvarpsi1}
\psi (x)
:=\begin{cases}
1, &\text{if } x\in [ \rho \gamma / 5 , \rho / 2 ] , \\
0, & \text{if } x\in\R\setminus ( \rho \gamma / 6 , 2 \rho / 3 ),
\end{cases}
\end{equation}
and let, for all $n \in \mathbb{N}$, $\psi_{\sigma_n}( x ) = \psi ( \sigma_n^{-1} x )$. In particular, for all $x \in \R$, we have that $\varphi( x ) = \varphi(x ) \psi ( x )$. Therefore, we can write
\begin{align*}
Ê\big [Ê\varphi_{ \sigma_n }( H - E ) , i A_{ \sigma_n } \big ] = &Ê\big [Ê\varphi_{ \sigma_n }( H - E ) , i A_{ \sigma_n } \big ] \psi_{\sigma_n}( H - E ) \\
& + \varphi_{\sigma_n}( H - E )Ê\big [Ê\psi_{ \sigma_n }( H - E ) , i A_{ \sigma_n } \big ] .
\end{align*}
Using the Helffer-Sj{\"o}strand functional calculus with the representation given by \eqref{eq:almost1}--\eqref{eq:almost4}, we obtain that
\begin{align*}
& \big \|Ê\big [Ê\varphi_{ \sigma_n }( H - E ) , i A_{ \sigma_n } \big ] \psi_{\sigma_n}( H - E ) \big \| \\
& \le \sigma_n \int \big \| [ H - E - z \sigma_n ]^{-1} [ H , i A_{ \sigma_n } ] \psi_{\sigma_n}( H - E ) [ H - E - z \sigma_n ]^{-1} \big \| d \tilde \varphi (z) .
\end{align*}
By \eqref{jhg1} (with $\varphi_{\sigma_n}$ replaced by $\psi_{\sigma_n}$), we have that
\begin{equation*}
\big \|Ê[ H , i A_{\sigma_n} ] \psi_{\sigma_n}( H - E ) \big \| \le C \sigma_n .
\end{equation*}
Together with the resolvent estimate $\| [ H - E - z \sigma_n ]^{-1} \| \le ( | \mathrm{Im} ( z ) | \sigma_n )^{-1}$ and the properties \eqref{eq:almost2}--\eqref{eq:almost4} of the almost analytic extension $\tilde \varphi$, this shows that
\begin{align*}
\big \|Ê\big [Ê\varphi_{ \sigma_n }( H - E ) , i A_{ \sigma_n } \big ] \psi_{\sigma_n}( H - E ) \big \| \le C ,
\end{align*}
uniformly in $n \in \N$. Similarly we have that
\begin{align*}
\big \| \varphi_{\sigma_n}( H - E )Ê\big [Ê\psi_{ \sigma_n }( H - E ) , i A_{ \sigma_n } \big ] \big \|Ê\le C ,
\end{align*}
and therefore $\big \|Ê\big [Ê\varphi_{ \sigma_n }( H - E ) , i A_{ \sigma_n } \big ] \big \| \le C$. This concludes the proof.
\end{proof}
The next result is due to \cite{BF} and is a consequence of (a second quantized version of) Hardy's inequality in $\mathbb{R}^3$. Recall that $\q$ denotes the position operator for the neutrinos, and $q = | \q |$. As before we use the isomorphism between $\mathfrak{F}_a$ and $\mathfrak{F}_n \otimes \mathfrak{F}_\infty^n$ implicitly, $P_{\Omega^n_\infty}$ denotes the projection onto the vacuum in $\mathfrak{F}^\infty_n$, and $P_{\Omega^n_\infty}^\perp = \1^n_\infty - P_{\Omega^n_\infty}$.
\begin{lem}\label{lm:Hardy}
There exists $C > 0$ such that, for all $\rho > 0$, $n \in \mathbb{N}$ and $\phi \in \mathfrak{F}_a$,
\begin{equation}
\big \|Ê( \mathrm{d} \Gamma( q ) + \rho )^{-1} ( \1_n \otimes P_{\Omega^n_\infty}^\perp ) \phi \big \|Ê\le C \sigma_n \| ( \1_n \otimes P_{\Omega^n_\infty}^\perp ) \phi \|.
\end{equation}
\end{lem}

\begin{proof}
The lemma is proven in \cite[Lemma 3.4]{BF} in a symmetric Fock space. The adaptation to fermions in the anti-symmetric Fock space $\mathfrak{F}_a$ is straightforward. We do not give the details.
\end{proof}
A direct consequence of Lemma \ref{lm:Hardy} is the next proposition. It should be noticed that the r.h.s. of order $\mathcal{O}( \sigma_n^{1/2} )$ in  \eqref{eq:lkj3} below is worse than the corresponding one in \cite{BF} (which is of order $\mathcal{O}( \sigma_n )$). This is due to the worse infrared behavior of the interaction Hamiltonian $H_I$ compared to the one in \cite{BF}.
\begin{prop}\label{prop:prop1}
Suppose that Hypotheses~\ref{hypothesis-1} and \ref{hypothesis-2} hold. There exist $g_0 > 0$ and $C > 0$ such that, for all $0 \le g \le g_0$ and $n \in \mathbb{N}$,
\begin{equation}
\big \|Ê\langle \mathrm{d} \Gamma( q_1 ) + \mathrm{d} \Gamma( q_2 ) \rangle^{-1} \varphi_{ \sigma_n }( H - E ) \big \|Ê\le C \sigma_n^{ \frac12 } . \label{eq:lkj3}
\end{equation}
\end{prop}
\begin{proof}
Since $\varphi_{ \sigma_n }( H_n - E_n ) = ( \1_n \otimes P_{\Omega^n_\infty \otimes \Omega^n_\infty }^\perp ) \varphi_{ \sigma_n }( H_n - E_n )$ as follows from \eqref{eq:cba1} and the properties of the support of $\varphi_{\sigma_n}$, we obtain from Lemma \ref{lm:Hardy} that
\begin{equation}
\big \|Ê\langle \mathrm{d} \Gamma( q_1 ) + \mathrm{d} \Gamma( q_2 ) \rangle^{-1} \varphi_{ \sigma_n }( H_n - E_n ) \big \|Ê\le C \sigma_n .
\end{equation}
Using in addition \eqref{eq:dcb1}, we arrive at \eqref{eq:lkj3}.
\end{proof}
Again, the next lemma is proven in \cite{BF} in a symmetric Fock space. The proof adapts directly to our context. We do not elaborate.
\begin{lem}\label{lm:aa1}
There exists $C > 0$ such that, for all $n \in \mathbb{N}$ and all $\phi \in \mathfrak{F}_a$,
\begin{equation*}
\big \|ÊA_{ \sigma_n } \langle \mathrm{d} \Gamma( q_1 ) + \mathrm{d} \Gamma( q_2 ) \rangle^{-1} \phi \big \|Ê\le C \sigma_n \| \phi \|.
\end{equation*}
\end{lem}
As a consequence we have the following proposition. The remark we made just before Proposition \ref{prop:prop1} concerning the estimate of order $\mathcal{O}( \sigma_n^{1/2} )$ holds too for Proposition \ref{prop:prop2}.
\begin{prop}\label{prop:prop2}
Suppose that Hypotheses~\ref{hypothesis-1} and \ref{hypothesis-2} hold. There exist $g_0 > 0$ and $C > 0$ such that, for all $0 \le g \le g_0$ and $n \in \mathbb{N}$,
\begin{equation}
\big \|Ê\langle \mathrm{d} \Gamma( q_1 ) + \mathrm{d} \Gamma( q_2 ) \rangle^{-1} \varphi_{ \sigma_n }( H - E ) A_{ \sigma_n } \big \|Ê\le C \sigma_n^{ \frac12 } . \label{eq:lkj4}
\end{equation}
\end{prop}
\begin{proof}
First, we claim that
\begin{equation}
\big \|Ê\langle \mathrm{d} \Gamma( q_1 ) + \mathrm{d} \Gamma( q_2 ) \rangle^{-1} \varphi_{ \sigma_n }( H_n - E_n ) A_{ \sigma_n } \big \|Ê\le C \sigma_n . \label{eq:lkj5}
\end{equation}
Indeed, by \eqref{eq:cba1} and the fact that $A_{ \sigma_n }$ acts on $\mathfrak{F}^n_\infty$, we have that
\begin{align*}
&\langle \mathrm{d} \Gamma( q_1 ) + \mathrm{d} \Gamma( q_2 ) \rangle^{-1} \varphi_{ \sigma_n }( H_n - E_n ) A_{ \sigma_n } \\
&= \langle \mathrm{d} \Gamma( q_1 ) + \mathrm{d} \Gamma( q_2 ) \rangle^{-1} A_{ \sigma_n } \varphi_{ \sigma_n }( H_n - E_n ) \\
&\quad + \langle \mathrm{d} \Gamma( q_1 ) + \mathrm{d} \Gamma( q_2 ) \rangle^{-1} \1_{ \{ E_n \} }( K_n ) \otimes \big [ \varphi_{\sigma_n}( \check{H}_\infty^n ) , A_{\sigma_n} \big ] .
\end{align*}
The first term is estimated thanks to Lemma \ref{lm:aa1}, which gives
\begin{equation}
\big \|Ê\langle \mathrm{d} \Gamma( q_1 ) + \mathrm{d} \Gamma( q_2 ) \rangle^{-1} A_{ \sigma_n } \varphi_{ \sigma_n }( H_n - E_n ) \big \| \le C \sigma_n. \label{eq:lkj6}
\end{equation}
As for the second term, a direct computation gives
\begin{equation*}
\big [ \varphi_{\sigma_n}( \check{H}_\infty^n ) , A_{\sigma_n} \big ] = - i \big ( \mathrm{d} \Gamma( \chi_{[0,\sigma_n/2]}^2( p_1 ) p_1 ) + \mathrm{d} \Gamma( \chi_{[0,\sigma_n/2]}^2( p_2 ) p_2 ) \big ) \varphi_{\sigma_n}'( \check{H}_\infty^n ) ,
\end{equation*}
where $\varphi'$ stands for the derivative of $\varphi$. Since
\begin{equation*}
\big ( \mathrm{d} \Gamma( \chi_{[0,\sigma_n/2]}^2( p_1 ) p_1 ) + \mathrm{d} \Gamma( \chi_{[0,\sigma_n/2]}^2( p_2 ) p_2 ) \big )^2 \le ( \check{H}_\infty^n )^2,
\end{equation*}
and since $x \varphi_{\sigma_n}'(x)$ is uniformly bounded in $n \in \mathbb{N}$, we deduce that $\big [ \varphi_{\sigma_n}( \check{H}_\infty^n ) , A_{\sigma_n} \big ]$ extends to a uniformly bounded operator. In addition, we obviously have that
\begin{equation*}
\big [ \varphi_{\sigma_n}( \check{H}_\infty^n ) , A_{\sigma_n} \big ] = P_{\Omega^n_\infty \otimes \Omega^n_\infty }^\perp \big [ \varphi_{\sigma_n}( \check{H}_\infty^n ) , A_{\sigma_n} \big ],
\end{equation*}
and therefore Lemma \ref{lm:Hardy} implies that
\begin{equation}
\big \|Ê\langle \mathrm{d} \Gamma( q_1 ) + \mathrm{d} \Gamma( q_2 ) \rangle^{-1} \1_{ \{ E_n \} }( K_n ) \otimes \big [ \varphi_{\sigma_n}( \check{H}_\infty^n ) , A_{\sigma_n} \big ] \big \| \le C \sigma_n. \label{eq:lkj7}
\end{equation}
Equations \eqref{eq:lkj6} and \eqref{eq:lkj7} prove \eqref{eq:lkj5}.

In order to prove \eqref{eq:lkj4}, it remains to establish that
\begin{equation}
\big \|Ê\langle \mathrm{d} \Gamma( q_1 ) + \mathrm{d} \Gamma( q_2 ) \rangle^{-1} \big ( \varphi_{ \sigma_n }( H - E ) - \varphi_{ \sigma_n }( H_n - E_n ) \big ) A_{ \sigma_n } \big \|Ê\le C \sigma_n^{ \frac12 } . \label{eq:lkj8}
\end{equation}
Commuting $A_{ \sigma_n }$ through $\varphi_{ \sigma_n }( H - E ) - \varphi_{ \sigma_n }( H_n - E_n )$ and using Lemma \ref{lm:aa1} as above, we see that it is in fact sufficient to verify that
\begin{equation}
\big \|Ê\langle \mathrm{d} \Gamma( q_1 ) + \mathrm{d} \Gamma( q_2 ) \rangle^{-1} \big [ \varphi_{ \sigma_n }( H - E ) - \varphi_{ \sigma_n }( H_n - E_n ) , A_{ \sigma_n } \big ] \big \|Ê\le C \sigma_n^{ \frac12 } . \label{eq:lkj9}
\end{equation}
Using the representation \eqref{eq:almost1} with $\tilde \varphi$ as in \eqref{eq:almost2}--\eqref{eq:almost4}, we can write
\begin{align*}
& \big [ \varphi_{\sigma_n}( H - E ) - \varphi_{\sigma_n} ( H_n - E_n ) , A_{ \sigma_n } \big ] \notag \\
&=- \sigma_n \int [ H_n - E_n - z \sigma_n ]^{-1} [ H_n , A_{ \sigma_n } ] [ H_n - E_n - z \sigma_n ]^{-1} ( H_n - H + E - E_n ) \notag \\
&\qquad \qquad [ H - E - z \sigma_n ]^{-1} d \tilde \varphi (z) \notag \\
&+ \sigma_n \int [ H_n - E_n - z \sigma_n ]^{-1} [ H_n - H , A_{ \sigma_n } ] [ H - E - z \sigma_n ]^{-1} d \tilde \varphi (z) \notag \\
&- \sigma_n \int [ H_n - E_n - z \sigma_n ]^{-1} ( H_n - H + E - E_n ) [ H - E - z \sigma_n ]^{-1}Ê[ H , A_{ \sigma_n } ] \notag \\
&\qquad \qquad [ H - E - z \sigma_n ]^{-1} d \tilde \varphi (z). \notag
\end{align*}
We estimate the three terms entering the r.h.s. of the last inequality. By \eqref{eq:qsd1}Ê and \eqref{eq:qsd2}, we have that
\begin{align*}
& \big \|Ê[ H_n - E_n - z \sigma_n ]^{-1} ( H_n - H + E - E_n ) [ H - E - z \sigma_n ]^{-1} \big \| \notag \\
& \le C g \sigma_n^{-\frac12} \big ( | \mathrm{Im}( z ) |^{-2} + | \mathrm{Im}( z ) |^{-3} \big ).
\end{align*}
Moreover, arguing similarly as in the proof of Theorem \ref{thm:mourre}, it is not difficult to verify that
\begin{align*}
& \big \|Ê[ H_n - E_n - z \sigma_n ]^{-1} [ H_n , A_{ \sigma_n } ] \big \| \le C |Ê\mathrm{Im}( z ) |^{-1} , \\
& \big \|Ê[ H , A_{ \sigma_n } ] [ H - E - z \sigma_n ]^{-1} \big \| \le C |Ê\mathrm{Im}( z ) |^{-1}.
\end{align*}
The last $3$ estimates together with the properties of $\tilde \varphi$ prove that the first and third term in the r.h.s. of the expansion of $[ \varphi_{\sigma_n}( H - E ) - \varphi_{\sigma_n} ( H_n - E_n ) , A_{ \sigma_n } ]$ are bounded by $C g \sigma_n^{1/2}$. To estimate the second term, it suffices to use that
\begin{align}
\big \| (  \check{H}_\infty^n )^{- \frac12}  [ H_n - H , A_{Ê\sigma_n }Ê] ( N_{ Z^0 } + 1 )^{-\frac12}  \big \| \le C g \sigma_n . \label{eq:rel_Ntau_38}
\end{align}
The latter inequality is proven exactly as in \eqref{eq:rel_Ntau_31} of Lemma \ref{lm:acompleter2}. Together with \eqref{eq:nbv2} and \eqref{eq:nbv3}, this shows that second term in the r.h.s. of the expansion of $[ \varphi_{\sigma_n}( H - E ) - \varphi_{\sigma_n} ( H_n - E_n ) , A_{ \sigma_n } ]$ is also bounded by $C g \sigma_n^{1/2}$. This proves that \eqref{eq:lkj9} holds, and therefore the proof is complete.
\end{proof}
We are now ready to prove Theorem \ref{thm:return}.
\begin{proof}[Proof of Theorem \ref{thm:return}]
Let $0 \le s \le 1$ and $0< \mu < s$. Recall that
\begin{equation*}
\mathrm{Q} := \mathrm{d} \Gamma( q_1 ) + \mathrm{d} \Gamma( q_2 ) .
\end{equation*}
Given $\chi \in \mathrm{C}_0^\infty( ( - \infty , m_{Z^0} - \varepsilon ) ; \R )$, we want to prove that
\begin{align*}
& \langle \mathrm{Q} \rangle^{-s} e^{ - i t H }Ê\chi( H ) \langle \mathrm{Q} \rangle^{-s} = e^{ - i t E }Ê\chi( E ) \langle \mathrm{Q} \rangle^{-s} P_{\mathrm{gs}} \langle \mathrm{Q} \rangle^{-s} + \mathrm{R}_0(t),
\end{align*}
with $\| \mathrm{R}_0(t) \|Ê\le C \langle t \rangle^{-s+\mu}$.

The definition \eqref{eq:defvarphi1}--\eqref{eq:defvarphi2} of $\varphi_{ \sigma_n }$ shows that, for all $x \in \mathrm{supp}( \chi ( \cdot + E ) )$,
\begin{equation*}
1 \le \1_{ \{ 0 \}Ê}( x ) + \sum_{n \in \mathbb{N} } \varphi_{ \sigma_n }( x ) \le 2.
\end{equation*}
This implies that, for all $x \in \mathrm{supp}( \chi ( \cdot + E ) )$, we can write $\chi(x+E) = \chi(E) \1_{ \{0\} }(x) + \sum_{n\in \mathbb{N}} \tilde \chi ( x + E ) \varphi_{\sigma_n}( x )$ for some $\tilde \chi \in \mathrm{C}_0^\infty( ( - \infty , m_{Z^0} - \varepsilon ) ; \R )$. Therefore the spectral theorem gives
\begin{align}
& \big \| \langle \mathrm{Q} \rangle^{-s} e^{ - i t H }Ê\chi( H ) \langle \mathrm{Q} \rangle^{-s} - e^{ - i t E }Ê\chi( E ) \langle \mathrm{Q} \rangle^{-s} P_{\mathrm{gs}} \langle \mathrm{Q} \rangle^{-s} \big \| \notag \\
&\le \sum_{ n \in \mathbb{N} } \big \| \langle \mathrm{Q} \rangle^{-s} e^{ - i t H }Ê\varphi_{ \sigma_n } ( H - E ) \tilde\chi( H ) \langle \mathrm{Q} \rangle^{-s} \big \| . \label{eq:fin1}
\end{align}
Recall that $\psi_{\sigma_n}$ has been defined in \eqref{eq:defvarpsi1} with the property that $\varphi_{\sigma_n} = \varphi_{ \sigma_n }Ê\psi_{ \sigma_n }$. For all $n \in \mathbb{N}$, we can write
\begin{align*}
& \big \| \langle \mathrm{Q} \rangle^{-s} e^{ - i t H }Ê\varphi_{ \sigma_n } ( H - E ) \tilde\chi( H ) \langle \mathrm{Q} \rangle^{-s} \big \| \\
& \leÊ\big \| \langle \mathrm{Q} \rangle^{-s}Ê\psi_{ \sigma_n } ( H - E ) \langle A_{ \sigma_n } \rangle^{s-\mu} \big \| \big \|Ê\langle A_{ \sigma_n } \rangle^{-s+\mu} e^{ - i t H }Ê\varphi_{ \sigma_n }( H -  E )Ê\langle A_{ \sigma_n } \rangle^{-s+\mu} \big \| \\
& \quad \times \big \|Ê\langle A_{ \sigma_n } \rangle^{s-\mu} \tilde\chi ( H ) \langle A_{ \sigma_n } \rangle^{-s+\mu} \big \| \big \|Ê\langle A_{ \sigma_n } \rangle^{s-\mu} \psi_{ \sigma_n } ( H - E ) \langle \mathrm{Q} \rangle^{-s} \big \|.
\end{align*}
From Propositions \ref{prop:prop1} and \ref{prop:prop2} (with $\varphi_{ \sigma_n }$ replaced by $\psi_{ \sigma_n }$), we deduce that
\begin{equation*}
\big \|Ê( A_{ \sigma_n } + i ) \psi_{ \sigma_n } ( H - E ) \langle \mathrm{Q} \rangle^{-1} \big \| \le C \sigma_n^{ \frac12Ê}.
\end{equation*}
An interpolation argument then gives
\begin{equation*}
\big \|Ê\langle A_{ \sigma_n } \rangle^{s} \psi_{ \sigma_n } ( H - E ) \langle \mathrm{Q} \rangle^{-s} \big \| \le C_s \sigma_n^{ \frac{s}{2}Ê} ,
\end{equation*}
and consequently
\begin{equation*}
\big \|Ê\langle A_{ \sigma_n } \rangle^{s - \mu} \psi_{ \sigma_n } ( H - E ) \langle \mathrm{Q} \rangle^{-s} \big \| \le C_s \sigma_n^{ \frac{s}{2}Ê} .
\end{equation*}
Furthermore, from Theorem \ref{thm:LAP}, we obtain that
\begin{equation*}
\big \|Ê\langle A_{ \sigma_n } \rangle^{-s + \mu} e^{ - i t H }Ê\varphi_{ \sigma_n }( H -  E )Ê\langle A_{ \sigma_n } \rangle^{-s+\mu} \big \| \le C_{s,\mu} \langle \sigma_n t \rangle^{-s+\mu} ,
\end{equation*}
for $0 < \mu \le s$, and Lemma \ref{lm:lmunif} shows that
\begin{equation*}
\big \|Ê\langle A_{ \sigma_n } \rangle^{s - \mu} \tilde\chi ( H ) \langle A_{ \sigma_n } \rangle^{-s+\mu} \big \| \le C_{s,\mu}.
\end{equation*}
Combining the previous estimates, we arrive at
\begin{align*}
& \big \| \langle \mathrm{Q} \rangle^{-s} e^{ - i t H }Ê\varphi_{ \sigma_n } ( H - E ) \tilde\chi( H ) \big \langle \mathrm{Q} \big \rangle^{-s} \big \|  \le C_{s,\mu} \sigma_n^\mu \langle t \rangle^{ - s + \mu }.
\end{align*}
Summing over $n \in \mathbb{N}$ (which is possible since $\sigma_n = \gamma^n \sigma_0$ and $\gamma < 1$), \eqref{eq:uniflocdec} then follows from \eqref{eq:fin1}.

It remains to prove \eqref{eq:returnunif}. Let $\phi = \chi( H ) \langle Q \rangle^{-s} \psi$. We write
\begin{align*}
& \langleÊ\phi , e^{ i t H }ÊO e^{ - i t H }Ê\phi \rangle = \langleÊ\psi , \langle Q \rangle^{-s} \chi( H ) e^{ i t H } \langle Q \rangle^{-s}  \langle Q \rangle^{s} O \langle Q \rangle^{s} \langle Q \rangle^{-s} e^{ - i t H }Ê\chi( H ) \langle Q \rangle^{-s} \psi \rangle .
\end{align*}
By \eqref{eq:uniflocdec}, we deduce that
\begin{align*}
& \langleÊ\phi , e^{ i t H }ÊO e^{ - i t H }Ê\phi \rangle \\
&= \chi( E )^2 \langleÊ\psi , \langle Q \rangle^{-s} P_{\mathrm{gs}} \langle Q \rangle^{-s}  \langle Q \rangle^{s} O \langle Q \rangle^{s} \langle Q \rangle^{-s} P_{\mathrm{gs}} \langle Q \rangle^{-s} \psi \rangle + | \mathrm{R}_1(t) | \\
& = \chi( E )^2 \langleÊ\langle Q \rangle^{-s} \psi , P_{\mathrm{gs}} O P_{\mathrm{gs}} \langle Q \rangle^{-s} \psi \rangle + | \mathrm{R}_1(t) |,
\end{align*}
with $| \mathrm{R}_1(t) | \le C_{s,\mu} \langle t \rangle^{-s+\mu}$. Since $\chi( E ) P_{ \mathrm{gs} } = \chi (H) P_{ \mathrm{gs} }$ and since $\phi = \chi( H ) \langle Q \rangle^{-s} \psi$, this finally gives
\begin{align*}
& \langleÊ\phi , e^{ i t H }ÊO e^{ - i t H }Ê\phi \rangle = \langle \phi , P_{\mathrm{gs}} O P_{\mathrm{gs}} \phi \rangle + | \mathrm{R}_1(t) |,
\end{align*}
and hence \eqref{eq:returnunif} is proven.
\end{proof}
%
%

\noindent\textbf{Acknowledgements.} 
The research of J.-M. B. and J. F. is supported by ANR grant ANR-12-JS0-0008-01. J.-M. B. thanks the Mathematisches Forschungsinstitut Oberwolfach, through the programme ``Research in Pairs" 2016, where part of this work was done.

\appendix


\section{Generalized eigenfunctions of the massless Dirac operator}\label{appendixA}

In this section we recall some properties of the generalized eigenfunctions of the massless Dirac operator $D_0 = - i \boldsymbol{\alpha} \cdot \nabla$. More details can be found in
\cite[section 9.9, (44), (45), (63)]{Greiner} and in \cite [chapter 4, section 4.6]{ref8}. The expressions of the generalized eigenfunctions can also be retrieved by fixing the mass of the fermions to zero in \cite[Appendix~A]{BFG4}.

As mentioned in Section \ref{S2}, the generalized eigenfunctions of $D_0$ are labeled by the angular momentum quantum numbers
 \begin{equation*}
    j \in \mathbb{N} + \frac12 , \quad  m_j \in \{ -j, -j+1, \ldots, j-1, j\},
\end{equation*}
 and by
\begin{equation*}
    \kappa_j \in \Big \{ \pm ( j+ \frac{1}{2}) \Big \}  .
\end{equation*}
Setting $\gamma_j :=|\kappa_j|$ and $\mathrm{e}^{2i\eta_j} = \frac{-\kappa_j}{\kappa_j}$, we define
\begin{align*}
  & g_{\kappa_j,\pm}(p,r) := 2 \pi^{-\frac12}  \frac{(2 p r )^{\gamma_j}}{r}  \frac{\Gamma(\gamma_j)}{\Gamma(2\gamma_j+1)} \\
  & \times \big ( \mathrm{e}^{-i p r} \mathrm{e}^{i\eta_j} \gamma_j L(\gamma_j+1, 2\gamma_j +1, 2 i p r)
  + \mathrm{e}^{i p r} \mathrm{e}^{-i\eta_j} \gamma_j L(\gamma_j +1, 2\gamma_j +1, -2 i p r)
  \big ) ,
\end{align*}
for $p,r \in [0,\infty)$, and
\begin{align*}
 & f_{\kappa_j,\pm}(p, r) := \pm i 2 \pi^{-\frac12} \frac{(2 p r)^{\gamma_j}}{r} \frac{\Gamma(\gamma_j)}{\Gamma(2\gamma_j+1)}\\
 & \times\big ( \mathrm{e}^{-i p r}\mathrm{e}^{i\eta_j} \gamma_j L(\gamma_j +1, 2\gamma_j +1, 2 i p r) - \mathrm{e}^{i p r}\mathrm{e}^{-i\eta_j} \gamma_j L(\gamma_j +1,
 2\gamma_j +1, -2 i p r)  \big ) ,
\end{align*}
where $\Gamma$ stands for Euler's gamma function and the functions $L$ are the confluent hypergeometric functions given, for $\gamma_j>1/2$, by
\begin{align}\label{eq:cdf}
  \L(\gamma_j+1, 2\gamma_j+1, \pm 2i pr)
 :=  \frac{\Gamma(2\gamma_j +1)}{\Gamma(\gamma_j+1) \Gamma(\gamma_j)} \int_0^1 \mathrm{e}^{\pm 2 i p r u}
 u^{\gamma_j} (1-u)^{\gamma_j} d u .
\end{align}

For $\xi = ( p , \gamma ) = ( p , ( j , m_j , \kappa_j ) )$, the generalized eigenfunctions $\psi_\pm (\xi, x)$ are given by
\begin{align*}
 \psi_\pm (\xi , x) :=
 \begin{pmatrix}
  i g_{\kappa_j,\pm}(p,r) \Phi^{(1)}_{m_j,\kappa_j} (\theta,\varphi) \\
  - f_{\kappa_j,\pm}(p,r) \Phi^{(2)}_{m_j,\kappa_j} (\theta,\varphi)
 \end{pmatrix} ,
\end{align*}
where $( r , \theta , \varphi )$ are the spherical coordinates associated to $x$. The spinors $\Phi^{(1)}_{m_j,\kappa_j}$ and $\Phi^{(2)}_{m_j,\kappa_j}$ are defined in \cite[Appendix A]{BFG3}.

For positive energies $\omega(p)$, using \eqref{eq:cdf}, we get for $\kappa_j= (j+\frac12) = \gamma_j>0$,
\begin{equation}\nonumber
\begin{split}
& g_{(j+\frac12), +}(p,r) = -\frac{p}{\sqrt{\pi}} \frac{(2pr)^{\gamma_j-1}}{\Gamma(\gamma_j)}
2 \int _0^1 \sin(pr(2u-1)) u_{\gamma_j} (1-u)^{\gamma_j-1} d u , \\
& f_{(j+\frac12), +}(p,r) = -\frac{p}{\sqrt{\pi}} \frac{(2pr)^{\gamma_j-1}}{\Gamma(\gamma_j)}
2 \int _0^1 \cos(pr(2u-1)) u_{\gamma_j} (1-u)^{\gamma_j-1} d u . \\
\end{split}
\end{equation}
For the other cases, namely $\kappa_j=-(j+\frac12)= - \gamma_j$ or negative energies, we have
\begin{align*}
& g_{(j+\frac12), +}(p,r) = g_{(j+\frac12), -}(p,r)
= f_{-(j+\frac12), +}(p,r) = - f_{-(j+\frac12), -}(p,r) , \\
& f_{(j+\frac12), +}(p,r) = - g_{-(j+\frac12), +}(p,r) = - g_{-(j+\frac12), -}(p,r)
=  - f_{(j+\frac12), -}(p,r) .
\end{align*}

Straightforward computations thus yield
\begin{equation}\label{eq:appA-1}
\begin{split}
 \big | g_{(j+\frac12), +}(p,r) \big |
\leq \frac{p}{\sqrt{\pi}} \frac{(2 p r)^{\gamma_j }}{\Gamma(\gamma_j)}
\ ,\quad
 \big | f_{(j+\frac12), +}(p,r) \big |
\leq \frac{2p}{\sqrt{\pi}} \frac{(2 p r)^{\gamma_j -1}}{\Gamma(\gamma_j)} ,
\end{split}
\end{equation}
\begin{equation}\label{eq:appA-2}
\begin{split}
 &\Big | \frac{d} {d p} g_{(j+\frac12), +}(p,r) \Big |
\leq \frac{\gamma_j +1}{\sqrt{\pi}} \frac{(2 p r)^{\gamma_j }}{\Gamma(\gamma_j)} \\
 &\Big | \frac{d} {d p} f_{(j+\frac12), +}(p,r) \Big |
\leq \frac{2\gamma_j}{\sqrt{\pi}} \frac{(2 p r)^{\gamma_j -1}}{\Gamma(\gamma_j)}
+ \frac{1}{2\sqrt{\pi}}  \frac{(2 p r)^{\gamma_j + 1}}{\Gamma(\gamma_j)},
\end{split}
\end{equation}
and
\begin{equation}\label{eq:appA-3}
\begin{split}
 &\Big | \frac{d^2} {d p^2} g_{(j+\frac12), +}(p,r) \Big |
\leq  \frac{6r}{\sqrt{\pi}} \gamma_j \frac{(2 p r)^{\gamma_j-1}}{\Gamma(\gamma_j)}
    + \frac{2r}{\sqrt{\pi}} (\gamma_j-1)^2 \frac{(2 p r)^{\gamma_j}}{\Gamma(\gamma_j)}
    + \frac{r}{2\sqrt{\pi}} \frac{(2 p r)^{\gamma_j+1}}{\Gamma(\gamma_j)} ,
\\
 & \Big | \frac{d^2} {d p^2} f_{(j+\frac12), +}(p,r) \Big |
\leq \frac{3r}{\sqrt{\pi}} \gamma_j \frac{(2 p r)^{\gamma_j}}{\Gamma(\gamma_j)}
    + \frac{4r}{\sqrt{\pi}} \gamma_j(\gamma_j-1) \frac{(2 p r)^{\gamma_j-2}}{\Gamma(\gamma_j)}.
\end{split}
\end{equation}


\section{Relative bounds}\label{appendixB}
In this section, we recall a technical result obtained in \cite{BDG, BDG2, bg4} related to the $N_\tau$ estimates of \cite{GlimmJaffe}.
\begin{lem}\label{lem:bg4}
For a.e. $\xi_3\in\Sigma_3$, let
\begin{equation}\nonumber
\begin{split}
 B^{(1)}(\xi_3)  & := - \int \overline{F^{(1)}(\xi_1, \xi_2, \xi_3)} b_+(\xi_1) b_-(\xi_2) d \xi_1 d \xi_2 , \\
 B^{(2)}(\xi_3)  & :=  \int F^{(2)}(\xi_1, \xi_2, \xi_3) b_+^*(\xi_1) b_-^*(\xi_2) d \xi_1 d \xi_2 .
\end{split}
\end{equation}
For all $\psi\in\D(\hD)$, and for $j=1,2$
\begin{align*}
 & \| B^{(1)}(\xi_3) \psi \| \leq  \left(\int\frac{ |F^{(1)}(\xi_1,\xi_2,\xi_3)|^2 }{p_j} d\xi_1 d \xi_2\right)^\frac12 \| \hD^\frac12 \psi \| , \\
 & \| (B^{(2)}(\xi_3))^* \psi \|  \leq  \left(\int\frac{ |F^{(2)}(\xi_1,\xi_2,\xi_3)|^2 }{p_j} d \xi_1 d \xi_2\right)^\frac12 \| \hD^\frac12 \psi \| , \\
 & \| (B^{(1)}(\xi_3))^* \psi \| \leq \left( 2 \|F^{(1)}(\cdot, \cdot, \xi_3)\|^2 \|\psi\|^2 + \|B^{(1)}(\xi_3)\psi\|^2\right)^\frac12 , \\
&  \| B^{(2)}(\xi_3) \psi \|  \leq \left( 2 \|F^{(2)}(\cdot, \cdot, \xi_3)\|^2 \|\psi\|^2 + \|(B^{(2)}(\xi_3))^*\psi\|^2\right)^\frac12 .
\end{align*}
\end{lem}
Details of the proof can be found in \cite{BDG} and \cite{bg4}.

\bibliographystyle{amsalpha}

\end{document}